\documentclass[a4paper, reqno, 11pt]{amsart} 
\usepackage{graphicx} 

\usepackage{amssymb}

\usepackage{epsfig,bbm}  		
\usepackage{epic,eepic}         

\usepackage{bm}                 
\usepackage{amsfonts}
\usepackage{amsthm}
\usepackage{amsmath}
\usepackage{amscd}
\usepackage[latin2]{inputenc}
\usepackage{t1enc}
\usepackage[mathscr]{eucal}
\usepackage{indentfirst}
\usepackage{graphicx}
\usepackage{graphics}
\numberwithin{equation}{section}
\usepackage[margin=2.9cm]{geometry}
\usepackage{hyperref}
\usepackage{dsfont}
\usepackage{csquotes}
\usepackage{stmaryrd}
\usepackage{blindtext}
\usepackage{tikz-cd}
\usepackage{mathrsfs}
\usepackage{cite}
\usepackage{mathtools}
\usepackage{extarrows}
\usepackage{epstopdf}
\usepackage[T1]{fontenc}
\usepackage{braket}
\usepackage{yfonts}

\usepackage{tikz}
\usetikzlibrary{decorations.markings}
\usetikzlibrary{arrows.meta}
\usepackage{color}

\let\oldtocsection=\tocsection
 
\let\oldtocsubsection=\tocsubsection
 
\let\oldtocsubsubsection=\tocsubsubsection
 
\renewcommand{\tocsection}[2]{\hspace{0em}\oldtocsection{#1}{#2}}
\renewcommand{\tocsubsection}[2]{\hspace{1em}\oldtocsubsection{#1}{#2}}
\renewcommand{\tocsubsubsection}[2]{\hspace{2em}\oldtocsubsubsection{#1}{#2}}

\theoremstyle{definition}
\newtheorem{definition}[equation]{Definition}
\newtheorem{example}[equation]{Example}
\newtheorem{proposition}[equation]{Proposition}
\newtheorem{lemma}[equation]{Lemma}
\newtheorem{theorem}[equation]{Theorem}
\newtheorem{remark}[equation]{Remark}
\newtheorem{corollary}[equation]{Corollary}

\newcommand{\sfgamma}{\mathsf{\Gamma}}
\newcommand{\cbrak}[1]{\llbracket #1 \rrbracket}
\newcommand{\ip}[1]{\langle #1 \rangle}
\newcommand{\cf}[1]{\mathcal{#1}}
\newcommand{\red}[1]{\underline{#1}}
\newcommand{\frk}[1]{\mathfrak{#1}}

\newcommand{\de}{\mathrm{d}}

\newcommand{\Tr}[1]{\:{\rm Tr}\,#1}

\newcommand{\e}{\mathrm{e}}

\newcommand{\IZ}{\mathbb{Z}}

\newcommand{\IR}{\mathbb{R}}

\newcommand{\IT}{\mathbb{T}}

\newcommand{\frg}{\mathfrak{g}}

\newcommand{\frh}{\mathfrak{h}}

\newcommand{\frt}{\mathfrak{t}}

\renewcommand{\Im}{\ensuremath{\mathrm{im}}}

\newcommand{\gric}{\mathrm{GRic}}
\newcommand{\ric}{\mathrm{Ric}}
\newcommand{\ad}{\mathsf{ad}}

\def\e{{\,\rm e}\,}

\newcommand{\midodot}{\text{\Large$\odot$}}

\newcommand{\cN}{{\mathcal N}}
\newcommand{\cM}{{\mathcal M}}
\newcommand{\cS}{{\mathcal S}}
\newcommand{\cB}{{\mathcal B}}

\newcommand{\cH}{{\mathcal H}}

\newcommand{\cQ}{{\mathcal Q}}

\newcommand{\cF}{{\mathcal F}}
\newcommand{\cD}{{\mathcal D}}
\newcommand{\cT}{{\mathcal T}}

\newcommand{\cV}{{\mathcal V}}

\newcommand{\cG}{{\mathcal G}}
\newcommand{\cU}{{\mathcal U}}

\newcommand{\sfK}{{\mathsf{K}}}

\newcommand{\sfS}{{\mathsf{S}}}

\newcommand{\sfT}{{\mathsf{T}}}

\newcommand{\sfH}{{\mathsf{H}}}
\newcommand{\sfG}{{\mathsf{G}}}

\newcommand{\sfOmega}{{\mathsf{\Omega}}}

\newcommand{\unit}{\mathds{1}}   			

\usepackage{enumitem}

\newcommand{\gr}{\text{gr}}
\newcommand{\ann}{\mathrm{Ann}}
\newcommand{\rel}{\dashrightarrow}
\newcommand{\mor}{\rightarrowtail}

\newcommand{\sfgammabas}{\sfgamma_{\text{bas}}}

\newcommand{\rk}{\mathrm{rk}}
\newcommand{\Coker}{\mathrm{CoKer}}
\newcommand{\Ker}{\mathrm{Ker}}

\newcommand{\Fi}{\varphi}
\renewcommand{\div}{{\rm div}}

\newtheoremstyle{case}{}{}{}{}{}{:}{ }{}
\theoremstyle{remark}

\makeatletter
\newcommand{\oset}[3][0ex]{%
  \mathrel{\mathop{#3}\limits^{
    \vbox to#1{\kern-3\ex@
    \hbox{$\scriptstyle#2$}\vss}}}}
\makeatother

\makeatletter
\newcommand{\osett}[3][0ex]{%
  \mathrel{\mathop{#3}\limits^{
    \vbox to#1{\kern-2\ex@
    \hbox{$\scriptstyle#2$}\vss}}}}
\makeatother

\makeatletter
\newcommand{\osettt}[3][0ex]{%
  \mathrel{\mathop{#3}\limits^{
    \vbox to#1{\kern-1\ex@
    \hbox{$\scriptstyle#2$}\vss}}}}
\makeatother

\makeatletter

\DeclareRobustCommand\mathflip[1]{%
    \mathpalette\@mathflip{#1}%
}

\newcommand\@mathflip[2]{%
    \mskip4mu
    \pdfsave
    \pdfsetmatrix{-1 0 .5 1}%
    \hb@xt@\z@{\hss$\m@th#1 #2$\hss}%
    \pdfrestore
    \mskip7mu
}

\DeclareRobustCommand\mathflipnu[1]{%
    \mathpalette\@mathflipnu{#1}%
}

\newcommand\@mathflipnu[2]{%
    \mskip2mu
    \pdfsave
    \pdfsetmatrix{-1 0 .5 1}%
    \hb@xt@\z@{\hss$\m@th#1 #2$\hss}%
    \pdfrestore
    \mskip8mu
}

\makeatother

\mathtoolsset{showonlyrefs}

\title[Courant Algebroid Relations, T-Dualities and Generalised Ricci Flow]{Courant Algebroid Relations, \\[4pt] T-Dualities and Generalised Ricci Flow}

\author[T.~C.~De Fraja]{Thomas C.~De Fraja}
\address[Thomas C.~De Fraja]
{Department of Mathematics and Maxwell Institute for Mathematical
  Sciences\\ Heriot-Watt
  University\\ Edinburgh EH14 4AS\\ United Kingdom}
\email{tcd2000@hw.ac.uk}

\author[V.~E.~ Marotta]{Vincenzo Emilio Marotta} 
\address[Vincenzo Emilio Marotta]
{{Dipartimento di Matematica e Geoscienze, Universit\`{a} di Trieste,   Via A. Valerio 12/1, 34127 Trieste, Italy}}
\email{vincenzoemilio.marotta@units.it}

\author[R.~J. Szabo]{Richard J.~Szabo}
  \address[Richard J.~Szabo]
  {Department of Mathematics and Maxwell Institute for Mathematical Sciences\\
  Heriot-Watt University\\
  Edinburgh EH14 4AS \\
  United Kingdom}
  \email{R.J.Szabo@hw.ac.uk}

\date{}

\begin{document}

\begin{abstract}
The  notion of Courant algebroid relation is used to introduce a definition of relation between divergence operators on Courant algebroids. By introducing invariant divergence operators, a notion of generalised T-duality between divergences is presented through an existence and uniqueness result for related divergence operators on T-dual pairs of exact Courant algebroids, which naturally incorporates the dilaton shift. When combined with the notion of generalised isometry, this establishes circumstances under which generalised Ricci tensors are related, proving that T-duality is compatible with generalised string background equations. This enables an analysis of the compatibility between T-duality and generalised Ricci flow, showing that the T-dual of a solution of generalised Ricci flow is also a solution of generalised Ricci flow. Our constructions are illustrated through many explicit examples.  
\end{abstract}

\maketitle

{\baselineskip=12pt
\tableofcontents
}

\section{Introduction}
\label{sec:intro}

 In this paper we will demonstrate the versatility of the notion of Courant algebroid relation by providing a framework for relating geometric objects such as divergence operators and generalised Ricci tensors on Courant algebroids. This is motivated by the compatibility between generalised Ricci flow, string background equations and the different flavours of T-duality. We will explore and deepen this link with a relational approach, which adds another building block to the construction of a ``category'' of Courant algebroids. 
 
 Let us start by giving a brief overview of the interplay between the three main constituents of the present paper: Courant algebroid relations, T-dualities and generalised Ricci flow.

\medskip

\subsection{Courant Algebroids and their Relations}~\\[5pt]
 Courant algebroids~\cite{Weinstein1997} may be regarded as smooth counterparts of quadratic Lie algebras. The main examples are provided by generalised tangent bundles $\IT M = TM\oplus T^*M$ of smooth manifolds $M$, known as exact Courant algebroids, which are the central objects of study in \emph{generalised geometry}. Exact Courant algebroids were first discussed in \cite{Courant1990} as a way to give a geometric description of the Dirac bracket appearing in constrained dynamical systems. The counterparts in generalised geometry of geometric structures such as metrics, complex structures and K{\"a}hler structures were discussed in detail in~\cite{Hitchin:2003cxu,gualtieri:tesi}.

A Courant algebroid is the vector bundle counterpart of a Lie algebra endowed with an adjoint-invariant pairing. More precisely, a Courant algebroid is a vector bundle $E$ over a smooth manifold $M$ endowed with a symmetric bilinear non-degenerate pairing $\ip{\, \cdot \, , \, \cdot \,}_E$, a bracket $\cbrak{\, \cdot \, ,\, \cdot \,}_E$ on the space of sections of $E$ satisfying the Jacobi identity, and a vector bundle map $\rho_E \colon E \to TM$ called the anchor. These data are required to satisfy a compatibility condition that can be regarded as invariance of the pairing with respect to the bracket. As a byproduct of this, the bracket does not need to be skew-symmetric and its symmetric part is governed by the pairing. Courant algebroids with skew-symmetric brackets have trivial anchor map $\rho_E$. 

Courant algebroid relations come into play when extending the category of Courant algebroids along the same lines as Weinstein's symplectic ``category'', in light of the interpretation of Courant algebroids as differential-graded symplectic manifolds. A Courant algebroid relation \cite{LiBland2009, Li-Bland:2011iaz, Vysoky2020hitchiker} is defined as an involutive (maximally) isotropic subbundle $R$ of the direct product of two Courant algebroids $E_1$ and $\overline{E}_2$, where $\overline{E}_2$ denotes the Courant algebroid with reflected pairing with respect to the Courant algebroid $E_2$. Because Courant algebroids are in one-to-one correspondence with degree two graded manifolds endowed with a degree $-2$ symplectic structure and a degree three integrable function~\cite{Roytenberg2002}, Courant algebroid relations should be linked  with genuine differential-graded Lagrangian correspondences in graded geometry.

One of the main goals of generalised geometry is to define  notions of differential geometry in this extended framework. For instance, one might wish to find the counterpart of Riemannian geometry in generalised geometry. The counterpart of Riemannian metrics is given by generalised metrics on exact Courant algebroids, i.e. maximally positive-definite subbundles (with respect to the pairing). Other structures associated with a metric can then be defined. For instance, connections on Courant algebroids have been extensively studied, see e.g. \cite{Gualtieri:2007bq,Garcia-Fernandez:2020ope,Streets:2024rfo}, and have even been explored in the context of Courant algebroid relations~\cite{Vysoky2020hitchiker}. Their caveat is the non-uniqueness of a generalised Levi-Civita connection which, in order to define further geometric structures such as Ricci tensors and scalar curvature, necessitates the introduction of divergence operators~\cite{Garcia-Fernandez:2016ofz}. 

In the present paper, we develop a relational approach to divergences as the basis for defining further generalised structures. This has the added benefit of bypassing the need to define Courant algebroid connections for transverse generalised metrics, a type of degenerate generalised metric used extensively in the construction of relations in \cite{DeFraja:2023fhe}. 

\medskip

\subsection{Generalised T-Duality}~\\[5pt]
 Courant algebroids $E\to M$ with generalised metrics and divergences further provide a useful geometric tool in string theory: their defining data is equivalent to that of two-dimensional non-linear sigma-models with target space $M$, as discussed in \cite{Severa-letters, Severa2015}, where a divergence operator naturally models the dilaton field~\cite{Garcia-Fernandez:2016ofz}. In this way fully geometric descriptions of T-dualities, which relate apparently distinct backgrounds that are equivalent from the sigma-model perspective, can be achieved with far reaching consequences; particularly notable special cases are given by abelian and Poisson-Lie T-dualities. The latest developments in this area, which include string backgrounds that are susceptible to more general notions of T-duality (see e.g.~\cite{Hull2009,Marotta2021born}),  necessitate the introduction of Courant algebroid relations~\cite{DeFraja:2023fhe}. 

In this rich framework, T-duality, which admits a precise description connecting Courant algebroid structures~\cite{cavalcanti2011generalized, Severa2015}, can find its natural interpretation, as expounded in \cite{Vysoky2020hitchiker, DeFraja:2023fhe}. In particular, a Courant algebroid relation preserves the Courant algebroid structures without requiring that the base manifolds be diffeomorphic, as is often the case in T-duality, in contrast to Courant algebroid isomorphisms. In this context, the role of generalised isometries, i.e. Courant algebroid relations linking a generalised metric to another one, is prominent in the characterisation of the geometric properties of T-duality~\cite{cavalcanti2011generalized, DeFraja:2023fhe, DeFraja:2025}.

Let us discuss Poisson-Lie T-duality in more detail, which originates as a non-isometric generalisation of non-abelian T-duality in the non-linear sigma-model~\cite{Klimcik1995}.
Let $E \to M$ be a Courant algebroid endowed with a generalised metric $V^+\subset E$, and set $V^-=(V^+)^\perp$. Let $\div$ be a divergence operator on $E$, i.e. a map from $\sfgamma(E)$ to $C^\infty(M)$ satisfying an anchored Leibniz-like rule. Let $\Fi \colon M' \to M$ be a smooth map. 

Then the Courant algebroid pullback $E' \coloneqq \Fi^* E$ of the Courant algebroid $E$ is the Courant algebroid $E'\to M'$ with the compatible structure maps
\begin{align}
\ip{\Fi^*e_1, \Fi^*e_2}_{E'} = \Fi^*\ip{e_1, &\, e_2}_E \quad , \quad
 \cbrak{\Fi^*e_1, \Fi^*e_2}_{E'} = \Fi^*\cbrak{e_1, e_2}_E \ , \\[4pt]
    &\Fi_* \big( \rho_{E'} (\Fi^*e)\big) =\rho_E(e) \ ,\label{it:rhoEprime}
\end{align}    
for all $e, e_1, e_2 \in \sfgamma(E)$. As discussed in \cite{LiBland2009}, an anchor map $\rho_{E'}:E'\to TM'$ satisfying Property~\eqref{it:rhoEprime} exists if and only if
\begin{align}
[\rho_{E'}(\Fi^*e_1) , \rho_{E'}(\Fi^*e_2)]_{TM'} = \rho_{E'}( \Fi^*\cbrak{e_1, e_2}_E ) \ ,
\end{align}
for all $e_1, e_2 \in \sfgamma(E)$, and $\ker(\rho_{E'})_{m'}$ is a coisotropic subspace of $E_{\Fi(m')}$, for any $m' \in M'$. 

In order to define equivariant Poisson-Lie T-duality, also known as Poisson-Lie T-duality of dressing cosets in the gauged non-linear sigma-model~\cite{Klimcik:1996np}, assume that the Courant algebroid $E \to M$ admits an action $\psi \colon \frg \to \sfgamma(E)$ induced by the action of a Lie group $\sfG$ with Lie algebra $\frg$. Further assume that this action pulls back by $\Fi$ to a $\frg$-action $\psi' \colon \frg \to \sfgamma(E')$ turning $E'$ into a $\sfG$-equivariant Courant algebroid. 
Then $\sfG$ acts freely and properly on $M'$ via the map $\rho_{E'}\circ\psi':\frg\to\sfgamma(TM')$.

We denote by $\red E'$ the reduced Courant algebroid over the quotient manifold $\cQ' = M'/ \sfG$. Suppose that the pullback  generalised metric $V' \coloneqq\Fi^*V^+$ and divergence $\div' \coloneqq \Fi^* \div$ are reducible by the $\frg$-action $\psi' \colon \frg \to \sfgamma(E')$, with reduced generalised metric $\red V'$ and divergence $\red \div'$ on $\red E'$. Then Poisson-Lie T-duality is the statement that the physically relevant properties of $(\red V', \red \div')$, such as the non-linear sigma-model and the generalised string background equations, can be completely formulated in terms of $(V^+, \div, \psi)$, without any reference to $\cQ'$ and $\red E'$~\cite{Klimcik1995,Klimcik:1996np,Severa:2016prq,Severa:2018pag}.  

A formulation of Poisson-Lie T-duality in terms of Courant algebroid relations, in a slightly different setting to the one presented above, can be found in the work of Vysok\'{y}~\cite{Vysoky2020hitchiker}. We will come back to Vysok\'{y}'s formulation in the main body of the paper.

\medskip

\subsection{Generalised Ricci Flow and String Backgrounds}~\\[5pt]
Geometric flows offer a natural way to explore topological and geometrical properties of manifolds by evolving a given initial manifold under the flow to canonical forms in the space of geometric data. The classic example is Hamilton's Ricci flow~\cite{Hamilton1982}, which is the analogue of the heat equation for a Riemannian metric and so is expected to improve the geometry of the underlying manifold; in this case canonical metrics include Einstein metrics, and more generally Ricci solitons. Geometric flows are also useful tools for constructing solutions of supergravity equations.

In string non-linear sigma-models, T-dualities are expected to be compatible in a non-trivial way with the one-loop renormalisation group flow, which corresponds geometrically to the generalised Ricci flow~\cite{Baraglia:2013wua,Streets:2017506,Severa:2016lwc,Garcia-Fernandez:2016ofz,Severa:2018pag,Pulmann:2020omk} of the target space metric, Kalb-Ramond field and dilaton. The flow equations are precisely the (generalised) $\cN=0$ supergravity equations of motion, also known as (generalised) string background equations, whose solutions are the physically viable (generalised) string backgrounds~\cite{Callan:1985ia,Oliynyk:2005ak}.
 Equivariant Poisson-Lie T-duality, as discussed in \cite{Severa2015,Severa:2016lwc, Severa:2018pag}, represents a particularly insightful instance in which T-dual backgrounds yield T-dual solutions of the generalised Ricci flow and string background equations.\footnote{Generalised T-dualities, such as Poisson-Lie T-duality, also preserve two-loop renormalisation group flows, though in a much more subtle and non-trivial way than at one-loop~\cite{Hassler:2020wnp}; in particular, the string background equations are not sufficient for study of the flows away from fixed points.}
 
 In this paper we are particularly interested in the interplay between Courant algebroid relations and generalised Ricci flow.  
We aim to study this in the framework of generalised T-duality developed in~\cite{DeFraja:2023fhe}, with an eye to deepening the understanding of the connection with T-duality developed in \cite{Severa:2018pag}.
 
Let us briefly discuss the basics of generalised Ricci flow following \cite{Severa:2018pag, Garcia-Fernandez:2020ope, Streets:2024rfo}.
The generalised Ricci tensor $\gric_{V^+, \div}$ on a Courant algebroid $E\to M$ for the pair $(V^+, \div)$ is given in \cite{Severa:2018pag} as the $C^\infty(M)$-bilinear map
\begin{align} 
     \gric_{V^+, \div} \colon \sfgamma(V^+) \times \sfgamma(V^-) \longrightarrow C^\infty(M)
\end{align}
defined by
\begin{align}\label{eq:gric}
     \gric_{V^+, \div}(v^+, w^-) \coloneqq \div \,\cbrak{v^+, w^-}_E^+ - \pounds_{\rho_E(w^-)}\, \div \, v^+ - \Tr_{V^+}(\cbrak{\cbrak{\, \cdot \, , w^-}_E^-\,,\, v^+}_E^+) \ ,
\end{align}
for all $v^+ \in \sfgamma(V^+)$ and $w^- \in \sfgamma(V^-),$ where $\pounds$ denotes the Lie derivative.

As already mentioned, one can also define generalised Ricci tensors via Courant algebroid connections. In fact, there are many definitions for the generalised Ricci tensor available in the literature, see e.g.~\cite{Gualtieri:2007bq,Garcia-Fernandez:2016ofz,Jurco2016courant,Streets:2024rfo}, however it has been recently shown in \cite{Cavalcanti:2024uky} that all of them are equivalent. In this paper we will focus on the expression given in Equation~\eqref{eq:gric}.
In \cite{Severa:2018pag} it was shown that the generalised Ricci tensor $\gric_{V^+, \div}$ on an exact Courant algebroid $E$ with \v{S}evera class $[H]$ coincides with the Ricci tensor \smash{${\rm Ric}_{g,H}$} of the Koszul connection \smash{$\nabla^{g,H}$} on $TM$ preserving the Riemannian metric $g$ defining $V^+,$ with skew-symmetric torsion $H.$

Following \cite{Garcia-Fernandez:2020ope, Streets:2024rfo} the full Ricci tensor of the pair $(V^+, \div)$ is given by
\begin{align}
     \overline{\gric}_{V^+, \div} \coloneqq \gric_{V^+, \div} - \gric_{V^-, \div} \ .
\end{align}
It can be used to define the generalised scalar curvature as
\begin{align}
     \mathrm{GR}_{V^+, \div} \coloneqq \Tr \big( \tau \, \overline{\gric}_{V^+, \div} \big) \ ,
\end{align}
where $\tau \in {\sf Aut}(E)$ is the involutive automorphism of $E$ that equivalently defines the generalised metric $V^+$; in the following we will use $\tau$ and $V^+$ interchangeably.
Both of these are central objects in the definition of the generalised Ricci flow equations.

Let us now assume that $M$ is a compact manifold. Denote by $\cM_E$ the space of generalised metrics on $E$, and by $\cH^\times$ the space of nowhere-vanishing half-densities on $M$, i.e. the invertible sections of the line bundle of half-densities on $M$. Following \cite{Streets:2024rfo}, the \emph{Einstein-Hilbert functional} is the function $\cS \colon \cM_E \times \cH^\times \to \IR$ given by
\begin{align} \label{eq:EHaction}
     \cS(\tau, \varkappa) = \int_M\, \mathrm{GR}_{\tau, \div_{\varkappa^2}}\,\varkappa^2 \ ,
\end{align}
where $\div_{\varkappa^2}$ is the divergence operator associated with the density $\varkappa^2$.
The \emph{generalised Ricci flow} is the gradient flow of $\tfrac{1}{2}\,{\rm grad}(\cS)$, where the gradient on $\cM_E \times \cH^\times$ is defined by the metric induced by integration over the measure $\varkappa^2$ and the generalised metric  (see \cite{Streets:2024rfo} for more details).

Explicitly, the generalised Ricci flow is given by 
\begin{align}
     \frac{\partial}{\partial  t}\tau =-2 \, \gric_{\tau, \div_{\varkappa^2}} \quad , \quad \frac{\partial}{\partial  t} \varkappa = -{\rm GR}_{\tau, \div_{\varkappa^2}} \, \varkappa \ .
\end{align}
In \cite{Streets:2024rfo} a formulation of the generalised Ricci flow equations in terms of general divergence operators $\div$ is provided by replacing these equations with
\begin{align}
   \frac{\partial}{\partial  t}\tau =-2 \, \gric_{\tau, \div} \quad , \quad  \frac{\partial}{\partial  t} \div = - \cD \, {\rm GR}_{\tau, \div} \ ,
\end{align}
where $\cD \colon C^\infty(M) \to \sfgamma(E)$ is the operator defined by the de Rham differential on $M$, the anchor $\rho_E$ and the pairing of the Courant algebroid $E$.
In this paper, we will be  interested in this latter formulation of the generalised Ricci flow.

The short-time existence and uniqueness of solutions of generalised Ricci flow for compact manifolds has been proven recently in \cite{Streets:2024rfo}.

The deep interplay between generalised Ricci flow and T-duality, in its different guises, has been discussed in \cite{Streets:2017506, Severa:2016lwc, Severa:2018pag}. These works are the main inspiration for the approach to generalised Ricci flow presented in this paper.
For instance, Poisson-Lie T-duality for the generalised Ricci tensor $\gric_{V^+, \div}$ is the statement that the generalised Ricci tensor of the pullback structures $(V', \div')$ to the Courant algebroid pullback $E'$ is the pullback of $\gric_{V^+, \div},$ see~\cite[Theorem 5.6]{Severa:2018pag}. A similar statement holds for the generalised scalar curvature. It follows that if generalised Ricci flow solutions exist at time $t$ for both $\gric_{V^+, \div}$ and \smash{$\gric_{V', \div'}$}, then the physically relevant properties of the solution $(V', \div')$ at time $t$ can be expressed in terms of the solution $(V^+, \div)$ at the same time $t$.

These features allow for a more in-depth discussion of the meaning of compatibility between Poisson-Lie T-duality and generalised Ricci flow.
In \cite{Severa:2018pag} a generalised metric $V^+$ on a $\sfG$-equivariant Courant algebroid $E$ is called \emph{admissible} if $V^+$ is $\frg$-invariant and $V^+ \subset \Im(\psi)^\perp$. The space of admissible generalised metrics on $E$ is denoted by $\cN_E \subset \cM_E$. Then as in \cite[Proposition~5.18]{Severa:2018pag}, the generalised Ricci tensor $\gric_{\, \cdot \, , \div}$, viewed as a vector field on $\cM_E$, is tangent to $\cN_E$.
Denote by $\varsigma \colon \cN_E \to \cM_{\red E'}$ the map that assigns to $V^+ \subset E$ the generalised metric $\red V' \subset \red E'$ obtained after pulling $V^+$ back by $\Fi$ and reducing it by the $\frg$-action. Then the map $\varsigma$ pushes the vector field $\gric_{\, \cdot \, , \div}$ forward to the vector field $\gric_{\, \cdot \, ,\red \div'}$ on $\cM_{\red E'}$, see \cite[Theorem~5.19]{Severa:2018pag}. 

This property is key to showing the following result.  Let $V^+ \subset E$ be a generalised metric admitting pullback $V' \subset E'$ which is admissible, together with a divergence operator $\div$ whose pullback $\div'$ is reducible. If $(V^+,\div)$ satisfies the \emph{generalised string background equations}
\begin{align}
    \gric_{V^+, \div} = 0 \quad , \quad {\rm GR}_{V^+, \div} = 0 \ ,
\end{align}
or in other words it is a fixed point of the generalised Ricci flow, then the pair $(\red V', \red \div')$ satisfies the generalised string background equations as well.  

This motivates studying the compatibility of geometric T-duality, as formulated in \cite{DeFraja:2023fhe}, with the generalised Ricci flow, whereby Courant algebroid relations play a key role and can lead to surprising new insights into generalised Ricci flow. 
In particular, the translation of the above formulation of Poisson-Lie T-duality to T-duality for torus bundles should imply that the physically relevant properties of the T-dual backgrounds are determined by the geometric structures on the transitive Courant algebroids defined on the base manifold of the correspondence space. We shall elucidate this expectation in detail within the present paper.

\medskip

\subsection{Outline and Summary of Results}~\\[5pt]
 A brief outline of this paper is as follows: in Section \ref{sect:CA} the basics of Courant algebroids and their relations is discussed, along with their application to give a very general definition of T-duality, which encompasses the known notions including non-isometric and non-abelian T-dualities. In Section \ref{sect:CAreldiv} the notion of Courant algebroid relation between divergence operators is introduced, and the properties of related divergences are discussed together with an application to T-duality. Section \ref{sect:genricci} deals with the interplay between Courant algebroid relations and generalised Ricci flow. Generalised isometries are shown to play a crucial role in proving that generalised Ricci tensors are related. Moreover, it is proven that the T-dual structure of a solution of generalised Ricci flow is also a solution of generalised Ricci flow. Section \ref{sect:examples} is dedicated to the construction of explicit examples of T-dual pairs of divergences and their compatibility with generalised Ricci flow.

In more detail, \textbf{Section \ref{sect:CA}} provides the definition of Courant algebroid, together with the basic properties of exact Courant algebroids and their classification. We describe the reduction of a Courant algebroid $E$ over $M$ by an isotropic subbundle $K$ to a Courant algebroid $\red E$ over the leaf space $\cQ$ of the foliation integrating the distribution $\rho_E(K)$, emphasising the notion of basic section of a Courant algebroid. Generalised metrics on Courant algebroids are further introduced. Finally, Courant algebroid relations are defined, as are their role in defining isometries for generalised metrics. 

These notions constitute the building blocks for geometric T-duality as discussed in \cite{DeFraja:2023fhe} and summarised in this section. Given a pair of isotropic subbundles $K_1$ and $K_2$ of a Courant algebroid $E\to M$ with equal rank, topological T-duality is defined as a particular Courant algebroid relation $R:\red E{}_1\rel\red E{}_2$ between the reduced Courant algebroids $\red E{}_1\to\cQ_1$ and $\red E{}_2\to\cQ_2$ over the corresponding leaf spaces. Starting from a generalised metric $\tau_1$ on $\red E{}_1$ which is invariant with respect to the bundle of isometries $D_1\subset T\cQ_1$, geometric T-duality is defined by constructing the unique generalised metric $\tau_2$ on $\red E{}_2$ that makes $R$ a generalised isometry.
A new local formulation of T-duality is also presented, extending the previous notion, that aids in explicit calculations and examples.

In \textbf{Section \ref{sect:CAreldiv}} we give a definition of Courant algebroid relation between divergence operators.

\begin{definition}[\textbf{Definition \ref{def:relateddiv}}]
    Let ${\rm div}_1$ and ${\rm div}_2$ be divergence operators on Courant algebroids $E_1\to M_1$ and $E_2\to M_2,$ respectively. Let $R \colon E_1 \rel E_2$ be a Courant algebroid relation supported on $C\subseteq M_1\times M_2$. Then $\div_1$ and $\div_2$ are \emph{$R$-related}, denoted ${\rm div}_1 \sim_R {\rm div}_2$, if ${\rm div}_1\, e_1 \sim_C {\rm div}_2 \,e_2$ for every $(e_1,e_2) \in \sfgamma(E_1 \times \overline E_2 ;R)$.
    \end{definition}  

Many examples are provided, among others reductions of divergences and Poisson-Lie T-duality. We also consider the interplay between the notion of related Courant algebroid connections and related divergences, establishing conditions under which a relation between connections induces a relation between the associated divergences. 

The key result of this section is an existence and uniqueness result for T-dual pairs of divergence operators. In order to prove the existence of the T-dual divergence, the notions of invariant sections of a Courant algebroid (with respect to a given distribution) and divergence operators are given. 

\begin{theorem}[\textbf{Theorem \ref{thm:tdualdivergence}}] 
Suppose that the distribution $D_1 \subset T\cQ_1$ is spanned pointwise by a (local) Lie subalgebra $\frk{k}_{\tau_1} \subset \sfgamma(D_1).$ Assume that $K_1 \cap K_2 = \set{0}$ and that the space of $D_1$-invariant sections $\sfgamma_{D_1}(\red E{}_1)$ spans $\red E{}_1$ pointwise.  If $\red E{}_1$ admits a divergence $\div{}_1$ which is (locally) compatible with $D_1,$ then there exists a unique divergence operator $\div{}_2$ on $\red E{}_2$ such that $\div{}_1 \sim_R \div{}_2$.
\end{theorem}

This result is thoroughly discussed for T-duality with a correspondence space. In particular, we give a generalisation of the Cavalcanti-Gualtieri isomorphism~\cite{cavalcanti2011generalized} for T-dual pairs of principal torus bundles to leaf spaces $\cQ_1$ and $\cQ_2$ which are arbitrary fibre bundles over a common base~$\cB$. Recall that T-dual pairs in the sense of \cite{bouwknegt2004tduality, cavalcanti2011generalized}
are given by principal $\sfT^k$-bundles $(\cQ_1, H_1)$ and $(\cQ_2, H_2)$ endowed with closed three-forms $H_1 \in \sfOmega_{\rm cl}^3(\cQ_1)$ and $H_2 \in \sfOmega^3_{\rm cl}(\cQ_2)$, respectively, satisfying a compatibility condition. It is shown in \cite[Theorem 3.1]{cavalcanti2011generalized} that such pairs induce an isomorphism of transitive Courant algebroids $\mathscr{F} \colon \IT \cQ_1 / \sfT^k \to \IT \cQ_2/ \sfT^k$ which is compatible with the twisted differentials on $\cQ_1$ and $\cQ_2$.
In the general setting of \cite{DeFraja:2023fhe}, we obtain

\begin{proposition}[\textbf{Proposition \ref{prop:transitiveTduality}}]
There exist transitive Courant algebroids \smash{$\red {\red E}{}_1$} and \smash{$\red {\red E}{}_2$} over $\cB$ together with a $C^\infty(\cB)$-module isomorphism \smash{$\mathscr{F} \colon \sfgamma(\red {\red E}{}_1) \to \sfgamma(\red {\red E}{}_2)$} induced by the T-duality relation~$R$.
\end{proposition}

Using this result, it is shown that the definition of T-duality given in \cite{Garcia-Fernandez:2020ope}, for the case of principal torus bundles with $\sfT^k$-invariant generalised metrics and divergences endowed on their generalised tangent bundles, is recovered in our relational approach. The dilaton shift is moreover naturally explained in terms of related divergence operators. We shall discuss the compatibility with the twisted differentials in an upcoming work \cite{DeFraja:2025}, which in type~II string theory is used to describe the transformation of Ramond--Ramond fields under T-duality.

\textbf{Section \ref{sect:genricci}} deals with the compatibility between geometric T-duality and generalised Ricci tensors, and ultimately generalised Ricci flow, using the machinery developed in Section \ref{sect:CAreldiv}. We show that a pair of related Courant algebroids endowed with related generalised metrics and divergence operators give rise to related generalised Ricci tensors.

\begin{proposition}[\textbf{Proposition \ref{cor:relatedric}}] \label{prop:reldricintro}
If $R \colon (E_1, V_1^+) \rel (E_2, V_2^+)$ is a generalised isometry supported on a submanifold $C\subseteq M_1\times M_2$ which is a Dirac structure in $E_1 \times \overline{E}_2$, and $\div_1 \sim_R \div_2$, then \smash{$\mathrm{GRic}_{V_1^+, \div_1} \sim_R \mathrm{GRic}_{V_2^+, \div_2}$} and \smash{$\mathrm{GR}_{V_1^+, \div_1} \sim_C \mathrm{GR}_{V_2^+, \div_2}$}.
\end{proposition}

This result implies, in particular, that our general notion of T-duality is compatible with the generalised string background equations. 

\begin{theorem}[\textbf{Theorem~\ref{cor:Tdualbackground}}]
    If a (locally) $D_1$-invariant pair $(V_1^+,\div_1)$ on $\red E{}_1$ satisfies the generalised string background equations, then so does the T-dual pair $(V_2^+,\div_2)$ on $\red E{}_2$.
\end{theorem}

The classic   examples are provided by Poisson-Lie T-duality and T-duality for torus bundles.
We further derive Buscher-like rules for the Ricci tensors in T-duality, which are useful for explicit calculational purposes. Explicitly, the Ricci tensor of the T-dual metric is determined by the Christoffel symbols and Ricci tensor of the original metric, as well as representatives of the \v{S}evera classes of the T-dual pair of exact Courant algebroids.

Proposition \ref{prop:reldricintro} represents the cornerstone of the main result of this paper: In our setting of geometric T-duality, given a unique solution of the generalised Ricci flow with isometries that are complete vector fields, there exists a unique T-dual solution of the generalised Ricci flow.

\begin{theorem}[\textbf{Theorem \ref{thm:compatibilityRicci}}] \label{thm:compRicciintro}
If there exists $T \in \IR_{>0}$ such that the family of pairs $(\tau_1(t), \div{}_1(t))$ on $\red E{}_1$ is a unique solution of the generalised Ricci flow for $t \in [0,T)$ with initial condition $(\tau_1(0), \div{}_1(0)) = (\tau_1, \div{}_1)$, and the Lie algebra of isometries $\frk{k}_{\tau_1}$ consists of complete vector fields, then there is a unique T-dual family of pairs $(\tau_2(t), \div{}_2(t))$ on $\red E{}_2$ that is a solution of the generalised Ricci flow. 
\end{theorem}

The local version of this result is stated as well and the case of correspondence spaces is discussed, showing that Theorem \ref{thm:compRicciintro} recovers the results of \cite{Streets:2017506} for torus bundles.

Finally, in \textbf{Section \ref{sect:examples}} many explicit examples are provided which go beyond the standard examples of T-duality for principal torus bundles and Poisson-Lie T-duality that are discussed throughout the rest of the paper. These include geometric T-duality for Hamilton's cigar soliton (known  in string theory as Witten's black hole) which produces a new generalised Ricci soliton, a new example involving a three-dimensional hyperbolic space, and two examples of geometric T-duality between non-principal circle bundles involving Klein bottles. 

\textbf{Appendix \ref{app:twistedaction}} at the end of the paper derives pertinent properties of our new definition of $\sigma$-adjoint action on an exact Courant algebroid for a given isotropic splitting $\sigma$. This was introduced in~\cite{DeFraja:2023fhe} and it plays a prominent role throughout the main text in the invariance requirements for our general notion of geometric T-duality.

\medskip

\subsection{Acknowledgements}~\\[5pt]
  We thank Dan Ag{\"u}ero, Alejandro Cabrera, Rui Loja Fernandes, Roberto Rubio, Pavol \v{S}evera and Fridrich Valach for helpful conversations and correspondence. The work of T.C.D. was supported by an EPSRC Doctoral Training Partnership Award. The work of V.E.M. was supported by PNRR MUR projects PE0000023-NQSTI and in part by the GACR Grant EXPRO 19-28268X. 
 
\medskip

\section{Courant Algebroids: Reduction, Relations and T-Duality}\label{sect:CA}

In this section we provide relevant background on Courant algebroids, reduction of Courant algebroids, and Courant algebroid relations. When combined with a notion of isometries of generalised metrics on Courant algebroids, we review the definition and main features of our general notion of T-duality introduced in~\cite{DeFraja:2023fhe}, and generalise it to a new local version which is useful in explicit examples and calculations.

\medskip

\subsection{Courant Algebroids}~\\[5pt]
 We provide a brief review of Courant algebroids and refer to \cite{Jurco2016courant, gualtieri:tesi} for more details.

\begin{definition}\label{def:CourantAlg}
A \emph{Courant algebroid} is a quadruple $(E,\llbracket\,\cdot\,,\,\cdot\,\rrbracket,\langle\,\cdot\,,\,\cdot\,\rangle ,\rho)$, where $E$ is a vector bundle over a manifold $M$
with a fibrewise non-degenerate bilinear pairing $\langle\,\cdot\,,\,\cdot\,\rangle \in\mathsf{\Gamma}(\midodot^2E^*),$ a vector bundle morphism 
$\rho \colon E \rightarrow TM $
called the \emph{anchor}, and an $\IR$-bilinear bracket operation   
\begin{equation} 
\llbracket\,\cdot\,,\,\cdot\,\rrbracket \colon \mathsf{\Gamma}(E) \times \mathsf{\Gamma}(E) \longrightarrow \mathsf{\Gamma}(E)
\end{equation}
called the \emph{Dorfman bracket}, which together satisfy
\begin{enumerate}[label=(\roman{enumi})]
    \item $\pounds_{\rho(e)}\langle e_1, e_2\rangle  = \langle\llbracket e,e_1\rrbracket, e_2\rangle  + \langle e_1, \llbracket e, e_2\rrbracket\rangle  \ ,$  \label{eqn:metric1} \\[-3mm]
    \item $\langle\llbracket e, e\rrbracket, e_1\rangle  = \tfrac{1}{2}\, \pounds_{\rho(e_1)}\langle e, e\rangle  \ ,$  \label{eqn:metric2}\\[-3mm]
    \item $\llbracket e, \llbracket e_1,e_2 \rrbracket  \rrbracket  = \llbracket \llbracket e,e_1\rrbracket  ,e_2 \rrbracket  +  \llbracket e_1, \llbracket e,e_2 \rrbracket  \rrbracket  \ , $ \label{eqn:Jacobiid}
\end{enumerate}
for all $e,e_1, e_2 \in \mathsf{\Gamma}(E)$, where $\pounds$ denotes the Lie derivative along a vector field on $M$. 
\end{definition}

\begin{remark}
The anchored Leibniz rule for the Dorfman bracket
\begin{align} \label{eqn:anchorLeibniz}
\llbracket e_1 ,f\,e_2 \rrbracket  = f\, \llbracket e_1, e_2 \rrbracket  + \big( \pounds_{\rho(e_1)} f\big)\,e_2 \ ,
\end{align}
for all $e_1,e_2\in\mathsf{\Gamma}(E)$ and $f\in C^\infty(M),$ follows from invariance~\ref{eqn:metric1} of the pairing.
The Jacobi identity~\ref{eqn:Jacobiid} and the anchored Leibniz rule \eqref{eqn:anchorLeibniz} imply that the anchor $\rho$ is a bracket homomorphism:
\begin{align}\label{eq:rhohomo}
\rho(\llbracket e_1, e_2\rrbracket)= [\rho(e_1), \rho(e_2)] \ ,
\end{align}
for all $e_1 , e_2 \in \mathsf{\Gamma}(E)$, where $[\,\cdot\,,\,\cdot\,]$ denotes the Lie bracket of vector fields. 
\end{remark}

Throughout this paper we will denote the Courant algebroid $(E,\llbracket\,\cdot\,,\,\cdot\,\rrbracket,\langle\,\cdot\,,\,\cdot\,\rangle ,\rho)$ simply by $E$ when there is no ambiguity in the structure maps. If clarity is required, we will sometimes use a subscript $_E$ to label the structures on a Courant algebroid $E$. If there are multiple Courant algebroids $E_i$ involved in the discussion, we label the structures on $E_i$ with a subscript $_i\,$.

\begin{example}[\textbf{Standard Courant Algebroids}] \label{ex:diffeostandard}
The (\emph{$H$-twisted}) \emph{standard Courant algebroid} over $M$ is the quadruple $(\IT M , \llbracket\,\cdot\,,\,\cdot\,\rrbracket_{H} ,\langle\,\cdot\,,\,\cdot\,\rangle_{\IT M}, {\rm pr}_1  )$, where 
\begin{align}
\IT M := TM \oplus T^*M
\end{align}
is the generalised tangent bundle, the Dorfman bracket is
\begin{align}\label{eqn:Hbracket}
    \llbracket X + \alpha, Y + \beta \rrbracket_H \coloneqq[X,Y] + \pounds_X \beta - \iota_Y\, \de \alpha + \iota_Y\, \iota_X H \ ,
\end{align}
for a closed three-form $H \in \mathsf{\Omega}^3_{\rm cl} (M)$, the pairing is given by
\begin{align}
    \ip{X+\alpha, Y+\beta} = \iota_Y \alpha + \iota_X\beta \ ,
\end{align} 
and the anchor ${\rm pr}_1$ is the projection to the first summand of $\IT M = TM \oplus T^*M$; here $\iota_X$ denotes the contraction of differential forms with the vector field $X\in\sfgamma(TM)$. From now on, we will denote the $H$-twisted standard Courant algebroid over $M$ by $(\IT M, H).$
\end{example}

\begin{remark}\label{rem:Courantcond}
For any Courant algebroid $E$ over $M$ there is an `adjoint' map $\rho^*\colon T^*M\to E$ given by
\begin{align}\label{eqn:rhostar}
\langle\rho^*(\alpha),e\rangle  \coloneqq \iota_e\, \rho^{\rm t}(\alpha) \ ,
\end{align}
for all $\alpha\in \mathsf{\Omega}^1(M)$ and $e\in\mathsf{\Gamma}(E)$, where $\rho^{\rm t}\colon T^*M\to E^*$ is the transpose of $\rho.$ 
The map $\rho^*$ induces a map $\cD\colon C^\infty(M)\to\mathsf{\Gamma}(E)$ defined by $$\cD f=\rho^*(\de f)\ ,$$ for all $f\in C^\infty(M)$, which obeys a derivation-like rule and is the natural generalisation of the exterior derivative in the Courant algebroid $E$. Then item~\ref{eqn:metric2} of Definition~\ref{def:CourantAlg} is equivalent to
\begin{align}\label{eq:bracketcD}
    \llbracket e,e\rrbracket = \tfrac12\,\cD\langle e,e\rangle \ ,
\end{align}
which together with Equation~\eqref{eqn:anchorLeibniz} imply the additional Leibniz rule
\begin{align} \label{eqn:Leibnizfirst}
    \llbracket f\,e_1,e_2\rrbracket = f\,\llbracket e_1,e_2\rrbracket -\big(\pounds_{\rho(e_2)} f\big)\,e_1 + \langle e_1,e_2\rangle \, \cD f \ .
\end{align}
\end{remark}

Applying Equation~\eqref{eq:rhohomo} to Equation~\eqref{eq:bracketcD} it follows that
\begin{align}
 \rho \circ \rho^* = 0 \ ,   
\end{align}
and this motivates 

\begin{definition}
A Courant algebroid $E$ over $M$ is \emph{transitive} if the anchor map $\rho$ is surjective, and \emph{exact} if the chain complex
\begin{align}\label{eqn:exCou}
0 \longrightarrow T^*M \xlongrightarrow{\rho^*} E \xlongrightarrow{\rho} TM \longrightarrow 0
\end{align}
is a short exact sequence of vector bundles.
\end{definition}  

An isotropic splitting\footnote{We will often drop the adjective `isotropic' and simply refer to $\sigma:TM\to E$ as a splitting for brevity.} $\sigma\colon TM \to E$ of the short exact sequence \eqref{eqn:exCou} defines an isomorphism $E\cong(\IT M,H_\sigma)$ with closed three-form $H_\sigma \in \mathsf{\Omega}_{\rm cl}^3(M)$ given by
\begin{align}\label{eqn:severaclass}
    H_\sigma(X,Y,Z) =  \ip{\llbracket\sigma(X),\sigma(Y)\rrbracket, \sigma(Z)}
\end{align}
for $X,Y,Z \in \mathsf{\Gamma}(TM)$, which measures the lack of integrability of the subbundle $\Im(\sigma)\subset E$ with respect to the Dorfman bracket. 

For the standard Courant algebroid $\IT M = TM \oplus T^*M$, any automorphism $\boldsymbol\Phi\in\mathsf{Aut}(\IT M)$  can be written in terms of a diffeomorphism $\Fi \in {\sf Diff}(M)$ and a $B$-field transformation such that 
\begin{align}
\boldsymbol\Phi = \big(\Fi_* + (\Fi^{-1})^*\big) \circ \e^B \ \colon \ \IT M \longrightarrow \IT M \ . 
\end{align}
Here $B$ is a two-form on $M$ regarded as a map $B^\flat \colon TM \to T^*M$, and $\e^{B}\, \colon \IT M \rightarrow \IT M$ is the vector bundle morphism given by
\begin{align}
\e^{B}\, (X + \alpha) = X + B^\flat(X) + \alpha \ ,
\end{align}
for all $X \in \sfgamma(TM)$ and $\alpha \in \sfgamma(T^*M)$. 

\begin{remark}[\textbf{\v{S}evera's Classification}\cite{Severa-letters}] 
Since the difference between two splittings $\sigma$ and $\sigma'$ defines a two-form $B\in \mathsf{\Omega}^2(M)$ via $(\sigma-\sigma')(X) = \rho^*(\iota_X B)$, the closed three-form $H_\sigma$ defined by Equation \eqref{eqn:severaclass} is shifted to $H_\sigma+\de B$ under a change of splitting. Hence there is a well-defined cohomology class $[H_\sigma]\in {\sf H}^3(M,\IR)$ associated to the exact sequence \eqref{eqn:exCou} which completely determines the Courant algebroid structure. This is called the \emph{\v{S}evera class} of the exact Courant algebroid; a representative $H\in\sfOmega^3_{\rm cl}(M)$ of the \v{S}evera class will sometimes be called the \emph{Kalb-Ramond flux} of the exact Courant algebroid. A general discussion can be found in~\cite[Section~2.2]{Garcia-Fernandez:2020ope}.
\end{remark}

Let $L$ be a subbundle of $E$ supported on a submanifold $N \subset M$. We denote by $\sfgamma(E; L)$ the space of sections of $E$ whose restriction to the submanifold $N$ becomes a section of $L$.
 
\begin{definition}
 Let $(E,\llbracket\,\cdot\,,\,\cdot\,\rrbracket,\langle\,\cdot\,,\,\cdot\,\rangle,\rho)$   be a Courant algebroid over $M$. An \emph{almost Dirac structure}  is a maximally isotropic subbundle $L \subset E.$ An almost Dirac structure is  a \emph{Dirac structure} if it is also involutive with respect to the Dorfman bracket. For a  smooth submanifold  $N \subset M$, a \emph{Dirac structure $L$ supported on $N$} is a maximally isotropic involutive subbundle $L \subset E \rvert_N,$ i.e.~$\cbrak{\sfgamma(E; L), \sfgamma(E; L)} \subseteq \sfgamma(E; L)$, such that $\rho(L) \subseteq TN.$ 
\end{definition}

\medskip

\subsection{Reduction of Courant Algebroids}~\\[5pt]
 We briefly explain the reduction of Courant algebroids by isotropic subbundles, following \cite{Zambon2008reduction}. We begin with the notion of basic sections for an isotropic subbundle $K$ of a Courant algebroid $E$. Let $K^\perp$ denote the annihilator of $K$ with respect to the pairing on $E$.

\begin{definition}
    The space of sections of $K^\perp$ which are \emph{basic with respect to} $K$ is given by
    \begin{align}
        \mathsf{\Gamma}_{\text{bas}}(K^\perp) \coloneqq \set{e \in \mathsf{\Gamma}(K^\perp) \ \vert \ \llbracket \mathsf{\Gamma}(K) ,e \rrbracket \subset \mathsf{\Gamma}(K)} \ .
    \end{align}
\end{definition}

When $\sfgamma_{\text{bas}}(K^\perp)$ spans $K^\perp$ pointwise, we say that $K^\perp$ has enough basic sections. In this instance $K$ is also an involutive subbundle, hence $\rho(K)\subset TM$ is an integrable distribution.

The reduction of Courant algebroids by a foliation is the content of the Bursztyn-Cavalcanti-Gualtieri-Zambon Theorem, see \cite{Bursztyn2007reduction, Zambon2008reduction, Bursztyn:2023onv}.

\begin{theorem}\label{thm:reduction}
    Let $E$ be a Courant algebroid over $M$, and let $K$ be an isotropic subbundle of $E$.
    Suppose that $\rho(K^\perp) = TM$, and that $K^\perp$ has enough basic sections. Assume that the quotient of $M$ by the foliation $\cf F$ integrating the regular distribution $\rho(K)$ is smooth, inducing the unique surjective submersion $\varpi \colon M \to \cf Q$. Then there is a Courant algebroid $\red E$ over $\cf Q$ fitting into the pullback diagram
    \[
\begin{tikzcd}
K^\perp / K \arrow{rr}{} \arrow[swap]{dd}{} & & \red E \arrow[]{dd}{} \\ & & \\
M \arrow[]{rr}{\varpi} & & \cf Q 
\end{tikzcd}
\]
of vector bundles. The reduced Courant algebroid $\red E$ is exact if $E$ is exact. 
\end{theorem}

We write $$\natural:K^\perp\longrightarrow \red E$$ for the vector bundle morphism covering $\varpi:M\to\cQ$ given by the quotients by $K$ and $\cF$.
We denote the anchor map inherited from $\rho:E\to TM$ by $\red\rho:\red E\to T\cQ$ through the commutative diagram
 \[
\begin{tikzcd}
K^\perp \arrow{rr}{\rho} \arrow[swap]{dd}{\natural\,} & &  TM \arrow[]{dd}{\varpi_*} \\ & & \\
\red E \arrow[]{rr}{\red\rho} & & T\cf Q 
\end{tikzcd}
\]

Reducibility of an exact Courant algebroid also plays a role in the choice of a splitting that is compatible with the reduction procedure.

\begin{definition}\label{defn:adaptedsplittings}
Let $E$ be an exact Courant algebroid over $M$ with an isotropic subbundle $K$. A splitting $\sigma \colon  TM \to E$ is \emph{adapted to $K$} if
\begin{enumerate}[label= (\alph{enumi}), labelwidth=0pt]
\item the image of $\sigma$ is isotropic,\\[-3mm]
\item $\sigma(TM)\subset K^\perp \ ,$\\[-3mm]
\item $\sigma(X)\in \mathsf{\Gamma}_{\mathrm{bas}}(K^\perp)$ for any vector field $X$ on $M$ which is projectable to $\cQ$.
\end{enumerate} 
\end{definition}

If $\sigma$ is an adapted splitting, then
\begin{align}
    \sigma(T\cF) = \sigma\big(\rho(K)\big)\subseteq K \ ,
\end{align}
by~\cite[Remark~5.2]{Zambon2008reduction}.
The condition of having enough basic sections for the reduction to be possible is closely related to the existence of adapted splittings: by \cite[Proposition 5.5]{Zambon2008reduction}, splittings adapted to $K$ exist if and only if $\mathsf{\Gamma}_{\rm{bas}}(K^\perp)$ spans $K^\perp$ pointwise. In that case a splitting $\sigma$ induces an isotropic splitting $\red\sigma:T\cQ\to\red E$ of the reduced Courant algebroid $\red E$ over $\cQ$ by~\cite[Proposition~5.7]{Zambon2008reduction}.

\medskip

\subsection{Courant Algebroid Relations}\label{ssec:CArelations}~\\[5pt]
 We give a brief introduction to Courant algebroid relations, which were first introduced in \cite{LiBland2009, Li-Bland:2011iaz}. For more details, see \cite{Vysoky2020hitchiker}. 

For two Courant algebroids $E_1$ and $E_2$ over manifolds $M_1$ and $M_2$, respectively, the product $(E_1\times E_2,\llbracket\,\cdot\,,\,\cdot\,\rrbracket ,\langle\,\cdot\,,\,\cdot\,\rangle,\rho)$ is a Courant algebroid over $M_1\times M_2$, with the Courant algebroid structures defined by
\begin{align}
    \rho(e_1, e_2) & \coloneqq (\rho_{1}(e_1), \rho_{2}(e_2)) \ , \\[4pt] \ip{(e_1, e_2), (e_1', e_2')} & \coloneqq \ip{e_1, e_1'}_{1} \circ {\rm pr}_1 + \ip{e_2, e_2'}_{2} \circ {\rm pr}_2 \ , \label{eqn:pairingprod} \\[4pt]
    \llbracket(e_1,e_2),(e_1',e_2')\rrbracket &\coloneqq (\llbracket e_1,e_1'\rrbracket_{1},\llbracket e_2, e_2'\rrbracket_{2}) \ ,
\end{align}
where ${\rm pr}_i\colon M_1\times M_2 \to M_i$ are the projection maps for $i=1,2$; for the sake of brevity, we will usually not write the projections explicitly in the following.
For a Courant algebroid $E$, denote by $\overline{E}$ the Courant algebroid $(E,\llbracket\,\cdot\,,\,\cdot\,\rrbracket_E ,-\langle\,\cdot\,,\,\cdot\,\rangle_E,\rho_E)$.

\begin{definition} \label{def:CArel}
Let $E_1$ and $E_2$ be Courant algebroids over $M_1$ and $M_2$ respectively.

A \emph{Courant algebroid relation from $E_1$ to $E_2$}, denoted $R\colon  E_1 \dashrightarrow E_2,$ is an isotropic involutive subbundle $R \subseteq E_1 \times \overline{E}_2$ supported on a submanifold $C\subseteq M_1 \times M_2$ such that $\rho(R^\perp) \subset TC.$

If $C$ is the graph $\gr(\varphi)$ of a smooth map $\varphi \colon M_1 \to M_2$, then $R$ is a \emph{Courant algebroid morphism over $\varphi$}, denoted $R\colon E_1 \rightarrowtail E_2$.

If $R$ is further the graph $\gr(\boldsymbol\Phi)$ of a vector bundle map $\boldsymbol\Phi \colon E_1 \to E_2$ covering $\varphi\colon M_1 \to M_2$, then $\boldsymbol\Phi$ is a \emph{classical Courant algebroid morphism over $\varphi$}.

By viewing a Courant algebroid relation $R$ as a subset of $E_2 \times \overline{E}_1$, we obtain the \emph{transpose relation} $R^\top\colon E_2 \dashrightarrow E_1$, whose support is denoted $C^\top\subseteq M_2\times M_1$, and similarly the \emph{transpose morphism} $R^\top\colon E_2 \rightarrowtail E_1$ when $\varphi$ is a diffeomorphism. 
\end{definition}

In this paper we will mostly consider Courant algebroid relations that are maximally isotropic, i.e.~Dirac structures, as in \cite{LiBland2009, Li-Bland:2011iaz}. Motivations for the more general Definition \ref{def:CArel} can be found in \cite{Vysoky2020hitchiker}.

\begin{example}
    Let $C \subset M_1 \times M_2$ be a submanifold and define the subbundle of $T^* (M_1 \times M_2)$ supported on $C$ given by
    \begin{align}
        \overline{\ann}(TC) = \set{(\alpha_1, -\alpha_2) \in T^*M_1 \times T^*M_2 \ | \ (\alpha_1, \alpha_2) \in \ann(TC)}  \ .
    \end{align}
     Then 
    \begin{align}
        R\coloneqq TC \oplus \overline{\ann}(TC) \ \colon \ \IT M_1 \rel \IT M_2 
    \end{align}
    is a Courant algebroid relation between the standard Courant algebroids $\IT M_1$ and $\IT M_2$ which is also a Dirac structure. If the standard Courant algebroids are twisted by Kalb-Ramond fluxes $H_1 \in \sfOmega_{\rm cl}^3(M_1)$ and $H_2 \in \sfOmega_{\rm cl}^3(M_2),$ respectively, then $R$ is a Courant algebroid relation when $(H_1, - H_2) \rvert_C = (0,0)$.
\end{example}

Defining relations among geometric objects associated to related Courant algebroids starts at the primitive level of the submanifold $C\subseteq M_1\times M_2$.

\begin{definition}
    Two functions $f_1\in C^\infty(M_1)$ and $f_2 \in C^\infty(M_2)$ are \emph{$C$-related}, denoted $f_1 \sim_C f_2$, if $f_1(c_1) = f_2(c_2)$ for all $(c_1,c_2) \in C$, that is,
 \begin{align}
 {\rm pr}_1^*\, f_1 \big\rvert_C = {\rm pr}_2^*\, f_2 \big\rvert_C \ .   
\end{align}
\end{definition}

For brevity, we will sometimes loosely write $f_1=f_2$ for $C$-related functions when no confusion can arise. The natural extension of this notion to sections of the associated vector bundles leads to

\begin{definition}
 Let $R \colon E_1 \dashrightarrow E_2$ be a Courant algebroid relation supported on $C\subseteq M_1\times M_2$. Two sections $e_1 \in \mathsf{\Gamma}(E_1)$ and $e_2 \in \mathsf{\Gamma}(E_2)$ are $R$-\emph{related}, denoted $e_1 \sim_R e_2,$ if $(e_1, e_2) \in \mathsf{\Gamma}(E_1 \times \overline{E}_2 ; R).$ 
\end{definition}

A useful result we will make use of is

\begin{lemma} \label{lemma:relationlieder}
    If $e_1 \sim_R e_2$ and $f_1 \sim_C f_2$, then $\pounds_{\rho_1(e_1)}f_1 \sim_C \pounds_{\rho_2(e_2)}f_2$.
\end{lemma}
\begin{proof}
    Let $X_i = \rho_i(e_i)$ for $i=1,2$. Since $R$ is compatible with the anchor, $(X_1,X_2)_c \in T_cC$ for every $c = (c_1,c_2) \in C$. Let $(\Phi^t_1, \Phi^t_2)$ be the flow of $(X_1,X_2)$ and $\epsilon>0$ such that $(\Phi^t_1, \Phi^t_2)(c) \in C$ for $t\in(-\epsilon, \epsilon)$. Then $f_1(\Phi_1^t(c_1)) = f_2(\Phi_2^t(c_2))$ since $f_1\sim_C f_2$, and the result follows by differentiating this equality with respect to $t$ at $t=0$.
\end{proof}

\begin{example}[\textbf{Reduction of Courant Algebroids}]\label{eg:reductionrel}
Consider the reduction of Courant algebroids in Theorem \ref{thm:reduction}. This gives the Courant algebroid morphism $Q(K) \colon E \rightarrowtail {\red E}$ over $\varpi$,
where $Q(K)$ is defined pointwise by
    \begin{align}
    Q(K)_{(m,\varpi(m))}=\set{(e,\natural(e))\, | \,  e\in K^\perp_m} \ \subset \ E_m\times \red{\overline{E}}{}_{\varpi(m)} \ ,
\end{align}
for any $m \in M.$ Then $Q(K)$ is a Dirac structure in $E \times \overline{\red E}$ supported on $\gr(\varpi)\subset M \times \cQ$. Any basic function on $M$ has a $\gr(\varpi)$-related function on $\cQ$. Similarly, any basic section of $K^\perp$ is $Q(K)$-related to a section of $\red E$.
\end{example}

\begin{definition} \label{def:kernelsR}
    Let $R \colon E_1 \dashrightarrow E_2$ be a Courant algebroid relation supported on $C\subseteq M_1\times M_2$. The \emph{kernel} $\Ker(R) \subseteq {\rm pr}_1^*(E_1)$ of $R$ is given by
    \begin{align}
        \Ker(R)_{{\rm pr}_1(c)} = \set{ e_1 \in (E_1)_{{\rm pr}_1(c)} \ | \ (e_1, 0 ) \in R_c} \ ,
    \end{align}
    for all $c \in C.$  The \emph{cokernel} $\Coker(R) \subseteq {\rm pr}_2^*(E_2)$ of $R$ is the kernel of the transpose relation $R^\top \colon E_2 \dashrightarrow E_1$:
    \begin{align}
       \Coker(R)_{{\rm pr}_2(c)} = \Ker(R^\top)_{{\rm pr}_2(c)} = \set{ e_2 \in (E_2)_{{\rm pr}_2(c)} \ | \ (0, e_2) \in R_c} \ , 
    \end{align}
    for all $c \in C.$
\end{definition}

\begin{example}
In the setting of Example \ref{eg:reductionrel}, the kernel and  cokernel of $Q(K)$ are given by $$\Ker\big(Q(K)\big)_{m}= K_{m} \quad , \quad \Coker\big(Q(K)\big)_{\varpi(m)} = \set{0^{\red E}_{\varpi(m)}} \ , $$ for all $(m,\varpi(m))\in\gr(\varpi)$, where $0^{\red E}:\cQ\to\red E$ is the zero section. This mimics a quotienting behaviour by $K$ for the morphism $Q(K)$.
\end{example}

\begin{definition}
Let $R\colon E_1 \rel E_2$ and $R' \colon  E_2 \rel E_3$ be Courant algebroid relations supported on $C$ and $C'$, respectively. The \emph{composition} $R' \circ R$ is the subset of $E_1 \times \overline{E}_3$ given by
\begin{align}\label{eqn:reldefinition}
R'\circ R = \set{(e_1,e_3) \in E_1 \times \overline{E}_3 \ | \ (e_1, e_2) \in R \ , \ (e_2,e_3) \in R' \ \text{ for some } e_2\in E_2} \ .
\end{align}
\end{definition}

According to \cite{Vysoky2020hitchiker}, the composition $R' \circ R$ is a Courant algebroid relation if $R$ and $R'$ \emph{compose cleanly}, i.e. $R'\circ R$ is a submanifold of $E_1 \times \overline{E}_3$ such that the induced map\footnote{Here $p \colon E_1 \times \overline{E}_2 \times E_2 \times \overline{E}_3 \to E_1 \times \overline{E}_3$ is the projection to the first and the fourth factor of the Cartesian product, and similarly for $\pi \colon M_1 \times M_2 \times M_2 \times M_3 \to M_1 \times M_3$ below.} $p\colon R' \diamond R \to R' \circ R$
is a smooth surjective submersion, where 
    $R' \diamond R \coloneqq (R \times R') \cap \big(E_1 \times \Delta(E_2) \times \overline{E}_3\big)$ and $\Delta(E_2)$ is the diagonal embedding of $E_2$ into $\overline E_2\times E_2$. Equivalently, $C'\circ C$ is a submanifold of $M_1 \times M_3$ such that the induced map $\pi\colon C' \diamond C \to C' \circ C$ is a smooth surjective submersion and the rank of the linear map $p_c\colon (R' \diamond R)_c \to (R' \circ R)_{\pi(c)}$ is independent of $c \in C' \diamond C$; here $C'\circ C$ and $C' \diamond C$ are defined analogously to $R'\circ R$ and $R' \diamond R$, respectively. In other words, these are the necessary conditions to ensure that $R' \circ R$ is a vector subbundle of $E_1 \times \overline{E}_3$ supported on the submanifold $C' \circ C \subset M_1 \times M_3$.

\medskip

\subsection{Generalised Metrics and Isometries}~\\[5pt]
 We  introduce generalised metrics for Courant algebroids and an extended notion of isometry between generalised metrics stated in terms of Courant algebroid relations, see \cite{gualtieri:tesi, Jurco2016courant, Vysoky2020hitchiker} for more details.
 
\begin{definition}\label{defn:generalisedmetric}
A \emph{generalised metric} on a Courant algebroid $E $ is an automorphism $\tau\colon E \to E$ covering the identity with $\tau^2 = \unit$ which, together with the pairing $\ip{\,\cdot\,,\,\cdot\,}$, defines a positive-definite fibrewise metric on $E$ given by
\begin{align}\label{eqn:generalisedmetrictau}
    \cG (e, e') = \ip{e, \tau(e')} \ ,
\end{align}
for $e, e'\in E$. Equivalently, a generalised metric on $E$ is a maximally positive-definite (with respect to $\ip{\,\cdot \, , \, \cdot\,}$) subbundle $V^+ \subset E;$ the subbundles $V^+$ and $V^-:=(V^+)^\perp$ give an orthogonal decomposition $E = V^+\oplus V^-$ into the eigenbundles
\begin{align}
    V^\pm = \ker(\tau\mp\unit) \ .
\end{align}
We denote the Courant algebroid $E$ endowed with a generalised metric $\tau$ or $V^+$ by $(E, \tau)$ or $(E,V^+)$. 
\end{definition}

\begin{example}[\textbf{Exact Courant Algebroids}]
     A generalised metric $\tau$ on a (twisted) standard Courant algebroid $E=\IT M$ corresponds to a pair $(g, b)$, where $g$ is a Riemannian metric on $M$ and $b \in \sfOmega^2(M)$, see \cite[Proposition~2.40]{Garcia-Fernandez:2020ope}. In this case $$V^\pm = \gr(\pm\, g+b) \ ,$$ that is, $X+\alpha\in V^\pm$ if and only if $\alpha=\iota_X(\pm\,g+b)$. The corresponding fibrewise metric on $\IT M = TM\oplus T^*M$ is given by
     \begin{align}\label{eqn:generalisedmetric}
    \cG =
    \begin{pmatrix}
    g - b\,g^{-1}\,b & b\, g^{-1} \\
    -g^{-1}\, b & g^{-1}
    \end{pmatrix} \ .
\end{align}

Changing the splitting of an exact Courant algebroid $E\cong(\IT M,H)$ sends $b$ to $b- B$ and $H$ to $H+\de B$,
so there is a preferred identification with $b=0$. It follows that a generalised metric $\tau$ on an exact Courant algebroid $E\to M$ is equivalent to a pair $(g,\sigma)$, where $g$ is a Riemannian metric on $M$ and $\sigma:TM\to E$ is an isotropic splitting.
     We will sometimes denote this Courant algebroid and generalised metric pair by $\big(E,(g,b)\big)$, when the context is clear. If $b=0$ we denote this simply by~$(E,g)$.
\end{example}

\begin{definition}\label{defn:generalisedisometry2}
Let $R \colon E_1 \rel E_2$ be a Courant algebroid relation supported on a submanifold $C\subseteq M_1 \times M_2$. Let $V_1^+$ and $V_2^+$ be generalised metrics on $E_1$ and $E_2$ respectively, and set $R^\pm=R\cap(V_1^\pm \times V_2^\pm)\rvert_C$. Then $R$ is a \emph{generalised isometry between $V_1^+$ and $V_2^+$} if
\begin{align}\label{eqn:Rdecomp}
    R_c = R^+_c \oplus R^-_c \ ,
\end{align}
for each $c \in C$.
\end{definition}

We will work under the assumption that $R^\pm=R\cap(V_1^\pm \times V_2^\pm)\rvert_C$ are subbundles, hence the splitting \eqref{eqn:Rdecomp} holds globally.

This definition can be equivalently formulated by considering the corresponding automorphisms $\tau_1 \in {\sf Aut}(E_1)$ and $\tau_2 \in {\sf Aut}(E_2)$, and setting $\tau \coloneqq \tau_1 \times \tau_2 \in {\sf Aut}(E_1 \times E_2).$ Then a Courant algebroid relation $R \colon E_1 \rel E_2$ supported on $C \subseteq M_1 \times M_2$ is a generalised isometry if $\tau(R)= R,$ see~\cite{Vysoky2020hitchiker, DeFraja:2023fhe}. When $R=\gr(\boldsymbol\Phi)$ for a classical Courant algebroid isomorphism $\boldsymbol\Phi:E_1\to E_2$, this recovers the usual notion of isometry.

It is shown in \cite[Proposition 4.17]{DeFraja:2023fhe} that a generalised isometry $R$ has trivial kernel $\Ker(R)$ and cokernel $\Coker(R)$.

\medskip

\subsection{Geometric T-Duality}~\\[5pt] \label{subsect:geomTdual}
 Let $E$ be an exact Courant algebroid over $M$ endowed with isotropic subbundles $K_1$ and $K_2$. Here and throughout the rest of the paper we will assume that 
\begin{align} \label{eqn:rankK}
\rk (K_1) = \rk(K_2) \ .    
\end{align}
Assume as well that $K_1^\perp$ and $K_2^\perp$ have enough basic sections, and denote by $\cF_i$ the foliation of $M$ given by the integral manifolds of the distribution $\rho(K_i).$ Suppose that the leaf spaces $\cQ_i=M / \cF_i$ have a smooth structure, hence there are unique surjective submersions $\varpi_i \colon M \to \cQ_i$ for $i=1,2,$ which we summarise in the bisubmersion diagram
\begin{equation}
    \begin{tikzcd}
  &\arrow[swap]{dl}{\varpi_1} M \arrow{dr}{\varpi_2}&   \\
\cQ_1 & & \cQ_2 
 \end{tikzcd}
\end{equation}

We consider the setting in which the diagram
\begin{equation}
\begin{tikzcd}\label{cd:T-dualitycd}
  &\arrow[tail,swap]{dl}{Q(K_1)} E \arrow[tail]{dr}{Q(K_2)}&   \\
\red E{}_1 \arrow[dashed]{rr}{R}& & \red E{}_2 
 \end{tikzcd}
 \end{equation}
defines a  Courant algebroid relation 
\begin{align}
 R =  Q(K_2) \circ Q(K_1)^\top \ \colon \ \red E{}_1 \dashrightarrow \red E{}_2   
\end{align} 
supported on the submanifold 
\begin{align}
    {\red C}=\set{(\varpi_1(m),\varpi_2(m)) \, | \, m \in M} \ \subset \ \cQ_1 \times \cQ_2 \ .
\end{align}

It was shown in \cite[Theorem 5.8]{DeFraja:2023fhe} that if $\red C$ is a smooth manifold and $T\cF_1 \cap T\cF_2$ has constant rank, then $R$ is a Courant algebroid relation if and only if $K_1 \cap K_2$ has constant rank. Then $R$ is a Dirac structure  in $\red E_1 \times \overline{\red E}_2$ supported on $\red C\,$, which is given explicitly by
\begin{align}\label{eqn:T-dualRelation}
    R = \set{(\natural_{1}(e), \natural_{2}(e))\,|\, e \in K_1^\perp \cap K_2^\perp} \ ,
\end{align}
where $\natural_i \colon  K_i^\perp \to \bigr( K_i^\perp /K_i \bigl) / \cF_i = \red E{}_i$ denotes the quotient map  for $i=1,2$.
Note that there is no general result guaranteeing the smoothness of $\red C\,$; it was shown in \cite[Lemma~5.7]{DeFraja:2023fhe}  that $\red C$ is a smooth manifold if $\cQ_1$ and $\cQ_2$ are fibred over the same manifold $\cB,$ in which case $\red C$ is the fibred product $\red C = \cQ_1 \times_\cB \cQ_2.$

In this instance $R$ is said to be a \emph{T-duality relation} between the reduced Courant algebroids $\red E{}_1$ and $\red E{}_2$~\cite{DeFraja:2023fhe}. It describes a topological T-duality between the manifolds $\cQ_1$ and $\cQ_2$ (generically endowed with Kalb-Ramond fluxes). Geometric T-duality is defined in \cite{DeFraja:2023fhe} by introducing generalised metrics $ V_1^+$ and $ V_2^+$, or equivalently $\tau_1$ and $\tau_2$, into this picture.

\begin{definition} \label{def:geomTdual}
Two Courant algebroids endowed with generalised metrics $(\red E{}_1,  V_1^+)$ and $(\red E{}_2,  V_2^+)$ fitting into Diagram \eqref{cd:T-dualitycd} are \emph{geometrically T-dual} if $R$ is a generalised isometry.
\end{definition}

The problem of transporting a generalised metric $ V_1^+ \subset \red E{}_1$ to $\red E{}_2$ such that $R$ becomes a generalised isometry is addressed in \cite{DeFraja:2023fhe}. In order to discuss this here, let us introduce some preliminary notions. 

We introduce a new action 
\begin{align} 
\ad^\sigma:\sfgamma(E)\times\sfgamma(E)\longrightarrow\sfgamma(E)
\end{align}
on an exact Courant algebroid $E$ for a given splitting $\sigma \colon TM\to E$, called the \emph{$\sigma$-adjoint action}. It is defined by
\begin{align} \label{eqn:sigmatwistact}
    {\sf ad}^\sigma_e\, e' \coloneqq \llbracket e , e' \rrbracket - \rho^*\,{\sigma}^*\llbracket e , {\sigma}(\rho(e')) \rrbracket \ ,
\end{align}
for all $e, \, e' \in \sfgamma(E).$
If $X\in \sfgamma(TM)$, we denote $\ad^\sigma_X \coloneqq \ad^\sigma_{\sigma(X)}$.
This was introduced in \cite[Definition~5.21]{DeFraja:2023fhe}; the role of $\ad^\sigma_X$ for $X\in\sfgamma(TM)$ is to essentially remove the contribution from the representative $H$ of the \v{S}evera class of $E$ to the action of the Courant algebroid on itself, see~\cite[Remark~5.23]{DeFraja:2023fhe} for details. See Appendix~\ref{app:twistedaction} for the properties of the $\sigma$-adjoint action ${\sf ad}^\sigma$.
 
The subbundle $$D_1 \coloneqq T {\red\cF}{}_1 = \varpi_{1*}(T \cF_2)$$ of $T\cQ_1$, induced by the subbundle $K_2 \subset E$, is called the \emph{distribution of T-duality directions}.
It plays the role of the bundle of isometries over $\cQ_1$.

\begin{definition} \label{def:D1invariance}
    Let $\red \sigma_1 \colon T\cQ_1 \to \red E{}_1 $ be the splitting of the exact Courant algebroid $\red E{}_1$ over $\cQ_1$ induced by an isotropic  splitting $\sigma_1 \colon TM \to E$ of $E$ which is adapted to $K_1\subset E$.   
    A generalised metric $V_1^+$ on $\red E{}_1$ is \emph{invariant with respect to $D_1$}, or simply \emph{$D_1$-invariant}, if there are generators $\set{\red X{}_1,\dots, \red X{}_{r_1}}\subset \sfgamma(T\cQ_1)$ of the $C^\infty(\cQ_1)$-submodule $\sfgamma(D_1)$ such that
    \begin{align}\label{eqn:genmetricinvariance}
        \ad^{\red \sigma_1}_{\red X{}_k} (\red w{}_1 ) \ \in \ \mathsf{\Gamma}(V_1^+) \ ,
    \end{align}
    for $k=1,\dots,r_1$ and $\red w{}_1\in \mathsf{\Gamma}(V_1^+)$. The Lie subalgebra $\frk{k}_{\tau_1} \coloneqq \text{Span}_{\IR}\set{\red X{}_1,\dots, \red X{}_{r_1}}$ of $\sfgamma(T\cQ_1)$ is the \emph{isometry algebra of $D_1$} or the \emph{Lie algebra of Killing vector fields for $V_1^+$}.
\end{definition}

\begin{remark}[\textbf{Invariant Differential Forms}]
    We denote by \smash{$\sfOmega_{D_1}^\bullet(\cQ_1)$} the space of $D_1$-invariant differential forms on $\cQ_1$, i.e. forms $\red\omega \in \sfOmega^\bullet(\cQ_1)$ such that \smash{$\pounds_{\red X} \, \red \omega = 0$} for all $\red X \in \frk{k}_{\tau_1}.$ In particular, we denote by $C^\infty_{D_1}(\cQ_1)$ the space of $D_1$-invariant functions; they are identified with the basic functions $C_{\rm bas}^\infty(\cQ_1,\red\cF{}_1)$ on $\cQ_1$ with respect to the foliation $\red \cF{}_1$.
\end{remark}

\begin{definition} \label{def:compatiblesplit}
    Let $V_1^+$ be a $D_1$-invariant generalised metric on $\red E{}_1$ with isometry algebra $\frk{k}_{\tau_1}$. Adapted splittings $\sigma_i:TM\to E$ with respect to $K_i$ for $i=1,2$ are \emph{compatible} if they satisfy
    \begin{align} \label{eqn:compatiblesplits}
        \llbracket \sigma_2  (X) , e \rrbracket = \ad^{\sigma_1}_X\, e \ ,
    \end{align}
    for every $X \in  \frk{k}_2$ and $e\in \mathsf{\Gamma}(E)$, where
    \begin{align}
       \frk{k}_2 \coloneqq \varpi_{1*}^{-1}(\frk{k}_{\tau_1}) \cap \mathsf{\Gamma}\big(\rho(K_2)\big) \ .
    \end{align}
\end{definition}

Recall that the difference between any two splittings is given by $(\sigma_1 - \sigma_2)(X) = \rho^*(\iota_X B)$, for some $B \in \sfOmega^2(M)$. By~\cite[Remark~5.25]{DeFraja:2023fhe}, Equation \eqref{eqn:compatiblesplits} then yields the invariance condition $\pounds_X B = 0$ for all $X \in \frk{k}_2$. 

We can now state the main result proven in \cite{DeFraja:2023fhe}. It provides global conditions under which a generalised metric $V_1^+$ on $\red E{}_1$ can be transported to a unique generalised metric on $\red E{}_2$ such that the relation $R$ is a generalised isometry.

\begin{theorem}\label{thm:maingeneral}
    Let $E$ be an exact Courant algebroid over $M$ endowed with isotropic subbundles $K_1$ and $K_2$ of the same rank such that $K_1^\perp$ and $K_2^\perp$ have enough basic sections, giving T-duality related Courant algebroids $\red E{}_1$ and $\red E{}_2$ with T-duality relation $R$. Assume further that the bisubmersion involutivity condition
    \begin{align} \label{eqn:bisubmersioncd}
\big[\mathsf{\Gamma}\bigl(\ker(\varpi_{1*})\bigr) , \mathsf{\Gamma}\big(\ker(\varpi_{2*})\big)\big] \ \subseteq  \ \mathsf{\Gamma}\bigl(\ker(\varpi_{1*})\bigr) + \mathsf{\Gamma}\bigl(\ker(\varpi_{2*})\bigr)
\end{align}
    holds, and that compatible adapted splittings $\sigma_1$ and $\sigma_2$ are given. Finally, assume that $\red E{}_1$ is endowed with a $D_1$-invariant generalised metric $V_1^+$. 
    Then the following are equivalent:
    \begin{enumerate}[label = (\roman{enumi})]
        \item\label{item:main1} $K_2^\perp \cap K_1 \subseteq K_2 \ . $ \\[-3mm]
        \item\label{item:main2} $K_2 \cap K_1^\perp \subseteq K_1 \ . $ \\[-3mm]
        \item\label{item:main3} There exists a unique generalised metric $V_2^+$ on $\red E{}_2$ such that $R$ is a generalised isometry between $V_1^+$ and $V_2^+$, i.e. $(\red E{}_1, V_1^+)$ and $(\red E{}_2, V_2^+)$ are geometrically T-dual.
    \end{enumerate}
\end{theorem}

The condition \eqref{eqn:bisubmersioncd} implies that there are unique foliations $\red \cF{}_1$ of $\cQ_1$, whose tangent distribution is $D_1$, and $\red \cF{}_2$ of $\cQ_2$ which are Hausdorff-Morita equivalent, see \cite{Garmendia2018}.

\begin{remark} \label{rmk:tdualsplitiso}
 Consider the setting of Theorem~\ref{thm:maingeneral} with the isomorphism 
 \begin{align}
 \Psi_1 = \sigma_1 \oplus \tfrac{1}{2}\,\rho^* \colon TM \oplus T^*M \longrightarrow E \ .
 \end{align}
 Recall that the difference between two isotropic splittings $\sigma_1$ and $\sigma_2$ is given by $(\sigma_1 - \sigma_2)(X)= \rho^*(\iota_X B)$, for some $B \in \sfOmega^2(M)$ and for all $X \in \sfgamma(TM)$. Since $\sigma_i$ is adapted to $K_i$, it then follows that
 \begin{align}
 \Psi_1^{-1}(K_1) = T\cF_1 \quad , \quad \Psi_1^{-1}(K_2)= \e^{-B}\,(T\cF_2) \ . 
 \end{align}
Thus, as shown in \cite[Theorem 5.44]{DeFraja:2023fhe}, given a $D_1$-invariant generalised metric on $\red E{}_1$, a T-dual generalised metric on $\red E{}_2$ exists if the conditions $\pounds_X B = 0$, for all $X \in \frk{k}_2$, and $T\cF_1 \cap T\cF_2 \subseteq \ker(B^\flat)$ hold.
\end{remark}

The classic example of the construction of Theorem \ref{thm:maingeneral} is given by the Cavalcanti-Gualtieri formulation of T-duality~\cite{cavalcanti2011generalized}, see \cite[Subsection 6.1]{DeFraja:2023fhe}. 

\medskip

\subsection{Local Geometric T-Duality}~\\[5pt]
It is possible to generalise the considerations of Subsection~\ref{subsect:geomTdual} to the case in which a generalised metric has only locally defined Killing vectors. 

\begin{definition}\label{def:localinvmetric}
    A generalised metric $V_1^+$ on $\red E{}_1$ is \emph{locally $D_1$-invariant} if there are local charts $\set{U^{(i)}}_{i\in I}$ covering $\cQ_1$ and generators \smash{$\set{\red X{}_1^{(i)}, \dots, \red X{}_{r_1}^{(i)}}\subset \sfgamma(U^{(i)},T\cQ_1) $} of $\sfgamma(U^{(i)},D_1)$ such that
    \begin{align}
        \ad_{\red X{}_k^{(i)}}^{\red \sigma_1}(\red w{}_1) \ \in \ \sfgamma\big(U^{(i)},V_1^+\big) \ ,
    \end{align}
    for every $k = 1,\dots,r_1$, $\red w{}_1 \in \sfgamma(U^{(i)},V_1^+)$ and $i \in I$.
\end{definition}

Note that $\text{Span}_{\IR}\set{\red X{}_1^{(i)}, \dots, \red X{}_{r_1}^{(i)}} = \text{Span}_{\IR}\set{\red X{}_1^{(j)}, \dots, \red X{}_{r_1}^{(j)}}$ on overlaps $U^{(ij)}=U^{(i)}\cap U^{(j)}$. Hence the isometry algebra $\frk{k}_{\tau_1}$ is well-defined, and so is $\frk{k}_2 \coloneqq (\varpi_1)_*^{-1}(\frk{k}_{\tau_1}) \, \cap \, \mathsf{\Gamma}(\rho(K_2))$. 
Denote by \smash{$\frk{k}^{(i)}_{\tau_1} \coloneqq \frk{k}_{\tau_1}|_{U^{(i)}} = \text{Span}_{\IR}\set{\red X{}_1^{(i)}, \dots, \red X{}_{r_1}^{(i)}}$}. If $U^{(i)}$ is a cover of $\cQ_1$, then $\cU^{(i)}\coloneqq \varpi_1^{-1}(U^{(i)})$ is a cover of $M$, hence we can also define \smash{$\frk{k}^{(i)}_2 \coloneqq \frk{k}_2|_{\cU^{(i)}}$}.

Similarly, we can speak about when splittings are locally compatible.  

\begin{definition}
    Adapted splittings $\sigma_l:TM\to E$ with respect to $K_l$ for $l=1,2$ are \emph{locally compatible} if they satisfy
    \begin{align}
        \llbracket \sigma_2(X^{(i)}) , e \rrbracket = \ad_{X^{(i)}}^{\sigma_1}\, e \ ,
    \end{align}
    for every $X^{(i)} \in \frk{k}^{(i)}_2$, $e \in \sfgamma(E)$ and $i\in I$.
\end{definition}

\begin{theorem}\label{thm:mainlocal}
    Theorem \ref{thm:maingeneral} holds in the local case, i.e. we can replace compatibility of splittings and $D_1$-invariance with their local counterparts.
\end{theorem}

\begin{proof}
    The proof is exactly the same as in \cite[Theorem 5.27]{DeFraja:2023fhe}: the algebra $\frk{k}_{\tau_1}$ only enters when one shows that the pre-transverse generalised metric $W_2$ is in fact a $K_2$-transverse generalised metric. However, the condition $\llbracket K_2, W_2 \rrbracket \subset W_2$ can be checked locally, and $\sigma_2(\frk{k}_2)$ spans $K_2$ locally.
\end{proof}

\medskip

\subsubsection{Buscher Rules}~\\[5pt]
    Let $\cB$ be the leaf space of the foliated manifold $(\cQ_1, \red \cF{}_1)$, and suppose that a trivialising open cover $U^{(i)}$ for $\cQ_1 \to \cB$ is adapted to the foliation. In the case that $\varpi_2 ( \cU^{(i)} )$ defines smooth trivialising charts for $\cQ_2$ adapted to $\cF_2$, we get smooth charts for $\red C$ trivialising $R$, and a local description of T-duality. The Buscher rules can then be obtained as in \cite[Remark~6.15]{DeFraja:2023fhe}. 
    
    \begin{example}[\textbf{Circle Bundles}]\label{rmk:buscher}
    Let $\cQ_1$ and $\cQ_2$ be T-dual circle bundles over the same base $\cB$, with trivialising charts $\cV^{(i)}$ for $\cB$, and let $U^{(i)} \coloneqq \cV^{(i)} \times \sfS^1$. Let $\theta_1$ and $\theta_2$ be connections on $\cQ_1$ and  $\cQ_2$, respectively. Set $M = \cQ_1 \times_{\cB} \cQ_2$ with the two-form $$B = \varpi_1^* \theta_1 \wedge \varpi_2^* \theta_2 \ , $$ 
    which satisfies the conditions of Remark~\ref{rmk:tdualsplitiso}.
    
    Then we get a T-duality relation 
    \begin{align}
    R \colon (\IT \cQ_1, H_1) \rel (\IT \cQ_2, H_2)
    \end{align}
    supported on $\red C = M$, where $\varpi_1^*H_1 - \varpi_2^* H_2 = \de B$.  
    Locally, if $H_1 = \de b_1$ for a Kalb-Ramond field $b_1 \in \sfOmega^2(U^{(i)})$, then $R$ is given by\footnote{Note that $\cV^{(i)}$ also trivialise the support $\red C = M$ of the relation $R$.}
    \begin{align}\label{eqn:Rlocally}
    \sfgamma\big(\cV^{(i)},R\big) = \text{Span}_{C^\infty(\cV^{(i)})}\set{(\partial_\alpha, \partial_\alpha - (b_1)_{\alpha \theta}\,\partial_{\theta_2})\,,\, (\de x^\alpha, \de x^\alpha)\,,\, (\partial_{\theta_1}, \theta_2)\,,\, (\theta_1, \partial_{\theta_2})} \ ,
\end{align}
where $\partial_\alpha$ with $\alpha=1,\dots, \dim \cB$ are horizontal coordinate vector fields on $U^{(i)}$, \smash{$\partial_{\theta_l}$} are vector fields dual to $\theta_l$ for $l=1,2$, while $\theta$ labels the coordinates of the circle direction for both $\cQ_1$ and $\cQ_2$.

This yields the Buscher rules: the components of the T-dual generalised metric $(g_2, b_2)$ on $\cQ_2$ can be written in terms of the $\sfS^1$-invariant generalised metric $(g_1, 0)$ on $\cQ_1$ as
\begin{align}\label{eqn:buscherrules}
\begin{split}
    (g_2)_{\theta \theta} = \frac{1}{(g_1)_{\theta \theta}} \ \ , & \ \ 
    (g_2)_{\alpha  \theta} = \frac{(b_1)_{\alpha \theta}}{(g_1)_{\theta \theta}} \ \ , \ \ (g_2)_{\alpha\beta} = (g_1)_{\alpha\beta} - \frac{ (g_1)_{\alpha \theta}\, (g_1)_{\beta \theta} - (b_1)_{\alpha \theta}\, (b_1)_{\beta \theta}}{(g_1)_{\theta \theta}}  \ , \\[4pt]
    (b_2)_{\alpha \theta} &= \frac{(g_1)_{\alpha\theta}}{(g_1)_{\theta \theta}} \ \ , \ \ (b_2)_{\alpha \beta} = (b_1)_{\alpha\beta} - \frac{(b_1)_{\alpha\theta}\, (g_1)_{\beta \theta} - (b_1)_{\beta \theta}\, (g_1)_{\alpha\theta}}{(g_1)_{\theta \theta}} \ .
    \end{split}
\end{align}
\end{example}

\section{Divergence Operators and Courant Algebroid Relations} \label{sect:CAreldiv}

In this section we shall discuss how Courant algebroid relations are instrumental to defining a robust notion of morphisms between divergence operators on Courant algebroids.  We will define the notion of related divergence operators on Courant algebroids admitting a Courant algebroid relation in the sense of \cite{LiBland2009, Vysoky2020hitchiker}. We will further apply this definition to the setting of T-duality given in \cite{DeFraja:2023fhe} and Poisson-Lie T-duality as described in \cite{Severa-letters, Vysoky2020hitchiker, Severa:2018pag}. 

\medskip

\subsection{Divergence Operators on Courant Algebroids}~\\[5pt]
Let us recall the definition of divergence operator on a Courant algebroid \cite{Alekseev2001, Garcia-Fernandez:2016ofz}. For further details on divergence operators on Courant algebroids, see \cite{Severa:2018pag, Garcia-Fernandez:2016ofz, Garcia-Fernandez:2020ope} and references therein.

\begin{definition}
Let $E$ be a Courant algebroid over a manifold $M$. A \emph{divergence operator}, or simply a \emph{divergence}, on $E$ is an $\IR$-linear map
\begin{align}
    \div \colon \sfgamma(E) \longrightarrow C^\infty(M)
\end{align}
satisfying an anchored Leibniz rule
\begin{align} \label{eqn:leibnizdiv}
    \div(f \, e) = f \, \div \, e + \pounds_{\rho(e)} f \ ,
\end{align}
for all $e \in \sfgamma(E)$ and $f \in C^\infty(M).$
\end{definition}

Divergence operators on a Courant algebroid $E$ form an affine space over $\sfgamma(E)$: given any two divergences $\div$ and $\div'$ on $E$, there exists a uniquely determined section $e_\div \in \sfgamma(E)$ such that
\begin{align}
    \div \, e - \div' \, e = \ip{e_\div, e} \ , 
\end{align}
for all $e \in \sfgamma(E).$

\begin{example}[\textbf{Densities}] \label{eg:mudiv}
    Let $E$ be a Courant algebroid over $M$ and let $\mu$ be a nowhere-vanishing density on $M$. The divergence operator $\div_\mu$ on $E$ associated with $\mu$ is defined by
    \begin{align}\label{eqn:mudivergences}
        \div_{\mu} \, e \coloneqq \mu^{-1}\,\pounds_{\rho(e)}\mu \ ,
    \end{align}
    for all $e \in \sfgamma(E)$.
\end{example}

\begin{remark} \label{rmk:associatedsection}
    Consider a Courant algebroid $E$ over $M$ endowed with a divergence operator  $\div$. Suppose that $M$ admits a nowhere-vanishing density $\mu$. Then there is a section $e_{\div} \in \sfgamma(E)$ induced by the difference
\begin{align}
    \div \, e - \div_{\mu}\, e =\ip{e_{\div}, e} \ ,
\end{align}
for all $e \in \sfgamma(E).$ The adjoint action $${\sf ad}_{{\div}} \coloneqq \cbrak{e_{\div} \, , \, \cdot\,} \ : \ \sfgamma(E)\longrightarrow\sfgamma(E)$$ associated with the divergence operator $\div$ is independent of the choice of density $\mu,$ see \cite{Severa:2018pag, Garcia-Fernandez:2016ofz}.
\end{remark}

\medskip

\subsubsection{The Dilaton}~\\[5pt] \label{sub:dilaton}
One application of divergence operators is to provide a definition of the missing piece of a generalised string background $(g,H)$ on a Courant algebroid, namely the dilaton. We follow \cite{Severa:2018pag, Garcia-Fernandez:2020ope} and \cite[Definition 3.7]{Streets:2024rfo} for the basic definitions and the main properties of dilatons.

\begin{definition}
A \emph{dilaton} on an exact Courant algebroid $E$ over $M$ with generalised metric $\tau$ is a divergence $\text{div}_\phi $ such that
\begin{align}
    \text{div}_\phi \, e - \div_{\mu_g} \, e = \langle \rho^*(\de \phi) , e \rangle \ ,
\end{align}
for some $\phi \in C^\infty(M)$ and for all $e \in \sfgamma(E),$ where $(g, \sigma)$ is the pair of a Riemannian metric $g$ with induced volume form $\mu_g$ on $M$ and an isotropic splitting $\sigma:TM\to E$ associated with the generalised metric $\tau$. The function $\phi$ is a \emph{dilaton field}.
\end{definition}

\begin{example}
Let $H \in \sfOmega^3_{\rm cl}(M)$ and let $\IT M$ be an $H$-twisted standard Courant algebroid over an oriented manifold $M,$  endowed with a generalised metric given by $V^+= \gr(g)$ where $g$ is a Riemannian metric on $M$. Any divergence $\div$ on $\IT M$ is associated with some section $X_\div + \alpha_\div \in \sfgamma(\IT M)$ by 
\begin{align}
    \div - \div_{\mu_g} = \ip{ X_\div+\alpha_\div , \, \cdot \, } \ ,
\end{align}
where $\mu_g$ is the volume form associated to $g$. Then the divergence operator $\div$ is \emph{compatible} with $V^+$, i.e. $\mathsf{ad}_\div(\sfgamma(V^+)) \subseteq \sfgamma(V^+)$, if and only if
\begin{align}
    \pounds_{X_\div} g = 0 \quad , \quad \de \alpha_\div - \iota_{X_\div} H =0 \ .
\end{align}

If we take $$\div = \div_{\e^{-2\phi}\, \mu_g} \ , $$  where $\phi \in C^\infty(M),$ then $$X_\div = 0 \quad , \quad \alpha_\div = -2\,\de \phi 
 \ , $$ which are the usual equations for the appearance of the dilaton field $\phi$ in supergravity. Thus $ \div_{\e^{-2\phi}\,\mu_g} = \div_{-2\phi}$ in this example.
\end{example}

\medskip

\subsection{Related Divergence Operators}~\\[5pt]
We will now give a notion of when two divergence operators are related by a Courant algebroid relation.

\begin{definition} \label{def:relateddiv}
Let ${\rm div}_1$ and ${\rm div}_2$ be divergence operators on Courant algebroids $E_1$ and $E_2$ over manifolds $M_1$ and $M_2$, respectively. Let $R \colon E_1 \rel E_2$ be a Courant algebroid relation supported on a submanifold $C\subseteq M_1\times M_2$. Then $\div_1$ and $\div_2$ are \emph{$R$-related}, denoted ${\rm div}_1 \sim_R {\rm div}_2$, if ${\rm div}_1\, e_1 \sim_C {\rm div}_2\, e_2$ for every $(e_1,e_2) \in \sfgamma(E_1 \times \overline E_2 ;R)$.
\end{definition}

Let us check that this notion of related divergence operators on Courant algebroids is well-defined.

\begin{lemma}
 Relations between divergences in the sense of Definition \ref{def:relateddiv} are well-defined: if $(e_1,e_2), (e_1',e_2') \in \sfgamma(E_1 \times \overline{E}_2 ; R)$ such that $(e_1, e_2)_c = (e_1', e_2')_c$ for all $c \in C,$ then $\div_1\, e_1 \sim_C \div_2\, e_2$ if and only if $\div_1\, e_1' \sim_C \div_2\, e_2'$.
\end{lemma}

\begin{proof}
    Define $\div \colon \sfgamma(E_1) \times \sfgamma(E_2) \to C^\infty(M_1) \times C^\infty(M_2)$ by $\div = (\div_1, \div_2)$, and extend this to $\widetilde \div \colon \sfgamma(E_1 \times \overline{E}_2) \to C^\infty(M_1 \times M_2)$. 
    Suppose that $r,r' \in \sfgamma(E_1 \times \overline{E}_2 ; R)$ with $r(c) = r'
    (c)$ for all $c\in C$. Then for any open subset $U \subset M_1 \times M_2,$ there is a frame $\set{r^i}\subset \sfgamma(U,E_1\times \overline{E}_2;R)$ such that $r = \sum_i\, f_i\, r^i$ and $r' = \sum_i\, f'_i\, r^i$ with $f_i,f_i'\in C^\infty(U)$ satisfying $f_i(c) = f'_i(c)$ for all $c\in C\cap U$. Denoting the anchor of $E_1 \times \overline{E}_2$ by $\rho$, we then find
    \begin{align}
        \widetilde \div \, r = \sum_i\, \big( f_i \,\widetilde\div \, r^i + \pounds_{\rho(r^i)} f_i \big) 
        = \sum_i\, \big( f'_i\, \widetilde\div \, r^i + \pounds_{\rho(r^i)} f'_i \big) = \widetilde \div \, r' 
    \end{align}
    on $C\cap U$, where we used compatibility of the relation with the anchor which implies that $\rho(r)\in \sfgamma\big(T(M_1 \times M_2);TC\big)$ for $r \in \sfgamma(E_1 \times \overline{E}_2;R)$. Thus, on $C$, $\widetilde\div$ is independent of the chosen section of $\sfgamma(E_1 \times \overline{E}_2; R)$, and hence $\div_1 \sim_R \div_2$ is well-defined.
\end{proof}

\begin{example}[\textbf{Classical Courant Algebroid Isomorphisms}]
  Let $E_1 \to M_1$ and $E_2 \to M_2$ be Courant algebroids. Let $\boldsymbol\Phi \in {\sf Hom}(E_1, E_2)$ be a Courant algebroid isomorphism covering a diffeomorphism $\Fi \in \mathsf{Diff}(M_1, M_2).$ If $\div_1$ is a divergence operator on $E_1,$ then $E_2$ is endowed with a divergence operator $\div_2$ given by 
  \begin{align} \label{eqn:divIso}
      \div_2 \coloneqq (\Fi^{-1})^* \circ \div_1 \circ \boldsymbol\Phi^{-1} \ .
  \end{align}
  It is easy to show that $\div_2$ defined in Equation~\eqref{eqn:divIso} satisfies the anchored Leibniz rule~\eqref{eqn:leibnizdiv}.
    The divergence operators $\div_1$ and $\div_2$ are $\gr(\boldsymbol\Phi)$-related.
\end{example}

\begin{example}[\textbf{Courant Algebroid Morphisms}]
    Let $R \colon E_1 \mor {E}_2$ be a Courant algebroid morphism over a smooth map $\Fi \colon  M_1 \to M_2$. Then the condition of Definition \ref{def:relateddiv} for $\div_1$ and $\div_2$ to be $R$-related divergences becomes
    \begin{align}\label{eqn:classicalCARdiv}
        \div_1 \, e_1 =\Fi^*(\div_2 \, e_2) \ ,
    \end{align}
    for all $(e_1, e_2) \in \sfgamma(E_1 \times \overline{E}_2 ; R).$

    Consider the divergences given by
\begin{align}
    {\rm div}_{\mu_i} \, e_i = \mu_i^{-1}\, \pounds_{\rho_{i}(e_i)} \mu_i \ ,   
\end{align}
    where $\mu_i$ are nowhere-vanishing densities on $M_i$ for $i=1,2$. Then Equation \eqref{eqn:classicalCARdiv} implies that $\div_{\mu_1} \sim_R \div_{\mu_2}$ if and only if $\mu_1 = \varphi^* \mu_2$.
\end{example}

\begin{example}[\textbf{Reduction by Isotropic $\sfG$-Actions}] \label{eg:isoreddiv}
    Let $E$ be a Courant algebroid over $M$ endowed with an isotropic action of a Lie algebra $\frg$, i.e. with a linear map $\chi \colon \frg \to \sfgamma(E)$ preserving the brackets such that $K = \chi(\frg)$ is an isotropic  subbundle of $E$. Let $\sfG$ be the connected Lie group integrating $\frg$, and suppose it acts freely and properly on $M$. Thus $\cQ = M / \sfG$ is a smooth manifold, with smooth quotient map  $\varpi:M\to\cQ$. Theorem~\ref{thm:reduction} then constructs the reduced Courant algebroid $\red E$ over $\cQ$ such that $\sfgamma(\red E) \cong \sfgamma_{\sfG}(K^\perp)\,/\, \sfgamma_{\sfG}(K)$, where the subscript ${}_\sfG$ denotes the $\sfG$-invariant sections. 

    A divergence operator $\div$ on $E$ is \emph{$\sfG$-equivariant} if it obeys
    \begin{align} \label{eqn:Ginvdiv}
        \pounds_{\rho(\chi(x))} \div \, e = \div\, \cbrak{\chi(x), e} \quad , \quad \div\, \chi(x) = - \Tr_{\frg}( {\sf ad}_x) \ , 
    \end{align}
    for all $e \in \sfgamma(E)$ and $x \in \frg$. Then by \cite[Proposition 5.11]{Severa:2018pag}, there exists a divergence operator $\red \div$ on the reduced Courant algebroid $\red E\,$, since $\div$ restricts to a map $\sfgamma_{\sfG}(E) \to C^\infty_{\sfG}(M)$ by the first condition of Equation~\ref{eqn:Ginvdiv} and $\div \,  \sfgamma_{\sfG}(K) = \set{0}$ by the second condition. In other words, the reduced divergence operator $\red \div$ satisfies
    \begin{align}
        \varpi^* (\red \div \, \red e) = \div \,  e \ ,
    \end{align}
    for any pair $(e, \red e) \in \sfgamma(E \times \overline{\red E}; Q(K))$, where $Q(K)$ is the reduction morphism constructed as in Example~\ref{eg:reductionrel}. Thus $\div \sim_{Q(K)} \red \div$.
\end{example}

\begin{example}[\textbf{Heterotic Reduction}]\label{eg:Heteroticreddiv}
    Let us consider the same setting as in Example \ref{eg:isoreddiv}, with a few variations appropriate to applications in heterotic string theory~\cite{Baraglia:2013wua}: suppose that $K= \chi(\frg)$ satisfies $K \cap K^\perp= \set{0}$ and that $\div$ only satisfies the first condition of Equation~\eqref{eqn:Ginvdiv}. Then $\div$ still restricts to a map $\sfgamma_{\sfG}(E) \to C^\infty_{\sfG}(M)$ and, because of our assumption on $K$, it descends to a divergence $\red \div$ on the reduced Courant algebroid $\red E = K^\perp /\, \sfG$ constructed in~\cite{Bursztyn2007reduction,Baraglia:2013wua}. Again the divergences $\div$ and $\red \div$ are $Q(K)$-related.
\end{example}

\begin{example}[\textbf{Poisson-Lie T-Duality}] \label{eg:vysokypoissonlie}
    As discussed in Section~\ref{sec:intro}, a particularly relevant example comes from the description of Poisson-Lie T-duality, as given in \cite{Severa2015}.
In the setting of \cite[Section~4.4]{Vysoky2020hitchiker}, let $E$ be a $\sfG$-equivariant exact Courant algebroid over a manifold $M$ with a free and proper action of a connected Lie group $\sfG$, whose Lie algebra $\frg$ is endowed with a split signature pairing $\ip{\,\cdot\,,\, \cdot\,}_\frg$ and an extended action $\chi \colon \frg \to \sfgamma(E)$. Suppose that $K=\chi(\frg)$ satisfies $K\cap K^\perp = \set{0}$, and let $\varpi_\sfG \colon M \to \cB = M/\sfG$ be the quotient map.

    Let $\frh\subset\frg$ be a Lie subalgebra which integrates to a closed connected subgroup $\sfH\subset\sfG$ such that $K_\sfH=\chi(\frh)$ is an isotropic  subbundle of $E$. Let $\varpi_\sfH \colon M \to \cQ = M/\sfH$ be the quotient map. Then there is a Courant algebroid morphism $R(\sfH) \colon \red E{}_\sfH \mor \red E$ over $\Fi$ between the reduced Courant algebroids
\begin{equation}
    \red E = K^\perp /\, \sfG \quad , \quad \red E{}_\sfH = \frac{K_\sfH^\perp}{K_\sfH} \, \Big / \, \sfH \ , 
\end{equation}
where $\red E\to\cB$ is a transitive Courant algebroid such that \smash{$\red E\,/\,\Im(\,\red\rho^*) = TM/\sfG$} and $\varphi:\cQ\to\cB$ is the canonical surjective submersion induced by the subgroup inclusion $\sfH\subset\sfG$, which fits into the commutative diagram
\begin{equation}
\begin{tikzcd}
    & M \arrow[dd,"\varpi_\sfG"] \arrow[dl,"\varpi_\sfH"']  \\
   \cQ \arrow[swap,dr, "\Fi"] & \\ & \cB
\end{tikzcd}
\end{equation}

If $\frh$ is a Lagrangian subalgebra, i.e. $\frh^\perp = \frh$, then $K_\sfH^\perp = K_\sfH\oplus K^\perp$ so that $K_\sfH^\perp/K_\sfH\cong K^\perp$ and $\red E{}_\sfH\cong K^\perp/\sfH$. Then the reduced Courant algebroid $\red E{}_{\sfH}\to\cQ$ is exact and is naturally the Courant algebroid pullback $\red E{}_{\sfH} \cong \varphi^*\red E$ of $\red E\to M$ (see Section~\ref{sec:intro}). It follows that $R(\sfH) = \gr(\boldsymbol\Phi_\sfH)$ for a classical Courant algebroid morphism $$\boldsymbol\Phi_\sfH \colon \red E{}_\sfH \longrightarrow \red E$$ over $\varphi$~\cite[Example 4.31]{Vysoky2020hitchiker}. If $\frh'$ is any other Lagrangian subalgebra of $\frg$ with integrating Lie group $\sfH'$, then also $R(\sfH') = \gr(\boldsymbol\Phi_{\sfH'})$ for  a classical Courant algebroid morphism $\boldsymbol\Phi_{\sfH'}:\red E{}_{\sfH'}\to\red E$ over the surjective submersion $\varphi':\cQ'\to\cB$, where $\cQ'=M/\sfH'$. 

The composition $$R_{\sfH,\sfH'} \coloneqq \gr(\boldsymbol\Phi_{\sfH'})^\top\circ \gr(\boldsymbol\Phi_\sfH) \ : \ \red E{}_\sfH\rel\red E{}_{\sfH'}$$ is defined and  supported on the fibred product $\cQ \times_\cB \cQ'$; it is equal to the fibred product $\red E{}_\sfH \times_{\red E} \red E{}_{\sfH'}$. Note that if $\frh\cap\frh'=\set{0}$, then $(\frg,\frh,\frh')$ is a Manin triple and $(\sfH,\sfH')$ is a dual pair of Poisson-Lie groups. Generally, the relation $R_{\sfH,\sfH'}$ is called a \emph{Poisson-Lie T-duality relation} (with spectators $\cB$), which we may view diagrammatically as
\begin{equation}
    \begin{tikzcd}
    & \red E   & \\
   \red E{}_\sfH \arrow[ur,"\boldsymbol\Phi_\sfH"] \arrow[rr,dashed, "R_{\sfH,\sfH'}"] & & \red E{}_{\sfH'} \arrow[ul,"\boldsymbol\Phi_{\sfH'}"']
\end{tikzcd}
\end{equation}
When $\sfG$ is abelian, it coincides with the T-duality relation for torus bundles~\cite{DeFraja:2023fhe} (see Subsection~\ref{sec:torusbundles} below).

In the spirit of \cite[Example 5.9]{Vysoky2020hitchiker}, if $\div$ is a $\sfG$-equivariant divergence on $E$ as in Example \ref{eg:Heteroticreddiv}, then $\red E$ inherits a divergence operator $\red \div\,$, so that $\div_\sfH \coloneqq \boldsymbol\Phi_\sfH^* (\red \div)$ and $  \div_{\sfH'} \coloneqq \boldsymbol\Phi_{\sfH'}^*( \red \div)$ define divergences on $\red E{}_\sfH$ and $ \red E{}_{\sfH'}$ respectively. It then follows that $\div_\sfH \sim_{R_{\sfH,\sfH'}} \div_{\sfH'}$.     
\end{example}

Recall that divergences on a Courant algebroid $E$ form an affine space over $\sfgamma(E).$ This leads to

\begin{lemma} \label{lem:relsectionsdiv}
    Let $R \colon E_1 \rel E_2$ be a Courant algebroid relation supported on $C\subseteq M_1\times M_2$, and let $\div_1 \sim_R \div_2$ be $R$-related divergence operators. Let $\div_i' -\div_i = \ip{e_{\div_i}\,,\, \cdot\,}_i $ for some $e_{\div_i} \in \sfgamma(E_i)$, for $i=1,2$. Then $\div_1' \sim_R \div_2'$ if and only if $(e_{\div_1},e_{\div_2}) \in \sfgamma(E_1 \times \overline{E}_2 ; R^\perp)$. If the relation $R$ is a Dirac structure, then $e_{\div_1}\sim_R e_{\div_2}$.
\end{lemma}

\begin{proof}
    Let $e' = (e_1',e_2') \in \sfgamma(E_1 \times \overline{E}_2 ; R)$. Since $\div_{1} \sim_R \div_{2}$, it follows that $\div_1' \sim_R \div_2'$ if and only if $\ip{e_{\div_1}, e_1'}_1 \sim_C \ip{e_{\div_2}, e_2'}_2$. Thus, for any $c = (m_1, m_2) \in C$, this requires
    \begin{align}
        0 = \ip{e_{\div_1},e_1'}_{1}(m_1) - \ip{e_{\div_2},e_2'}_{2}(m_2) = \ip{(e_{\div_1}, e_{\div_2})\,,\, e'}(c) \ .
    \end{align}
    Since $e'$ was arbitrary, it follows that $(e_{\div_1},e_{\div_2}) \in \sfgamma(E_1 \times \overline{E}_2; R^\perp)$. If $R$ is a Dirac structure, then $R^\perp= R.$
\end{proof}

When the relation $R$ has trivial kernel and cokernel, there is a uniqueness result for divergences given by

\begin{lemma} \label{lem:uniquereldiv}
    Let $R \colon E_1\rel E_2$ is a Dirac structure with trivial kernel and cokernel, supported on $C\subseteq M_1\times M_2$ such that ${\rm pr}_2(C) = M_2$. Let $\div_1$ be a divergence on $E_1$. Then there is at most one divergence $\div_2$ on $E_2$ such that $\div_1 \sim_R \div_2$.
\end{lemma}

\begin{proof}
    Suppose that divergence operators $\div_2$ and $\div_2'$ on $E_2$ are $R$-related to $\div_1$. Then there is a section $e \in \sfgamma(E_2)$ such that $\div_2 - \div_2' = \ip{e, \,\cdot\,}_2$. Thus at a point $m_2\in M_2$ with $(m_1,m_2) \in C$ for some $m_1\in M_1$, we get
    \begin{align}
        0 = (\div_1\, e_1 - \div_1\, e_1)(m_1) = (\div_2\, e_2 - \div'_2\, e_2)(m_2) = \ip{e, e_2}_2(m_2)
    \end{align}
    for all $e_2 \in (E_2)_{m_2}$, with $e_1 \in (E_1)_{m_1}$ such that $(e_1, e_2) \in R_{(m_1, m_2)}$. The Dirac condition, combined with the trivial kernel and cokernel condition, ensure that we can find such an element $e_1 \in (E_1)_{m_1}$ for all $m_2  \in M_2$ and $e_2 \in (E_2)_{m_2}$. It follows that $e=0$.
\end{proof}

Let us consider two $R$-related Courant algebroids $E_1$ and $E_2$ endowed, respectively, with divergence operators $\div_1$ and $\div_2$. Suppose that $M_1$ and $M_2$ admit nowhere-vanishing densities $\mu_1$ and $\mu_2$. Then we get sections $e_{\div_1}$ and $e_{\div_2}$ of $E_1$ and $E_2$, respectively, defined as in Remark \ref{rmk:associatedsection}.

\begin{proposition}
Suppose that the divergence operators $(\div_1, e_{\div_1})$ and $(\div_2, e_{\div_2})$ on the Courant algebroids $E_1$ and $E_2$ are $R$-related, where $R$ is  a Dirac structure in $ E_1 \times \overline{E}_2$.
Then the derivations ${\sf ad}_{\div_1} \coloneqq \cbrak{e_{\div_1} \,,\, \cdot\,}_1$ and ${\sf ad}_{\div_2} \coloneqq \cbrak{e_{\div_2} \,,\, \cdot\,}_2$  are $R$-related.  
\end{proposition}

\begin{proof}
Since $e_{\div_1} \sim_R e_{\div_2},$ by definition $(e_{\div_1}, e_{\div_2}) \in \sfgamma(E_1 \times \overline{E}_2 ; R)$. By definition of the Dorfman bracket on the product Courant algebroid $E_1 \times \overline{E}_2$ we get
\begin{align}
    \cbrak{(e_{\div_1}, e_{\div_2}) \,,\, (\, \cdot \, ,\, \cdot \, )} =(\cbrak{e_{\div_1},\, \cdot \,}_1\,,\, \cbrak{e_{\div_2}, \, \cdot \,}_2) = ({\sf ad}_{\div_1}, {\sf ad}_{\div_2}) \ . 
\end{align}
Since $R$ is involutive, it follows that
\begin{align}
\cbrak{(e_{\div_1}, e_{\div_2}) , (e_1', e_2')} = ({\sf ad}_{\div_1}e'_1, {\sf ad}_{\div_2}e'_2) \ \in \ \sfgamma(E_1 \times \overline{E}_2 ; R) \ ,   
\end{align}
for all $(e'_1, e'_2) \in \sfgamma(E_1 \times \overline{E}_2 ; R).$ Hence ${\sf ad}_{\div_1} \sim_R {\sf ad}_{\div_2}$ in the sense of \cite[Definition 6.1]{Vysoky2020hitchiker} extended to covariant differential operators on $E_1$ and $E_2$.
\end{proof}

Note that this result holds only if $R$ is a Dirac structure, since it relies on the fact that the sections $e_{\div_1}$ and $e_{\div_2}$ are $R$-related. 

\medskip

\subsubsection{Divergences from Connections}~\\[5pt]
A broad class of divergence operators is obtained via Courant algebroid connections. A \emph{Courant algebroid connection} $\nabla^E$ on a Courant algebroid $E\to M$ is an $\IR$-bilinear map 
\begin{align}
\nabla^E \colon \sfgamma(E) \times \sfgamma(E) \longrightarrow \sfgamma(E) \ ,
\end{align}
denoted $(e_1,e_2)\mapsto\nabla^E_{e_1}e_2$, which satisfies
\begin{align}
    \nabla^E_{f \, e_1} e_2 =  f\, \nabla^E_{e_1} e_2 \quad , \quad
    \nabla^E_{e_1} (f \, e_2) = f\, \nabla^E_{e_1} e_2 + \left( \pounds_{\rho(e_1)} f \right)\, e_2 \ ,
\end{align}
together with the metric compatibility condition
\begin{align}
     \langle\nabla^E_e e_1, e_2\rangle  + \langle e_1, \nabla^E_e e_2\rangle = \pounds_{\rho(e)}\langle e_1, e_2\rangle   \ ,
\end{align}
for all $e, \, e_1, \, e_2 \in \sfgamma(E)$ and $f \in C^\infty(M)$.

Following \cite{Garcia-Fernandez:2016ofz, Garcia-Fernandez:2020ope}, the divergence operator of a Courant algebroid connection $\nabla^E\colon \sfgamma(E) \times \sfgamma(E) \to \sfgamma(E) $ is defined by
\begin{align} \label{eqn:divergconnection}
    \div_{\nabla^E} (e) \coloneqq \Tr(\nabla^E e) \ ,
\end{align}
for $e \in \sfgamma(E),$ where the trace operator $\Tr \in \sfgamma( {\sf End}(E)^*)$ is given on decomposable endomorphisms by $\Tr(\alpha \otimes e) \coloneqq \alpha(e)$ for  $\alpha \in \sfgamma(E^*)$ and $e \in \sfgamma(E).$ We shall discuss the circumstances under which divergences defined by Equation \eqref{eqn:divergconnection} are related as in Definition \ref{def:relateddiv} if the underlying Courant algebroid connections are related.

Courant algebroid connections $\nabla^1$ on $E_1$ and $\nabla^2$ on $E_2$ are \emph{$R$-related}, denoted $\nabla^1 \sim_R \nabla^2$, if $\nabla^1_{e_1} e'_1 \sim_R \nabla^2_{e_2}e'_2$ for any $(e_1, e_2), \, (e'_1, e'_2) \in \sfgamma(E_1 \times \overline{E}_2 ; R)$; see \cite{Vysoky2020hitchiker} for more details. 

\begin{proposition} \label{prop:reldivfromcon}
  Let $R \colon E_1 \dashrightarrow E_2$ be a Courant algebroid relation supported on $C\subseteq M_1\times M_2$ with trivial kernel $\Ker(R)$ and cokernel $\Coker(R)$ (see Definition \ref{def:kernelsR}), which is a Dirac structure in $E_1 \times \overline{E}_2$. Let $\nabla^1 \sim_R \nabla^2$ be $R$-related Courant algebroid connections on $E_1$ and $E_2$. Then $\div_{\nabla^1} \sim_R \div_{\nabla^2}$.  
\end{proposition}

\begin{proof}
Because of $C^\infty(M_i)$-linearity in the first entry of $\nabla^i$ for $i=1,2$, we can define linear maps $A^1_{e_1} = \nabla^1 e_1 \colon \sfgamma(E_1) \to \sfgamma(E_1)$, for any $e_1 \in \sfgamma(E_1),$ and $A^2_{e_2} = \nabla^2 e_2 \colon \sfgamma(E_2) \to \sfgamma(E_2),$ for any $e_2 \in \sfgamma(E_2)$. Thus if $(e_1, e_2) \in \sfgamma(E_1 \times \overline{E}_2;R),$ then
$A^1_{e_1} (e_1') \sim_R A^2_{e_2}(e_2')$
for any pair $(e_1', e_2') \in \sfgamma(E_1 \times \overline{E}_2 ; R).$ Since $R$ is a Dirac structure with trivial kernel and cokernel, there are $R$-related local frames $\set{\psi^1_i}$ for $E_1$ and $\set{\psi^2_i}$ for $E_2$ coming from a choice of any local frame for the bundle $R \to C$.\footnote{As in \cite{Vysoky2020hitchiker}, we always assume that $C$ is an embedded submanifold of $M_1\times M_2$.} At any point $c  \in C$, it follows that $(A^1_{e_1})^j_k = (A^2_{e_2})^j_k$ with respect to the $R$-related bases $\set{\psi^1_i}$ and $\set{\psi^2_i}$. Therefore
$\Tr_1(\nabla^1 e_1)_{{\rm pr}_1(c)} = \Tr_2(\nabla^2 e_2)_{{\rm pr}_2(c)}$, for all $c \in C$.
\end{proof}

Note that the trace operator in Equation \eqref{eqn:divergconnection} constrains the type of Courant algebroid relations that are suitable to bridge the notion of related connections with related divergences. On the other hand, for the relevant applications of this paper, the assumptions of Proposition~\ref{prop:reldivfromcon} always hold since the T-duality relation will always be a generalised isometry that is also a Dirac structure.

\medskip

\subsubsection{The Dilaton Shift}~\\[5pt]
 The natural step to take now is to see how dilatons, discussed in Subsection~\ref{sub:dilaton}, behave when their corresponding divergences are related. In particular, we address how a generalised isometry yields a natural definition for the dilaton shift, see for instance \cite{Jurco2018}.

For $i=1,2$, let $E_i\to M_i$ be exact Courant algebroids endowed with generalised metrics $\tau_i$ and divergence operators $\div_i$. Define sections $e_{\div_i}\in\sfgamma(E_i)$ by $\langle e_{\div_i},\,\cdot\,\rangle_i=\div_i-\div_{\mu_{g_i}}$, where $\mu_{g_i}$ is the volume form on $M_i$ induced by the generalised metric $\tau_i$. Suppose that $R \colon (E_1 , \tau_1) \dashrightarrow (E_2, \tau_2)$ is a generalised isometry which is a Dirac structure in $E_1\times\overline E_2$ such that $\div_1\sim_R\div_2$. 

For any $(e, {}^Re) \in \sfgamma(E_1\times \overline{E}_2; R)$, we then find
\begin{align}
    \ip{e_{\div_1},e}_1+\div_{\mu_{g_1}}\, e = \div_1\, e = \div_2\, {}^Re = \ip{e_{\div_2},{}^Re}_2+\div_{\mu_{g_2}}\, {}^Re \ .
\end{align}
Write $e_{\div_2} - {}^Re_{\div_1} =: e_{\tt dil}\in\sfgamma(E_2)$, where ${}^Re_{\div_1}$ is the unique section of $E_2$ with $(e_{\div_1}, {}^Re_{\div_1}) \in \sfgamma(E_1 \times \overline{E}_2 ; R)$. Since $R$ is a Dirac structure, we then get
\begin{align}
    \ip{e_{\div_2}, {}^Re}_2 + \div_{\mu_{g_2}}\,{}^Re &= \ip{{}^Re_{\div_1}, {}^Re}_2 + \ip{e_{\tt dil}, {}^Re}_2 + \div_{\mu_{g_2}}\,{}^Re \\[4pt]
    &= \ip{e_{\div_1}, e}_1 + \ip{e_{\tt dil}, {}^Re}_2 + \div_{\mu_{g_2}}\,{}^Re \ .
\end{align}
This leaves
\begin{align}\label{eqn:dilatonshift}
    \ip{e_{\tt dil},{}^Re}_2 = \div_{\mu_{g_1}}\,e - \div_{\mu_{g_2}}\, {}^Re \ ,
\end{align}
which measures the failure of the divergences $\div_{\mu_{g_1}}$ and $\div_{\mu_{g_2}}$ from being $R$-related.

\begin{definition}
    The \emph{dilaton shift} is the section 
    \begin{align} 
    e_{\tt dil} = e_{\div_2} - {}^Re_{\div_1}
    \end{align}
    of $E_2$, given by Equation \eqref{eqn:dilatonshift} for $(e,{}^Re)\in\sfgamma(E_1\times\overline E_2;R)$.
\end{definition}

\medskip

\subsection{Invariant Divergence Operators}~\\[5pt]
 In order to introduce a suitable notion of invariance for divergence operators, let us first characterise the sections of a Courant algebroid $ E\to M$ which are invariant with respect to a distribution $D\subset TM.$ In the setting of T-duality, $D_1\subset T\cQ_1$ is the isometry distribution for a generalised metric on $\red E{}_1$, which is induced by the reduction of Courant algebroids $Q(K_2):E\mor\red E{}_2$.
The $D_1$-invariant sections are the sections of $\red E{}_1$ which lift to basic sections of $K_2^\perp$ in the setting of Theorem \ref{thm:maingeneral}. Thus we provide a notion inspired by Definition \ref{def:D1invariance}. 

\begin{definition} \label{def:invariantsection}
Let $ E$ be an exact Courant algebroid over $M$
and let $D \subset TM$ be a distribution. Suppose that there exists a Lie subalgebra $\frk{k}$ of $\sfgamma(D)$ that spans $D$ pointwise.  
Given a splitting $ \sigma:TM\to E$, the space of \emph{$D$-invariant sections} of $ E$ is given by
\begin{align}
    \sfgamma_{D}( E) := \set{ e \in \sfgamma( E) \ | \ {\sf ad}^{ \sigma}_{ X}\,  e = 0  \ \ \ \text{for all} \  X \in \frk{k}} \ .
\end{align}

Given an open cover $\set{U^{(i)}}_{i\in I}$ of $M$ and a local presentation $\set{\frk{k}^{(i)}}_{i\in I}$ of $\frk{k}$, the space of \emph{locally $D$-invariant sections} of $ E$ is
\begin{align}
    \sfgamma^{(i)}_{D}( E) := \set{ e \in \sfgamma( E) \ | \ {\sf ad}^{ \sigma}_{ X^{(i)}}\,  e = 0 \ \ \ \text{for all} \  X^{(i)} \in \frk{k}^{(i)}} \ .
\end{align}
\end{definition}

\begin{remark} 
  The existence of a splitting $ \sigma$ such that Definition \ref{def:invariantsection} can be stated represents a key requirement and a non-trivial condition to guarantee.
  If there exists a splitting $ \sigma$ such that Definition \ref{def:invariantsection} holds, then there exists a subset of the affine space of splittings given by all the splittings that differ by a $D$-invariant two-form on $M$ that can be used to state  Definition \ref{def:invariantsection} equivalently.    
\end{remark}

The space of (locally) $D$-invariant sections is involutive, see Lemma \ref{cor:invariantsinvolutive}, and is a $C^\infty_{D}(M)$-module, 
with \smash{$C^\infty_{D}(M) \cong C^\infty_{\text{bas}}(M,\cF)$}, where $\cF$ is the foliation of $M$ that integrates the distribution~$D$. 

 Definition \ref{def:invariantsection} recovers the usual notion of $\frg$-invariant sections for the action of a Lie algebra $\frg$. In that instance $\ad_{\#_{M} (x)}^{\sigma_\chi} = \llbracket\chi(x), \,\cdot\,\rrbracket$ for $x\in\frg$, where $\chi:\frg\to\sfgamma(E)$ is a trivially extended action in the sense of \cite[Definition 2.12]{Bursztyn2007reduction}, with Lie algebra action on $TM$ denoted by $\#_{M} \colon \frg \to \sfgamma(T M)$, and $\sigma_\chi$ is a  splitting adapted to $\chi(\frg)$ such that $\Im(\sigma_\chi\circ \#_{M}) = \Im(\chi)$.  Further motivation for Definition~\ref{def:invariantsection} comes from the desire to have a notion of invariance even when there is no Lie or Courant algebra action on the Courant algebroid $ E.$

We shall also provide the counterpart, in this setting, of the notion of equivariant divergence operator with respect to a group action on a Courant algebroid, see e.g. \cite[Equation (13)]{Severa:2018pag}.

\begin{definition} \label{def:K2compdiv}
   Let $D \subset TM$ be a distribution which is  spanned pointwise by a (local) Lie subalgebra $\frk{k} \subset \sfgamma(D).$ A divergence operator $\div$ on $ E$ is (\emph{locally}) \emph{compatible with} $D$ if
    \begin{align} \label{eqn:adjointdivact}
         \pounds_{ X} \, \div \, e =  \div \, {\sf ad}^{ \sigma}_{ X}\,  e \ ,
    \end{align}
    for every $ e \in \sfgamma( E)$ and $ X \in \frk{k}$ (respectively $ X \in \frk{k}^{(i)} $ for all $i \in I$).
\end{definition}

\begin{lemma} \label{lem:compatiblediverg}
If $\div$ is a (locally) $D$-compatible divergence, then $\pounds_{ X}  \, \div \, e^\circ = 0$ for all $ e^\circ \in \sfgamma_{D}( E)$.
Conversely, if the space of (locally) $D$-invariant sections spans $E$ pointwise and $\pounds_{ X} \, \div \,  e^\circ = 0$ for all $ X \in \frk{k}$ and $e^\circ \in \sfgamma_{D}( E)$, then $\div$ is (locally) compatible with $D$.
\end{lemma}  

\begin{proof}
The first statement follows easily from Equation~\eqref{eqn:adjointdivact}.

For the second statement, assume that $ E$ is  spanned pointwise by $\sfgamma_{D}( E)$, i.e. any section $e \in \sfgamma( E)$ can be expressed as $e = \sum_i\, f_i\, e^\circ_i$ where $\set{f_i}\subset C^\infty(M)$  and $\set{e^\circ_i} \subset \sfgamma_{D}( E)$,  and that $\pounds_X \, \div \, e^\circ = 0$ for all  $X \in \frk{k}$ and $e^\circ \in \sfgamma_{D}( E)$. Then
    \begin{align}
     \pounds_X\,\div\,e =   \pounds_X \,  \div \, \Big( \sum_i\, f_i\, e^\circ_i \Big) &=\sum_i\, \big( (\pounds_X f_i)\,  \div\, e^\circ_i + f_i\, \pounds_X \, \div\, e^\circ_i + \pounds_X\, \pounds_{\rho(e^\circ_i)}f_i \big) \\[4pt]
        &=\sum_i\, \big( (\pounds_X f_i )\, \div\, e^\circ_i + \pounds_{\rho(e^\circ_i)}\, \pounds_X f_i \big) = \div\, \Big( \sum_i \, (\pounds_X f_i)\,e^\circ_i \Big) \ .
    \end{align}
    We now use the Leibniz rule for the $\sigma$-adjoint action (see Lemma~\ref{lem:sigmaadleibniz}) to show that
    \begin{align}
        \ad_X^\sigma\,e = \ad^{ \sigma}_X\Big( \sum_i\, f_i\, e^\circ_i\Big) = \sum_i\, (\pounds_X f_i)\,e^\circ_i  \ ,
    \end{align}    
    and the result follows. This calculation can also be carried out locally.
\end{proof}

\begin{example}
Let $ E$ be an exact Courant algebroid over $M$ endowed with the divergence operator $\div_\mu$ associated to a nowhere-vanishing density $\mu$ on $M$ which is $D$-invariant, i.e. $\pounds_{ X} \mu =0$ for all $ X \in \frk{k}.$ Then $\div_\mu$ is compatible with $D$ in the sense of Definition \ref{def:K2compdiv}. 

For instance, let $g$ be a $D$-invariant Riemannian metric on $M$ with corresponding Hodge operator $\star_g$, and take $\mu = \mu_{g}= \star_g 1$. Then $\pounds_{ X}\circ \star_g=\star_g\circ\pounds_X,$ i.e. $\mu_g$ is $D$-invariant as well, and
\begin{align} \label{eqn:divmugD1}
\pounds_{ X}\, \div_{\mu_g}\, e^\circ = 0   
\end{align}
for any $e^\circ \in \sfgamma_{D}( E)$ and  $ X \in \frk{k}$.
\end{example}

\begin{lemma}\label{lem:edivinvariant}
    Let $E\to M$ be an exact Courant algebroid endowed with a $D$-invariant generalised metric, and any divergence operator $ \div$. Set $\div - \div_{\mu_g} =  \ip{e_{\text{div}},\,\cdot\,}$, where $\mu_g$ is the volume form on $M$ induced by the generalised metric. Then
    $\pounds_X\, \div \, e^\circ = 0,$ for all $e^\circ \in \sfgamma_{D}(E)$ and $X \in \frk{k}$, if and only if $e_{\text{div}} \in \sfgamma_{D}( E)$.
    If $\sfgamma_{D}( E)$ spans $E$ pointwise, then $\div$ is moreover compatible with $D$. 
\end{lemma}

\begin{proof}
For any $e^\circ \in \sfgamma_{D}( E)$ and $X \in \frk{k}$, using Equation \eqref{eqn:divmugD1} we compute
    \begin{align}
        \pounds_X\, \div \, e^\circ &= \pounds_X\, \div_{\mu_g}\, e^\circ + \pounds_{X} \ip{e_{\text{div}}, e^\circ} \\[4pt]
        &= \ip{\llbracket \sigma(X), e_{\text{div}}\rrbracket, e^\circ} + \ip{e_{\text{div}}, \llbracket \sigma(X), e^\circ \rrbracket}\\[4pt]
        &= \ip{\llbracket \sigma(X), e_{\text{div}}\rrbracket, e^\circ} + \ip{e_{\text{div}}, \rho^*\sigma^*\llbracket \sigma(X),  \sigma \,\rho(e^\circ)\, \rrbracket}\\[4pt]
        &= \ip{\llbracket \sigma(X), e_{\text{div}}\rrbracket, e^\circ} - \ip{\llbracket \sigma(X), \sigma\, \rho(e_{\text{div}}) \rrbracket, \sigma\, \rho(e^\circ)} + \pounds_X \ip{\sigma\, \rho( e_{\text{div}}), \sigma\, \rho(e^\circ)}\\[4pt]
        &= \ip{\llbracket \sigma(X), e_{\text{div}}\rrbracket - \rho^* \,\sigma^*\llbracket \sigma(X), \sigma\, \rho(e_{\text{div}})\rrbracket), e^\circ} \ ,
    \end{align}
where in the second equality we used $ \rho( \sigma( X))= X,$ in the third equality we used the $D$-invariance of $e^\circ$, and in the fourth equality we used item~\ref{eqn:metric1} of Definition \ref{def:CourantAlg}, together with the adjunction of $\sigma^*$ and $\rho^*$ with respect to the pairing. This last expression vanishes if and only if $e_{\text{div}} \in \sfgamma_{D}( E)$, establishing the first statement.

The second statement follows from Lemma \ref{lem:compatiblediverg}.
\end{proof}

\medskip

\subsection{T-Duality for Divergence Operators}~\\[5pt] \label{sect:Tdualitydiv}
 We shall now prove a  result similar to Theorem~\ref{thm:maingeneral} for divergence operators, under the condition
 \begin{align}
     K_1 \cap K_2 = \set{0} \ ,
 \end{align}
 which we assume throughout this subsection. 
For this, we first need a property of invariant sections in the T-duality framework, namely that $D_1$-invariant sections of $\red E{}_1$ lift to basic sections of $K_2^\perp$.

\begin{lemma}\label{lem:invsectionsarebasic}
    Under the assumptions of Theorem \ref{thm:maingeneral}, let $s \colon K_1^\perp / K_1 \to K_1^\perp$ be the unique\footnote{By \cite[Proposition 5.31]{DeFraja:2023fhe}, $s$ is unique because of the assumption that $K_1 \cap K_2 = \set{0}$.} splitting of the short exact sequence
    \begin{align}
  0  \longrightarrow K_1 \longrightarrow K_1^\perp \longrightarrow K_1^\perp/K_1 \longrightarrow 0 \ .
 \end{align}
    If $\red e^\circ \in \sfgamma_{D_1}(\red E{}_1)$, then $s(\red e^\circ)\in \sfgamma_{\text{bas}}(K_2^\perp)$, where $s(\red e^\circ)$ is the lift of $\red e^\circ$ to $\sfgamma_{\rm bas}(K_1^\perp)$ using the $C^\infty(\cQ_1)$-module isomorphism $\sfgamma(\red E{}_1) \cong \sfgammabas(K_1^\perp) / \sfgamma(K_1)$ as in \cite[Proposition 5.31]{DeFraja:2023fhe}.
\end{lemma}

\begin{proof}
    Denote for simplicity by $s \colon \sfgamma(\red E{}_1) \to \sfgammabas(\Im(s))$ the map given by the composition of the lift from $\sfgamma(\red E_1)$ to $\sfgamma_{\rm bas}(K_1^\perp)/ \sfgamma(K_1)$ and the duality splitting of \cite[Proposition 5.31]{DeFraja:2023fhe}, and note that $\natural_1 \circ s = \unit_{\red E{}_1}$. Since $\natural_1$ is a bracket homomorphism between $\cbrak{\,\cdot\, ,\, \cdot\,}$ and $\cbrak{\,\cdot\, ,\, \cdot\,}_1$, it then follows that
    \begin{align}
        \natural_1\big(\llbracket s(\red \sigma{}_1 (\red X)), s(\red e^\circ)\rrbracket - s \llbracket \red \sigma{}_1 (\red X), \red e^\circ\rrbracket_1\big) = 0 \ ,
    \end{align}
    for all $\red e^\circ \in \sfgamma_{D_1}(\red E{}_1)$ and $\red X \in \frk{k}_{\tau_1},$ hence
    \begin{align} \label{eqn:diffinK1}
        \llbracket s(\red \sigma{}_1 (\red X)), s(\red e^\circ)\rrbracket - s \llbracket \red \sigma{}_1 (\red X), \red e^\circ\rrbracket_1 \ \in \ \sfgamma(K_1) \ .
    \end{align}
    
    Let us show that 
    \begin{align} 
    s\big(\red \rho{}_1^*\, \red \sigma{}_1^*\llbracket \red \sigma{}_1(\red X), \red \sigma{}_1( \red \rho{}_1(\red e^\circ)) \rrbracket_1\big) = \rho^*\, \sigma_1^*\llbracket \sigma_1(X) , \sigma_1(\rho(s(\red e^\circ)))\rrbracket \ ,
    \end{align}
    where $X$ is the canonical lift of $\red X$ to $\frk{k}_2 \subset \sfgamma(TM)$ (see Definition \ref{def:compatiblesplit}). The difference is again in $\sfgamma (\ker(\natural_1) ) = \sfgamma(K_1)$. Under the anchor $\rho$ the difference lands in $T\cF_1 \cap T\cF_2 = \set{0}$. But $\sigma_1 \circ \rho$ is the identity on $K_1$ by~\cite[Lemma 2.34]{DeFraja:2023fhe}, hence the difference vanishes.
    
     We also have $s(\red\sigma{}_1(\red X)) = \sigma_1(X)$, for all $\red X \in \frk{k}_{\tau_1}$. Putting everything together we get
    \begin{align}
        \llbracket s(\red \sigma{}_1 (\red X)), s(\red e^\circ)\rrbracket - s \llbracket \red \sigma{}_1 (\red X), \red e^\circ\rrbracket_1 &= \llbracket s(\red \sigma{}_1 (\red X)), s(\red e^\circ)\rrbracket - s\big(\red \rho{}_1^*\, \red \sigma{}_1^*\llbracket \red \sigma{}_1(\red X), \red \sigma{}_1( \red \rho{}_1(\red e^\circ)) \rrbracket_1\big)\\[4pt]
        &= \llbracket \sigma_1 (X), s(\red e^\circ)\rrbracket - \rho^* \, \sigma_1^*\llbracket \sigma_1(X) , \sigma_1(\rho(s(\red e^\circ))\rrbracket \\[4pt]
        &=\llbracket \sigma_2(X), s(\red e^\circ) \rrbracket \ ,
    \end{align}
    where for the first equality we use the $D_1$-invariance of $\red e^\circ$, while for the last equality we use the compatibility of the splittings (see Definition \ref{def:compatiblesplit}).

    Now $\sigma_2(X) \in \sfgamma(K_2)$, and $s(\red e^\circ) \in \sfgamma(K_1^\perp \cap K_2^\perp)$. Since $\llbracket \sfgamma(K_2), \sfgamma(K_2^\perp) \rrbracket \subset \sfgamma(K_2^\perp)$, it follows from Equation \eqref{eqn:diffinK1} that 
    \begin{align}
    \llbracket \sigma_2(X), s(\red e^\circ) \rrbracket \ \in \ \sfgamma(K_2^\perp \cap K_1) = \sfgamma ( K_2 \cap K_1 ) = \set{0} \ .
    \end{align}
    Using the Leibniz rule for the Dorfman bracket, this extends to all elements of $\sfgamma(K_2)$, and hence $s(\red e^\circ)$ is basic with respect to~$K_2$.
\end{proof}

We have just shown how $D_1$-invariant sections of $\red E{}_1$ can be lifted to $K_2$-basic sections. This property is crucial for the behaviour of divergence operators under T-duality.

\begin{theorem}\label{thm:tdualdivergence}
Under the assumptions of Theorem \ref{thm:maingeneral}, suppose that $D_1 \subset T\cQ_1$ is  spanned pointwise by a (local) Lie subalgebra $\frk{k}_{\tau_1} \subset \sfgamma(D_1).$ Assume further that $K_1 \cap K_2 = \set{0}$ and that $\sfgamma_{D_1}(\red E{}_1)$ spans $\red E{}_1$ pointwise.  If $\red E{}_1$ admits a divergence $\div{}_1$ which is (locally) compatible with $D_1,$ then there exists a unique divergence operator $\div{}_2$ on $\red E{}_2$ such that $\div{}_1 \sim_R \div{}_2$.
\end{theorem}

\begin{proof}
Let $\div{}_1 $ be a divergence on $\red E{}_1$. By \cite[Proposition 5.31]{DeFraja:2023fhe}, there is a unique splitting $s_1 \colon K_1^\perp/ K_1 \to K_1^\perp$ of the short exact sequence
\begin{align}\label{eqn:dualitysplitting}
    0 \longrightarrow K_1 \longrightarrow K_1^\perp \longrightarrow K_1^\perp / K_1 \longrightarrow 0 
\end{align}
such that $\Im(s_1)\subset K_1^\perp \cap K_2^\perp$. Since $K_1\cap K_2 = \set{0}$ and $\rk(K_1)=\rk(K_2)$, it follows that $\Im(s_1)= K_1^\perp \cap K_2^\perp$.

Define the map $\delta \colon \sfgamma\big(\Im (s_1)\big) \to \varpi_1^*\big(C^\infty (\cQ_1)\big)$ by
\begin{align}
    \delta(s_1(e)) = \div{}_1 \, \red e \ ,
\end{align}
where $e$ is the element corresponding to $\red e\in\sfgamma(\red E{}_1)$ using the identifications $\sfgamma(\red E{}_1) \cong \sfgamma_{\text{bas}}(K_1^\perp) / \sfgamma(K_1)$ and $C^\infty_{\text{bas}}(M,\cF_1) \cong C^\infty(\cQ_1)$.  Note that $\delta$ behaves like a divergence operator, i.e. it satisfies a Leibniz-like rule, because $\varpi_{1*}(\rho(s_1(e)))= \red \rho{}_1 (\red e).$ 
Using Definition \ref{def:K2compdiv},  the divergence $\delta$ satisfies $\pounds_X\, \delta (s_1(e)) = 0 $, for any $X \in \frk{k}_2$. One extends this to the whole of $\sfgamma\big( \Im(s_1) \big)$ by the Leibniz rule.

Any section $e \in \sfgamma \big( \Im(s_1) \big) \cap \sfgammabas(K_2^\perp) $ can be written as $e = \sum_i\, f_i\, s_1(\red e{}_i)$, where $\red e{}_i \in \sfgamma(\red E{}_1)$ and $f_i \in \varpi_2^* \big(C^\infty(\cQ_2)\big)$. 
Likewise any vector field $X \in \sfgamma \big( \rho(K_2) \big)$ can be written as $X =\sum_j\, g_j\, X_j$, where $X_j \in \frk{k}_2$ and $g_j \in C^\infty (M)$. Thus
\begin{align}
   \pounds_X\, \delta (e) &= \sum_{i,j}\, \Big( g_j\,\pounds_{X_j} \big( f_i\, \, \delta ( s_1(\red e{}_i)) \big)+ g_j\, \pounds_{X_j} \, \pounds_{ \rho(s_1(\red e{}_i))} f_i \Big)\\[4pt]
    &= \sum_{i,j}\, \big( g_j\, (\pounds_{X_j} f_i)\, \delta (s_1(\red e{}_i)) + g_j\, f_i\, \pounds_{X_j}\, \delta (s_1(\red e{}_i)) \\[-10pt]
    & \hspace{5cm} + g_j\,\pounds_{\rho(s_1(\red e{}_i))}\, \pounds_{X_j} f_i + g_j\,\pounds_{[X_j,\rho{}(s_1(\red e{}_i))]} f_i \big) = 0 \ .
\end{align}
The final term vanishes since $[\varpi_{1*} (X_j), \red \rho{}_1(\red e{}_i)] = 0$, hence $[X_j,\rho (s_1 (\red e{}_i))] \in \sfgamma \big(\rho(K_1) \big)$, whereas
\begin{align}
[X_j,\rho(s_1(\red e{}_i))] = \rho(\llbracket\sigma_2(X_j), s_1(\red e{}_i)\rrbracket) \ \in \ \sfgamma \big( \rho(K_2^\perp \cap K_1) \big)\subset \sfgamma \big( \rho(K_2) \big) \ .
\end{align}
Hence the restriction of $\delta$ to $\sfgamma \big(\Im(s_1) \big) \cap \sfgammabas(K_2^\perp)$ maps to $\varpi_2^* \big(C^\infty(\cQ_2)\big)$.

By \cite[Proposition 5.31]{DeFraja:2023fhe} there also exists a unique splitting $s_2:K_2^\perp/K_2\to K_2^\perp$ of the short exact sequence 
\begin{align}
    0 \longrightarrow K_2 \longrightarrow K_2^\perp \longrightarrow K_2^\perp / K_2 \longrightarrow 0
\end{align}
induced by $K_2^\perp$ and $K_2$. As for $s_1$, its image is $\Im(s_2)= K_1^\perp \cap K_2^\perp$.
Thus we may define a divergence operator $\div{}_2 \colon \sfgamma(\red E{}_2) \to C^\infty(\cQ_2)$ by
\begin{align}
    \div{}_2 \, \red e{}_2 = \delta \big(e_2 \big) \ ,
\end{align}
where $e_2$ is the lift of $\red e{}_2\in\sfgamma(\red E{}_2)$ to $\sfgamma_{\text{bas}}(K^\perp_2)$ using the unique splitting $s_2$ and the identifications $\sfgamma(\red E{}_2) \cong \sfgamma_{\text{bas}}(K_2^\perp)/\sfgamma(K_2)$ and $C^\infty_{\text{bas}}(M,\cF_2) \cong C^\infty(\cQ_2)$.

From Lemma~\ref{lem:invsectionsarebasic} it follows that any $K_2$-basic section $e$ of $\Im(s_1)$ is written as $e =\sum_i\, f_i\, s_1(\red e{}_i)$, where $\red e{}_i \in \sfgamma_{D_1}(\red E{}_1)$ and $f_i\in \varpi_2^* \big(C^\infty(\cQ_2)\big)$. 
Hence $e$ is also  basic with respect to $K_1$ when $f_i \in \varpi^*_1 \big(C^\infty(\cQ_1)\big)\,\cap\, \varpi_2^* \big(C^\infty(\cQ_2)\big)$. Thus $R$-related sections $ \psi{}_1 \sim_R  \psi{}_2$ are given by $ \psi{}_1 =\sum_i\, f_i\, \red e{}_i$ and $\psi{}_2 =\sum_i \, \varpi_{2*} \big( \varpi_1^*(f_i) \big)\,\natural_2 (s_1(\red e{}_i))$, for $f_i \in C^\infty_{D_1}(\cQ_1)$ and $\red e{}_i \in \sfgamma_{D_1}(\red E{}_1)$. By the above construction, $\div{}_1\, \psi{}_1 = \div{}_2\, \psi{}_2$.

Uniqueness follows from Lemma \ref{lem:uniquereldiv}.

The local version follows the same construction, by noting that the invariance of $\delta(e)$ can be checked locally.
\end{proof}

\begin{remark}[\textbf{Local T-duality for Divergences}] \label{rmk:tdualdivlocal}
   These results also hold in the local case, that is, if $V_1^+$ is locally $D_1$-invariant in the sense of Definition \ref{def:localinvmetric}. This is because, in the proof of Lemma \ref{lem:invsectionsarebasic}, checking whether $s(\red e^\circ)$ is basic with respect to $K_2$ for $\red e^\circ \in \sfgamma_{D_1}(\red E{}_1)$ can be performed on locally defined sections of $K_2$. Similarly, the invariance of $\delta(e)$ with respect to $T\cF_2=\rho(K_2)$ can be checked locally. It follows that Theorem \ref{thm:tdualdivergence} also holds in the local case.
\end{remark}

\medskip

\subsubsection{Correspondence Spaces}\label{sec:correspondencediv}~\\[5pt]
 We shall discuss how the results of this subsection apply to the most prominent example of T-duality, that of two fiber bundles $\cQ_1$ and $\cQ_2$ over the same base manifold $\cB$. This setup can be depicted in the commutative correspondence diagram
\begin{equation} \label{cd:correspondencespace}
        \begin{tikzcd}
             & M \arrow[swap]{dl}{\varpi_1} \arrow{dr}{\varpi_2} & \\
            \cQ_1 \arrow[swap]{dr}{\pi_1} &  & \cQ_2 \arrow{dl}{\pi_2} \\
             & \cB & 
        \end{tikzcd}
    \end{equation}
    
Because of the universal property of fibred products, we can take the correspondence space to be $$M = \cQ_1 \times_{\cB} \cQ_2 \ . $$ In particular, $T\cF_1 \cap T\cF_2 = \set{0}$, with $\varpi_1$ and $\varpi_2$ the canonical projection maps. Since $\rho$ is injective on $K_1$ and $K_2$ (see~\cite[Lemma~2.34]{DeFraja:2023fhe}), it follows that $K_1 \cap K_2 = \set{0}$.
The particular case of principal torus bundles is covered extensively in \cite{Bouwknegt2003topology, cavalcanti2011generalized, Garcia-Fernandez:2020ope} and will be treated as a special case of our general formalism in Subsection~\ref{sec:torusbundles} below.

\begin{proposition}\label{prop:tau2isD2invariant}
Suppose that the assumptions of Theorem \ref{thm:maingeneral} are satisfied, with T-duality relation $R\colon (\red E{}_1, \tau_1) \rel (\red E{}_2, \tau_2)$. Assume further that there is a spanning set of vector fields $\set{X_1,\dots,X_{r_2}}\subset T\cF_1$ which are projectable by $\cF_2$ such that
\begin{align} \label{eqn:tdualcond}
    \llbracket \sigma_1(X_i), s(\red e^\circ) \rrbracket = \ad^{\sigma_2}_{X_i}\, s(\red e^\circ)
\end{align}
for every $i=1,\dots,r_2$ and $\red e^\circ \in \sfgamma_{D_1}(\red E{}_1)$, where $s:K_1^\perp/K_1\to K_1^\perp$ is the unique splitting from Lemma~\ref{lem:invsectionsarebasic}.
Then the T-dual generalised metric $\tau_2$ is invariant with respect to the distribution $$D_2:=T\red\cF{}_2 = \varpi_{2*}(T\cF_1) \ , $$ with Lie algebra $\frk{k}_{\tau_2} := \text{Span}_{\IR}\set{\varpi_{2*}X_1,\dots,\varpi_{2*}X_{r_2}}\subset\sfgamma(T\cQ_2)$. Furthermore, there is an isomorphism $\sfgamma_{D_1}(\red E{}_1) \cong \sfgamma_{D_2}(\red E{}_2)$ of $C^\infty(\cB)$-modules.
\end{proposition}

\begin{proof}
Let $(\IT M,H_1)$ and $(\IT M,H_2)$ be the twisted standard Courant algebroids which are isomorphic to $E$ via the isomorphisms induced by the compatible adapted splittings $\sigma_1$ and $\sigma_2$, respectively. Define $B\in\sfOmega^2(M)$ by $$\sigma_1(X)-\sigma_2(X)=\rho^*(\iota_XB)$$ for $X\in\sfgamma(TM)$. For every \smash{$\red Y^\circ + \red \eta^\circ \in\sfgamma_{D_1}(\IT \cQ_1)$}, let $\eta = \varpi_1^*\, \red \eta^\circ$ and $Y + \eta = s(\red Y^\circ + \red \eta^\circ)$. In other words, $Y$ is the unique lift of $\red Y^\circ$ such that
    \begin{align}
        \eta(Z) = B(Y,Z) \ ,
    \end{align}
    for all  $Z \in \sfgamma(T\cF_2)$.

Equation \eqref{eqn:tdualcond} then reads
    \begin{align}
        \iota_{Y}\, \iota_{X_i}\, \de B + \iota_{Y} \,\pounds_{X_i} B=\llbracket - \iota_{X_i}B, Y + \eta \rrbracket &= \mathrm{pr}_2(\llbracket \sigma_2(X_i), \sigma_2(Y) \rrbracket_{H_1})\\[4pt]
        &= \mathrm{pr}_2(\llbracket \sigma_1(X_i), \sigma_1(Y) \rrbracket_{H_2}) = \iota_{Y}\,\iota_{X_i}\, \de B \ ,
    \end{align}
    hence $\iota_{Y}\pounds_{X_i} B = 0$ for every $i=1,\dots,r_2$ and $\red Y^\circ + \red \eta^\circ \in\sfgamma_{D_1}(\IT \cQ_1)$.
    Thus
    \begin{align}
        0=\pounds_{X_i} \big(\eta(Z) - B(Y,Z)\big) &=  \pounds_{X_i}\big(\eta(Z)\big) - (\pounds_{X_i}B)(Y,Z) - B([X_i,Y], Z)- B(Y,[X_i,Z])\\[4pt]
        &=(\pounds_{X_i}\eta)(Z) + \eta([X_i,Z]) - B([X_i,Y], Z)- \eta([X_i,Z])\\[4pt]
        & = -B([X_i,Y], Z) \ ,
    \end{align}
    for all $Z\in\sfgamma(T\cF_2)$. 
    But $[X_i,Y] \in \sfgamma(T\cF_1)$ since $Y$ is $\cF_1$-projectable, and $B$ is non-degenerate on $T\cF_1 \otimes T\cF_2$ since $K_1\cap K_2=\set{0}$, hence $[X_i, Y] = 0$. 
    
    It therefore follows that
    \begin{align}
        \ad^{\red \sigma{}_2}_{\varpi_{2*} X_i}\, \natural_2 \big( s(\red Y^\circ+ \red \eta^\circ) \big)= 0 \ .
    \end{align}
    By virtue of the generalised isometry $R$ and the uniqueness of the lift, it thus follows that the map $\natural_2 \circ s \colon \sfgamma_{D_1}(\IT \cQ_1) \to \sfgamma_{D_2}(\IT \cQ_2)$ is an isomorphism. Hence we obtain the isomorphism of $C^\infty(\cB)$-modules
    $\sfgamma_{D_1}(\red E{}_1) \cong \sfgamma_{D_2}(\red E{}_2)$.
    By taking $\red Y^\circ + \red\eta^\circ \in \sfgamma(V_1^+)$, it then follows that $V_2^+$ is $D_2$-invariant. 
\end{proof}

\begin{proposition} \label{prop:transitiveTduality}
Under the conditions of Proposition \ref{prop:tau2isD2invariant}, there exists transitive Courant algebroids $\red {\red E}{}_1$ and  $\red {\red E}{}_2$ over $\cB$ together with a $C^\infty(\cB)$-module isomorphism $\mathscr{F} \colon \sfgamma(\red {\red E}{}_1) \to \sfgamma(\red {\red E}{}_2)$ induced by the T-duality relation $R$.
\end{proposition}

\begin{proof}
Define a vector bundle $\red{\red E}{}_1$ over $\cB$ by
\begin{align}
    \red{\red E}{}_1 = \red E{}_1 / \red \cF{}_1 \ .
\end{align}
The space of sections of $\red{\red E}{}_1$ is identified with the space $\sfgamma_{D_1}(\red E{}_1)$ as a module over $C^\infty_{D_1}(\cQ_1) \cong C^\infty(\cB)$. This gives $\red{\red E}{}_1$ the structure of a Courant algebroid as follows.
The involutivity of $\sfgamma_{D_1}(\red E{}_1)$ established in Lemma \ref{cor:invariantsinvolutive} gives the bracket structure on $\red{\red E}{}_1$, while the pairing descends to $\cB$ by Lemma~\ref{lem:pairingpreserved}. The anchor \smash{$\red{\red \rho}{}_1:\red{\red E}{}_1\to T\cB$} is defined by \smash{$\red{\red \rho}{}_1(\red{\red e}) = \pi_{1*}\,\red \rho{}_1(\red e^\circ)$}, where $\red\rho{}_1(\red e^\circ)$ is projectable to $\cB$ since $\red e^\circ \in \sfgamma_{D_1}(\red E{}_1)$. Since $\sfgamma_{D_1}(\red E{}_1)$ spans $\red E{}_1$ pointwise, it follows that \smash{$\red{\red \rho}{}_1$} is surjective, and hence that \smash{$\red{\red E}{}_1$} is a transitive Courant algebroid. 

By a completely analogous argument using the Hausdorff-Morita equivalent foliation $\red\cF{}_2$ of $\cQ_2$, there is also an isomorphism \smash{$\sfgamma(\red{\red E}{}_2) \cong \sfgamma_{D_2}(\red E{}_2)$}, where
\begin{align}
    \red{\red E}{}_2 = \red E{}_2 / \red \cF{}_2
\end{align}
is a transitive Courant algebroid over $\cB$.
By Proposition \ref{prop:tau2isD2invariant}, there is an isomorphism $\sfgamma_{D_1}(\red E{}_1) \cong \sfgamma_{D_2}(\red E{}_2)$. The chain of isomorphisms
\begin{align}
    \sfgamma(\red{\red E}{}_1) \cong \sfgamma_{D_1}(\red E{}_1) \cong \sfgamma_{D_2}(\red E{}_2) \cong \sfgamma(\red{\red E}{}_2)
\end{align}
defines the isomorphism $\mathscr{F} \colon \sfgamma(\red{\red E}{}_1) \to \sfgamma(\red{\red E}{}_2)$.
\end{proof}

We can now show that T-dual divergences induce isomorphic divergences, under $\mathscr{F}$,  on the transitive Courant algebroids $\red{\red E}{}_1$ and $\red{\red E}{}_2$.

\begin{lemma} \label{lem:isomTdualdiv}
Under the assumptions of Theorem \ref{thm:tdualdivergence}, let \smash{$\red{\red \div}{}_1$} be a divergence on \smash{$\red{\red E}{}_1$}. Then there exists a divergence operator $\red{\red \div}{}_2$ on $\red{\red E}{}_2$ such that
\begin{align} \label{eqn:divergencesbase}
    \red{\red \div}{}_1 \, \red{\red e} = \red{\red \div}{}_2 \, \mathscr{F}(\red{\red e}) \ , 
\end{align}
for all $\red{\red e} \in \sfgamma(\red{\red E}{}_1)$.
\end{lemma}

\begin{proof}
  The pullback $\pi_1^*\,\red{\red \div}{}_1$ defines a $D_1$-compatible divergence on $\red E{}_1$. By Theorem \ref{thm:tdualdivergence}, there is a T-dual divergence $\div{}{}_2$ on $\red E{}_2$.  By Lemma \ref{lem:compatiblediverg}, it follows that $\div{}{}_2$ is $D_2$-compatible, and hence defines a divergence $\red{\red \div}{}_2 \coloneqq \pi_{2*}\,\div{}{}_2$ on $\red{\red E}{}_2$ which by construction satisfies Equation~\eqref{eqn:divergencesbase}.
\end{proof}

\medskip

\subsubsection{Principal Torus Bundles}\label{sec:torusbundles}~\\[5pt]
We may specialise to the case of principal torus bundles, where $\red\cF{}_1$ and $\red\cF{}_2$ are given by the orbits of principal actions of the torus group $\sfT^k$ of rank $k$ on $\cQ_1$ and $\cQ_2$, respectively. We recall the notion of T-duality from~\cite[Definition~10.13]{Garcia-Fernandez:2020ope}.

\begin{definition}\label{def:garcia-fernandezstreetsTdual}
    Let $\red E{}_1 \to \cQ_1$ and $\red E{}_2 \to \cQ_2$ be $\sfT^k$-equivariant exact Courant algebroids, where $\cQ_1$ and $\cQ_2$ are principal $\sfT^k$-bundles over the same base $\cB$. Suppose that they are T-dual in the sense of Cavalcanti-Gualitieri \cite{cavalcanti2011generalized}; then there is an isomorphism $\mathscr{R} \colon \red E{}_1/\sfT^k \to \red E{}_2 / \sfT^k$ of transitive Courant algebroids over $\cB$. A $\sfT^k$-invariant pair $(\tau_2, \div{}_2 )$ on $\red E{}_2$ is \emph{T-dual} via $\mathscr{R}$ to a $\sfT^k$-invariant pair $(\tau_1, \div{}_1)$ on $\red E{}_1$ if 
    \begin{align} 
        \mathscr{R} \circ \pi_1{}_* \tau_1 = \pi_2{}_* \tau_2 \circ \mathscr{R} \quad , \quad \pi_1{}_*\, \div{}_1 = \mathscr{R}^*\,  \pi_2{}_*\, \div{}_2 \ , \label{eqn:divGarcia} 
    \end{align}
    where $\pi_i \colon \cQ_i \to \cB$ with $i=1,  2$ are the bundle projections.
\end{definition}

Here the notion of $\sfT^k$-invariance coincides with our sense: they are equivalent with the identification of $\frk{k}_{\tau_1} = \frk{k}_{\tau_2} = \frt^k$ with the Lie algebra of the torus group $\sfT^k$.

\begin{proposition} \label{prop:garciastreetstdual}
    If $(\red E{}_2, \tau_2, \div{}_2)$ is T-dual  to $(\red E{}_1, \tau_1, \div{}_1)$ in the sense of Definition \ref{def:garcia-fernandezstreetsTdual}, then there is a generalised isometry $R \colon (\red E{}_1, \tau_1) \rel (\red E{}_2, \tau_2)$ such that $\div{}_1 \sim_R \div{}_2$. 
\end{proposition}

\begin{proof}
    The generalised isometry $R$ is constructed in \cite[Lemma 6.5]{DeFraja:2023fhe}. The result then follows immediately from Proposition \ref{prop:transitiveTduality} and Lemma \ref{lem:isomTdualdiv} with $\mathscr{R} = \mathscr{F}$.
\end{proof}

\begin{example}[\textbf{Circle Bundles}]\label{eg:circledilatonshift}
    Let $\cQ_i$ be circle bundles over a common base $\cB$, for $i=1,2$, and let $M = \cQ_1 \times_{\cB} \cQ_2$. Let $\theta_i$ be connections on the $\sfS^1$-bundles $\pi_i \colon \cQ_i\to \cB$, and $\partial_{i}$ the global vector fields defining the circle fibres. Let $g_1$ be a metric on $\cQ_1$ of the form $$g_1 = g \, \theta_1 \otimes \theta_1 + \pi_1^* \, g_{\cB} \ , $$ where $g \in C^\infty(\cB)$ is an everywhere-positive function, and $g_\cB$ is a Riemannian metric on $\cB$. Then $g_1$ is $\partial_{1}$-invariant (i.e. $\sfS^1$-invariant), and hence we may construct the metric $$g_2 = g^{-1} \, \theta_2 \otimes \theta_2 + \pi_2^* \, g_\cB$$ on $\cQ_2$ such that the T-duality relation $R \colon (\IT \cQ_1,0, g_1) \rel (\IT \cQ_2,0, g_2)$ is a generalised isometry. 
    
    The corresponding densities are
    \begin{align}
        \mu_{g_1} = \sqrt{g} \, |\theta_1 \wedge \mu_{g_\cB}| \quad, \quad \mu_{g_2} = \sqrt{g^{-1}} \, |\theta_2 \wedge \mu_{g_\cB}|
    \end{align}
    and one calculates $$\div{}_{\mu_{g_1}}\, e_1 = \sqrt{g^{-1}} \, \pounds_{\rho_1(e_1)} \sqrt{g} \quad , \quad \div{}_{\mu_{g_2}}\, e_2 = \sqrt{g} \, \pounds_{\rho_2(e_2)}\sqrt{g^{-1}} \ , $$ for all $e_i \in \sfgamma(\IT \cQ_i)$. If $X\in \sfgamma(T\cB)$, then\footnote{With a slight abuse of notation, here we consider the pullback of $X$ to $\cQ_i$ by $\pi_i$.} $X \sim_ R X$, but $\div{}_{\mu_{g_1}}\,X \neq \div{}_{\mu_{g_2}}\,X$ in general, and hence the divergence operators are not $R$-related. By Equation \eqref{eqn:dilatonshift}, this necessitates a dilaton shift.

    Indeed, by imposing Equation \eqref{eqn:dilatonshift} for $X + \xi \in \sfgamma(\IT\cB)$, we find
    \begin{align}
        \ip{Y_{\tt dil} + \eta_{\tt dil}, X + \xi  + \partial_{2}+ \theta_2}_2 = \sqrt{g^{-1}} \,  X(\sqrt{g}) - \sqrt{g} \, X(\sqrt{g^{-1}}\,)  = X(\log g) \ .
    \end{align}
    This gives $$Y_{\tt dil} = 0 \quad , \quad \eta_{\tt dil} = \de \log g \ , $$ which is the usual dilaton shift. A similar argument can be used for the dilaton shift of any T-dual pair of torus bundles.
\end{example}

\medskip

\section{Generalised Ricci Flow and Courant Algebroid Relations} \label{sect:genricci}

In this section we discuss the problem of whether the generalised Ricci flow is compatible with geometric T-duality of Definition~\ref{def:geomTdual}, as well as  Theorems~\ref{thm:maingeneral} and~\ref{thm:tdualdivergence}. In other words, provided with the short-time existence of solutions of generalised Ricci flow  for  
$(\red E{}_1, V_1^+,\div_1),$ as developed in \cite{Streets:2024rfo}, we investigate whether the T-dual generalised metric $V_2^+$ of Theorem~\ref{thm:maingeneral} and divergence $\div_2$ of Theorem~\ref{thm:tdualdivergence} on $\red E{}_2$ at a given time is a solution of generalised Ricci flow as well. Our starting point will be the introduction of a suitable notion of relation for generalised Ricci tensors, which enables a proof of the compatibility of T-duality with generalised string background equations.

\medskip

\subsection{Generalised Ricci Tensor and Scalar Curvature}~\\[5pt]
 Let us introduce the main object of our study in this section, the \emph{generalised Ricci tensor}, following the definition from \cite{Severa:2018pag}.
 
\begin{definition}
    Let $E$ be a Courant algebroid over $M$ endowed with a generalised metric $V^+$ and a divergence operator $\div$. The \emph{generalised Ricci tensor} of the pair $(V^+, \div)$ is the map
    \begin{align}
    \gric_{V^+, \div} \colon \sfgamma(V^+) \times \sfgamma(V^-) \longrightarrow C^\infty(M)
\end{align}
given by
\begin{align}
    \gric_{V^+, \div}(v^+, w^-) \coloneqq \div\, \cbrak{v^+, w^-}^+ - \pounds_{\rho(w^-)}\, \div \, v^+ - \Tr_{V^+}(\cbrak{\cbrak{\, \cdot \, , w^-}^-, v^+}^+) \ ,
\end{align}
for $v^+\in\sfgamma(V^+)$ and $w^-\in\sfgamma(V^-)$, 
where $\cbrak{\,\cdot \, , \,  \cdot\,}^\pm$ denotes the projection of the Dorfman bracket to the subbundle $V^\pm$. We will equivalently use the notation $\gric_{\tau, \div},$ where $\tau$ is the automorphism of $E$ defined by the generalised metric $V^+$. 

By interchanging the roles of $V^+$ and $V^-$, we similarly define $\gric_{V^-, \div}$.
\end{definition}

It is shown in \cite[Proposition 3.2]{Severa:2018pag} that $\gric_{V^+, \div} \in \sfgamma(V^{+*} \otimes V^{-*})$, i.e. $\gric_{V^+, \div}$ is a covariant tensor on the subbundles $V^+$ and $ V^-$. There are many definitions of generalised Ricci tensor available in the literature, however it has been shown in \cite{Cavalcanti:2024uky} that all of them are equivalent. We choose the present formulation in order to discuss how generalised Ricci tensors relate via generalised isometries which are also relations between divergence operators.

\begin{remark}\label{rem:exactRicci}
    When $E$ is an exact Courant algebroid, the generalised Ricci tensor is related to the classical Ricci tensor $\ric_g$ of the Riemannian metric $g$ induced by $V^+$, as well as the Kalb-Ramond flux $H\in\sfOmega^3_{\rm cl}(M)$ and the divergence $\div_{\mu_g}$: If $\div -\div_{\mu_g} = \ip{e,\,\cdot\,}$ with $e=e^++e^-\in\sfgamma(V^+\oplus V^-)$, we denote $\Fi^\pm = g\big(\rho(e^\pm), \,\cdot\,\big)$. Then for $X,Y\in\sfgamma(TM)$, it follows from \cite[Proposition 3.30]{Garcia-Fernandez:2020ope} that
\begin{align}\label{eqn:gricclassically}
    \gric_{V^\pm,\div}(X^\pm, Y^\mp) = \big(\ric_g - \tfrac 14\, H^2 \mp \tfrac 12\, \de^* H \pm \nabla^\pm \Fi^\pm \big)(X,Y) \ ,
\end{align}
where $X^\pm = X \pm g(X,\,\cdot\,)$, $H^2(X,Y) = g(\iota_X H, \iota_Y H)$, $\de^*$ is the codifferential corresponding to the exterior differential $\de$, and $$\nabla^\pm = \nabla^g\pm\tfrac12\,g^{-1}\,H$$ are the Bismut connections\footnote{That is, $\nabla^\pm$ is the unique metric compatible connection with skew-symmetric torsion $\pm H$.} with $\nabla^g$ the Levi-Civita connection of $g$.
\end{remark}

Following \cite{Streets:2024rfo} we introduce the \emph{full generalised Ricci tensor} $\overline{\mathrm{GRic}}_{V^+,\div}$ defined by
\begin{align}
    \overline{\gric}_{V^+,\div} \coloneqq \mathrm{GRic}_{V^+,\div} - \mathrm{GRic}_{V^-, \div} \ . 
\end{align}
In~\cite[Definition 3.52]{Garcia-Fernandez:2020ope} a pair $(V^+, \div)$ on a Courant algebroid $E$ is called a \emph{generalised Einstein pair} if $\overline{\gric}_{V^+,\div}=0$.

The full generalised Ricci tensor also allows for a definition of the generalised scalar curvature in this setting \cite[Definition 4.9]{Streets:2024rfo}.

\begin{definition}
    The \emph{generalised scalar curvature} of the pair $(V^+, \div)$ is the function $\mathrm{GR}_{V^+, \div}\in C^\infty(M)$ given by
    \begin{align}
        \mathrm{GR}_{V^+, \div} \coloneqq \Tr \big( \tau \ \overline{\gric}_{V^+, \div} \big) \ ,
    \end{align}
    where $\tau$ is the automorphism of $E$ associated with the generalised metric $V^+$.
\end{definition}

A relevant property of the generalised scalar curvature, proven in \cite{Severa:2018pag,Streets:2024rfo}, is that it can be written as
    \begin{align}
    \mathrm{GR}_{V^+, \div} = \mu^{-1/2}\, \Delta_{V^+}\mu^{1/2} + \ip{e_{\div}, e_{\div}} + 2\, \div_{\mu}\, e_{\div} \ ,
\end{align}
for some everywhere-positive density $\mu$ on $M$, where $e_{\div}\in\sfgamma(E)$ is given by $\ip{e_{\div},\,\cdot\,} = \div - \div_\mu$, while $\Delta_{V^+}$ is the Laplace operator defined by the generalised metric $V^+$ and a Levi-Civita Courant algebroid connection $\nabla^E$ on $E$ for $V^+$. Note that $$\Delta_{V^+} f = \div_{\nabla^E} \big( \tau (\cD f) \big) \ ,$$ for all $f \in C^\infty(M)$, and that the generalised scalar curvature does not depend on the choice of density $\mu$.

\begin{remark} \label{rmk:genstringbackground}
    In the setting of Remark \ref{rem:exactRicci}, on an exact Courant algebroid $E$ with generalised metric $V^+$ and associated Riemannian metric $g$, as well as with the dilaton field $ \phi\in C^\infty(M)$, the generalised scalar curvature is given by
    \begin{align}
        \mathrm{GR}_{V^+, \div_{\e^{-2\phi}\, \mu_g}} = \mathrm{R}_g - \tfrac{1}{12}\, \|H\|^2_g  - 4\e^{\phi}\,\Delta_g \e^{-\phi} \ .
    \end{align}
    Here $\mathrm{R}_g$ is the classical scalar curvature and $\Delta_g$ the Laplace operator associated with $g$, while the Hodge norm of the Kalb-Ramond flux $H\in\sfOmega_{\rm cl}^3(M)$ is given by $$\|H\|^2_g \ \mu_g = H\wedge\star_g\,H \ . $$ When $M$ is compact, the corresponding Einstein-Hilbert functional \eqref{eq:EHaction} recovers the familiar action functional for $\cN=0$ supergravity in the string frame. 
\end{remark}

Remark \ref{rmk:genstringbackground} together with the discussion below Equation~\eqref{eq:EHaction} motivates

\begin{definition} \label{def:genstringbackground}
The \emph{generalised string background equations} for a pair $(V^+,\div)$ on $E$ are given by $$\gric_{V^+,\div} = 0 \quad , \quad {\rm GR}_{V^+, \div} = 0 \ . $$
\end{definition}

\medskip

\subsection{Related Generalised Ricci Tensors}~\\[5pt]
 In order to discuss Courant algebroid relations between generalised Ricci tensors, we weaken the notion of \cite[Definition 3.7]{Vysoky2020hitchiker} by considering tensors defined on subbundles of related Courant algebroids.
 
\begin{definition} \label{def:relatedtensors}
    Let $R \colon E_1\rel E_2$ be a Courant algebroid relation supported on $C\subseteq M_1\times M_2$. Let \smash{$\cT_i \in \sfgamma\big( V_i^{(1)*} \otimes \cdots \otimes V_i^{(k)*}\big)$}, for $i=1,2$, be covariant $k$-tensors defined on subbundles \smash{$V_i^{(j)} \subset E_i$} for $j = 1,\dots,k$. Then $\cT_1$ and $\cT_2$ are \emph{$R$-related}, denoted $\cT_1 \sim_R \cT_2$, if for every $(c_1,c_2) \in C$ and every $k$-tuple \smash{$\big(e_1^{(j)}, e_2^{(j)}\big) \in R_{(c_1,c_2)} \cap\big(V_1^{(j)} \times V_2^{(j)}\big)_{(c_1,c_2)}$} they satisfy
    \begin{align}
        \cT_1 \big(e_1^{(1)},\dots,e_1^{(k)}\big)(c_1) = \cT_2 \big(e_2^{(1)},\dots,e_2^{(k)}\big)(c_2) \ .
    \end{align}
\end{definition}

With this notion we can now establish when generalised Ricci tensors are related by a generalised isometry.

\begin{lemma}\label{lem:riccomponents}
Let $R \colon (E_1,V_1^+) \dashrightarrow (E_2,V_2^+) $ be a generalised isometry supported on $C \subseteq M_1 \times M_2$ which is a Dirac structure in $E_1 \times \overline{E}_2$. Let $\div_1$ and $\div_2$ be divergence operators on $E_1$ and $E_2$, respectively, such that $\div_1 \sim_R \div_2$. Then
\begin{enumerate}[label=(\roman{enumi})]
    \item\label{item:Gricterm1} $\div_1\, \llbracket v_1^+, v_1^- \rrbracket_1^+ \sim_C \div_2\, \llbracket v_2^+, v_2^- \rrbracket_2^+$ \ , \\[-3mm]
    \item\label{item:Gricterm2} $\pounds_{\rho_1(v_1^-)}\, \div_1\, v_1^+ \sim_C \pounds_{\rho_2(v_2^-)}\, \div_2\, v_2^+$ \ , \\[-3mm]
    \item\label{item:Gricterm3} $\Tr_{V_1^+}( \llbracket \llbracket \, \cdot \,, v_1^- \rrbracket_1^-, v_1^+ \rrbracket_1^+ ) \sim_C \Tr_{V_2^+}( \llbracket \llbracket \, \cdot \,, v_2^- \rrbracket_2^-, v_2^+ \rrbracket_2^+ )$ \ ,
\end{enumerate}
for any $(e_1, e_2) \in \sfgamma(E_1 \times \overline E_2; R)$ which under Definition \ref{defn:generalisedisometry2} decomposes in $\sfgamma(R)$ as
$$(e_1,e_2) \big\rvert_C = (v_1^+, v_2^+) + (v_1^-, v_2^-) \ , $$ with $(v_1^\pm, v_2^\pm) \in \sfgamma(R^\pm)$.
\end{lemma}

\begin{proof}
\ref{item:Gricterm1} follows from the fact that if $e_1 \sim_R e_2$, then the corresponding projections satisfy $e_1^\pm \sim_R e_2^\pm$ by Equation~\eqref{eqn:Rdecomp}.

    \ref{item:Gricterm2} follows from Lemma \ref{lemma:relationlieder}, since $R$ is compatible with the anchors.
    
    \ref{item:Gricterm3} follows from the same argument as for \ref{item:Gricterm1}, and the fact that the trace is compatible with $R$ in the following sense. At any point $c = (c_1, c_2) \in C$, let $A_i \colon (V_i^+)_{c_i} \to (V_i^+)_{c_i}$ for $i=1,2$ be linear maps such that $A_1(v_1) \sim_R A_2(v_2)$ for every pair $(v_1,v_2)$ such that $v_1 \sim_R v_2$. Then with respect to any given bases \smash{$\set{v^{(i)}_1, \dots ,v^{(i)}_n}$} of $V_i^+$ such that \smash{$v^{(i)}_j \sim_R v^{(i)}_j$}, the corresponding matrix elements are equal, \smash{$(A_1)^{j}_k = (A_2)^{j}_k$}, hence $\Tr_{V_1^+} (A_1) = \Tr_{V_2^+} (A_2)$.
\end{proof}

\begin{proposition}\label{cor:relatedric}
    If $R \colon (E_1, V_1^+) \rel (E_2, V_2^+)$ is a generalised isometry supported on $C \subseteq M_1 \times M_2$ which is a Dirac structure in $E_1 \times \overline{E}_2$, and $\div_1 \sim_R \div_2$, then \smash{$\mathrm{GRic}_{V_1^+, \div_1} \sim_R \mathrm{GRic}_{V_2^+, \div_2}$} and $\mathrm{GR}_{V_1^+, \div_1} \sim_C \mathrm{GR}_{V_2^+, \div_2}$.
\end{proposition}

\begin{proof}
    This follows immediately from Lemma~\ref{lem:riccomponents} and from $$\mathrm{GRic}_{V_1^-, \div_1} \sim_R \mathrm{GRic}_{V_2^-, \div_2} \ , $$ which can be shown using Lemma \ref{lem:riccomponents} by exchanging the roles of $V_i^+$ and $V_i^-$ for $i=1,2$. 
\end{proof}

Note that we do not need to require that the divergence operators are compatible with the generalised metrics.
Moreover, the same proof carries over when considering divergence operators coming from related Courant algebroid connections, as in Proposition \ref{prop:reldivfromcon}.

\medskip

\subsection{T-Duality for Generalised Ricci Tensors and String Backgrounds}~\\[5pt]
The relevant application of Proposition~\ref{cor:relatedric} for us is when $R$ is a T-duality relation.

\begin{corollary}\label{cor:Tdualric}
    Under the assumptions of Theorem \ref{thm:tdualdivergence}, if $\red E{}_1$ admits a divergence $\div_1$ which is (locally) compatible with $D_1$, then the generalised Ricci tensors \smash{$\gric_{V_1^+,\div_1}$} and \smash{$\gric_{V_2^+,\div_2}$} are T-duality related, where $\div_2$ is the divergence given by Theorem \ref{thm:tdualdivergence}. Similarly, the generalised scalar curvatures \smash{${\rm GR}_{V_1^+,\div_1}$} and \smash{${\rm GR}_{V_2^+,\div_2}$} are T-duality related.
\end{corollary}

Recalling Definition~\ref{def:genstringbackground}, we can then immediately infer that generalised T-dualities preserve solutions of generalised string background equations.

\begin{theorem}\label{cor:Tdualbackground}
    Under the assumptions of Theorem \ref{thm:tdualdivergence}, if a $D_1$-invariant pair $(V_1^+,\div_1)$ is a solution of the generalised string background equations on $\red E{}_1$, then the T-dual pair $(V_2^+,\div_2)$ is a solution of the generalised string background equations on $\red E{}_2$. More generally, if $(V_1^+,\div_1)$ is a generalised Einstein pair on $\red E{}_1$, then $(V_2^+,\div_2)$ is a generalised Einstein pair on~$\red E{}_2$.
\end{theorem}

\begin{proof}
    The first statement follows from Corollary \ref{cor:Tdualric} and Definition \ref{def:relatedtensors}: If \smash{${\rm GRic}_{V_1^+,\div_1}=0$} and \smash{${\rm GR}_{V_1^+,\div_1}=0$} on $\cQ_1$, then the property ${\rm pr}_i(\red C)=\varpi_i(M)=\cQ_i$ for $i=1,2$ ensures that \smash{$\gric_{V^+_2, \div_2}$} and \smash{${\rm GR}_{V^+_2, \div_2}$} vanish on all of $\cQ_2$.  
    
    The second statement follows similarly.
\end{proof}

\begin{example}[\textbf{Poisson-Lie T-Duality}]
    Consider the setting of Example \ref{eg:vysokypoissonlie}, with the related divergence operators \smash{$\div_\sfH \sim_{R_{\sfH,\sfH'}} \div_{\sfH'}$}. Let \smash{$V^+ \subset E$} be a $\sfG$-invariant generalised metric. Its reduction by both the $\frh$-action and the $\frh'$-action yields the generalised metrics $V^+_\sfH$ and $V^+_{\sfH'}$ on $\red E{}_\sfH$ and $\red E{}_{\sfH'}$ respectively. As shown in \cite{Vysoky2020hitchiker}, $R_{\sfH,\sfH'}$ is a generalised isometry between $V^+_\sfH$ and $V^+_{\sfH'}$. It then follows that $\gric_{V^+_\sfH, \div_\sfH} \sim_{R_{\sfH, \sfH'}} \gric_{V^+_{\sfH'}, \div_{\sfH'}}$.
\end{example}

\begin{example}[\textbf{T-Duality for Torus Bundles}]\label{eg:garciastreetsRicci}
Another example is provided by the generalised Ricci tensors for T-duality between torus bundles, as formulated in \cite[Proposition 6.3]{Garcia-Fernandez:2016ofz}. In this case, given two T-dual $\sfT^k$-equivariant exact Courant algebroids $\red E{}_1$ and $\red E{}_2$  endowed with T-dual pairs $(V^+_1, \div_1)$ and $(V^+_2, \div_2)$, respectively, as in Definition \ref{def:garcia-fernandezstreetsTdual}, their generalised Ricci tensors satisfy
\begin{align}
    {\pi_1}_*\, \gric_{V_1^+, \div_1} = \mathscr{R}^*\, {\pi_2}_*\, \gric_{V_2^+, \div_2} \ .
\end{align}
This is equivalent to \smash{$\mathrm{GRic}_{V_1^+, \div_1} \sim_R \mathrm{GRic}_{V_2^+, \div_2}$}, where $R$ is the generalised isometry which appears in Proposition \ref{prop:garciastreetstdual}.
\end{example}

\medskip

\subsubsection{Ricci-Buscher Rules}\label{sec:RicciBuscher}~\\[5pt]
 Let $\cQ_1$ and $\cQ_2$ be T-dual circle bundles over a base $\cB$, with T-duality relation $$R \colon (\IT \cQ_1, H_1, \tau_1) \rel (\IT \cQ_2, H_2, \tau_2)$$ as in Example \ref{rmk:buscher}. 
 For any open cover $U^{(i)}$ trivialising $\cQ_1$, the Kalb-Ramond flux on $\cQ_1$ is given locally by $H_1=\de b_1$ for some two-form $b_1\in \sfOmega^2(U^{(i)})$. If $\tau_1 = (g_1,0)$ for an $\sfS^1$-invariant Riemannian metric $g_1$ on $\cQ_1$, then the T-dual generalised metric $\tau_2=(g_2,b_2)$ is given by the Buscher rules~\eqref{eqn:buscherrules}.
 
On $\cQ_1$ we take the dilaton\footnote{We adhere to the usual Einstein summation convention over repeated upper and lower indices throughout.} $$e =  4\, \de \log \,(g_1)_{\theta\theta} = \frac{4}{(g_1)_{\theta\theta}} \, \partial_\alpha (g_1)_{\theta\theta} \, \de x^{\alpha} \ .$$ This ensures that \smash{$\div{}_1:=\div_{4\log\,(g_1)_{\theta\theta}} \sim_R \div_{\mu_{g_2}}=:\div{}_2$}. 
For this choice $\Fi^+ = \frac12\,e$. 

If $\nabla^+$ is the Bismut connection for $(g_1,H_1)$, then its Christoffel symbols are given by $$\Gamma^{+\,i}{}_{jk} = \Gamma^i{}_{jk} + \tfrac{1}{2}\, (H_1)^i{}_{jk} \ , $$ where $\Gamma^i{}_{jk}$ are the Christoffel symbols for the Levi-Civita connection $\nabla^{g_1}$ of $g_1$, and indices are raised using the metric $g_1$ in the standard way. Hence
\begin{align}
    2\,(g_1)_{\theta\theta}\, \nabla^+ \Fi^+ &=\Big(\partial_{\alpha} \,\partial_{\beta} (g_1)_{\theta\theta} - \frac{1}{(g_1)_{\theta\theta}} \,\partial_\alpha (g_1)_{\theta\theta}\, \partial_\beta (g_1)_{\theta\theta}\Big ) \, \de x^\beta \otimes \de x^\alpha \\
    & \quad \, - \partial_{\alpha} (g_1)_{\theta\theta}\,\Big ( \Gamma^\alpha{}_{ij} \, \de x^i \odot \de x^j - \frac 12\, (H_1)^\alpha{}_{ij} \, \de x^i \wedge \de x^j \Big) \ .
\end{align}

Let ${\rm Ric}_{g_1}$ be the Ricci tensor of~$g_1$, with coordinate frame components $({\rm Ric}_{g_1})_{ij}:={\rm Ric}_{g_1}(\partial_i,\partial_j)$. Then the generalised Ricci tensor of $(\tau_1,\div{}_1)$ is given from Equation~\eqref{eqn:gricclassically} by
\begin{align}
    (\gric_1)_{ij} :\!&= \gric_{g_1,\div{}_1}\big( (\partial_i)^+, (\partial_j)^-\big) \\[4pt]
    &= ({\rm Ric}_{g_1})_{ij} - \frac{1}{4}\, (H_1^2)_{ij} - \frac12\,(\de^*H_1)_{ij} \\
    & \hspace{1cm} \, + \frac{1}{(g_1)_{\theta\theta}}\,\partial_i\, \partial_j (g_1)_{\theta\theta} - \frac{1}{(g_1)_{\theta\theta}^2}\,\partial_i (g_1)_{\theta\theta}\, \partial_j (g_1)_{\theta\theta} \\
    & \hspace{2cm} \, - \frac{1}{2\,(g_1)_{\theta\theta}}\,\partial_\alpha (g_1)_{\theta\theta} \, \Gamma^\alpha{}_{ij} - \frac{1}{4\,(g_1)_{\theta\theta}}\,\partial_\alpha (g_1)_{\theta\theta} \,(H_1)^\alpha{}_{ij} \ ,
\end{align}
where $(H_1^2)_{ij} = (H_1)_{ikl}\,(H_1)_j{}^{kl}$ and $(\de^*H_1)_{ij} = -g_1^{kl}\,\nabla_k^{g_1}(H_1)_{lij}$.

From the local form of the relation $R$ given by Equation \eqref{eqn:Rlocally}, together with Corollary~\ref{cor:Tdualric}, we obtain
\small
\begin{align}
    \gric_{(g_2,b_2), \div{}_{2}}\big((\partial_\alpha - (b_1)_{\alpha \theta} \,\partial_{\theta_2})^+ , (\partial_\beta - (b_1)_{\beta \theta} \,\partial_{\theta_2})^-\big) &= \gric_{g_1, \div{}_1}\big((\partial_\alpha)^+ , (\partial_\beta)^-\big) \ ,\\[4pt]
    \gric_{(g_2,b_2), \div{}_{2}}\big((\partial_\alpha - (b_1)_{\alpha \theta} \,\partial_{\theta_2})^+ , (\partial_{\theta_2})^-\big) &= \gric_{g_1, \div{}_1}\big((\partial_\alpha)^+ , -(\tfrac{1}{(g_1)_{\theta\theta}} \,\partial_{\theta_1})^-\big) \ ,\\[4pt] 
    \gric_{(g_2,b_2), \div{}_{2}}\big((\partial_{\theta_2})^+ , (\partial_\alpha - (b_1)_{\alpha \theta}\, \partial_{\theta_2})^-\big) &= \gric_{g_1, \div{}_1}\big((\tfrac{1}{(g_1)_{\theta\theta}} \,\partial_{\theta_1})^+ , (\partial_\alpha)^-\big) \ ,\\[4pt] 
    \gric_{(g_2,b_2), \div{}_{2}}\big((\partial_{\theta_2})^+ , (\partial_{\theta_2})^-\big) &= \gric_{g_1, \div{}_1}\big((\tfrac{1}{(g_1)_{\theta\theta}}\, \partial_{\theta_1})^+ , -(\tfrac{1}{(g_1)_{\theta\theta}}\, \partial_{\theta_1})^-\big) \ .   
\end{align}
\normalsize

In the coordinate frame $\set{\partial_{\alpha}\,,\, \partial_\theta}_{\alpha=1,\dots,\dim \cB}$, using Equation~\ref{eqn:gricclassically} the Ricci tensor ${\rm Ric}_{g_2}$ of $g_2$ is thus given by
\begin{align}
    ({\rm Ric}_{g_2})_{\theta \theta} &= -\frac{1}{(g_1)_{\theta\theta}^2}\, (\gric_1)_{\theta \theta} +\frac{1}{4}\,(H_2^2)_{\theta\theta} + \frac12\,(\de^*H_2)_{\theta\theta} \ , \\[4pt]
    ({\rm Ric}_{g_2})_{\alpha \theta} &= -\frac{1}{(g_1)_{\theta\theta}}\, (\gric_1)_{\alpha \theta} +\frac{1}{4}\,(H_2^2)_{\alpha \theta} + \frac 12\,(\de^*H_2)_{\alpha \theta} + \frac{(b_1)_{\alpha \theta}}{(g_1)_{\theta\theta}^2} \,(\gric_1)_{\theta \theta} \ , \\[4pt]
    ({\rm Ric}_{g_2})_{\alpha \beta } &= (\gric_1)_{\alpha \beta} +\frac{1}{4}\,(H_2^2)_{\alpha \beta} + \frac 12\,(\de^*H_2)_{\alpha\beta} \\
    &\quad \, + \frac{(b_1)_{\beta \theta}}{(g_1)_{\theta\theta}}\,(\gric_1)_{\alpha\theta} - \frac{(b_1)_{\beta \theta}\, (b_1)_{\alpha \theta}}{(g_1)_{\theta\theta}^2} \, (\gric_1)_{\theta \theta} -\frac{(b_1)_{\alpha\theta}}{(g_1)_{\theta\theta}}\,(\gric_1)_{\theta \beta} \ . 
\end{align}

\medskip

\subsection{Generalised Ricci Flow and T-Duality}~\\[5pt]
 We are now ready to discuss the interplay between geometric T-duality, introduced in Subsection~\ref{subsect:geomTdual}, and the generalised Ricci flow, inspired by the treatment of~\cite{Severa:2018pag}.

Let $E$ be a Courant algebroid over $M$ endowed with a generalised metric $V^+$ and a divergence operator $\div$.   The \emph{generalised Ricci flow} on $E$ is the system of second order non-linear partial differential equations
\begin{align}
    \frac{\partial}{\partial t} \tau = -2 \, \mathrm{GRic}_{\tau,\div}\quad , \quad \frac{\partial}{\partial t} \div = -\cD \, \mathrm{GR}_{\tau,\div} \ ,
\end{align}
where $\tau$ is the automorphism of $E$ whose $+1$-eigenbundle is $V^+$, and $t\in[0,\infty)$ is the flow parameter which we call `time'.

In the setting of geometric T-duality described in Subsections~\ref{subsect:geomTdual} and~\ref{sect:Tdualitydiv}, whereby a generalised isometry $R$ plays the role of the T-duality relation, we consider a $D_1$-invariant pair $(\tau_1, \div{}_1)$ of a generalised metric and divergence on $\red E{}_1$, yielding a pair $(\tau_2, \div{}_2)$ on the T-dual Courant algebroid $\red E{}_2$. These pairs give \emph{a priori} unrelated generalised Ricci flows. We will now show that they are indeed related. 

\begin{theorem} \label{thm:compatibilityRicci}
    Under the assumptions of Theorems \ref{thm:maingeneral} and \ref{thm:tdualdivergence}, suppose that a unique solution $(\tau_1(t), \div{}_1(t))$ of the generalised Ricci flow on $\red E{}_1$ with initial condition $(\tau_1(0), \div{}_1(0)) = (\tau_1, \div{}_1)$ exists for $t \in [0,T),$ for some $T \in \IR_{>0},$  and that the Lie algebra of isometries $\frk{k}_{\tau_1}$ consists of complete vector fields. Then there is a unique T-dual family of pairs $(\tau_2(t), \div{}_2(t))$ that is a solution of the generalised Ricci flow on $\red E{}_2$ for all $t\in[0,T)$. 
\end{theorem}

\begin{proof}
    By assumption, any  Killing vector $X$ for $\tau_1$ integrates to a one-parameter family of Courant algebroid automorphisms $\boldsymbol\Phi_s$ such that $\boldsymbol\Phi_s^* \,\tau_1 = \tau_1$ and $\boldsymbol\Phi^*_s\, \div{}_1 = \div{}_1$.
    If $(\tau_1(t), \div{}_1(t))$ is a solution of the generalised Ricci flow, then by the naturality of the generalised Ricci tensor and generalised scalar curvature, it follows that $(\boldsymbol\Phi_s^*\,\tau_1(t), \boldsymbol\Phi_s^*\,\div{}_1(t))$ is also a solution of the generalised Ricci flow. By the uniqueness of the solution, it follows that $\boldsymbol\Phi_s$ is a one-parameter group of  isometries of $\left( \tau_1(t), \div{}_1(t) \right)$ for all $t\in[0,T)$. 
    Hence the vector field $X$ is a Killing vector for $\tau_1(t)$, and we may take $\frk{k}_{\tau_1(t)} \coloneqq \frk{k}_{\tau_1}$, for which $\div{}_1(t)$ is compatible. By Theorems \ref{thm:maingeneral} and \ref{thm:tdualdivergence}, we may construct the unique T-dual pair $(\tau_2(t), \div{}_2(t))$ for all $t\in[0,T)$, which by Corollary \ref{cor:Tdualric} satisfies the generalised Ricci flow.    
\end{proof}

In this sense, we may therefore say that the T-duality relation $R$ is \emph{compatible} with the generalised Ricci flow. In particular, it follows from Theorem~\ref{cor:Tdualbackground} that fixed points of the generalised Ricci flow, i.e. solutions to the generalised string background equations of Definition~\ref{def:genstringbackground}, are preserved under T-duality.

The proof of Theorem~\ref{thm:compatibilityRicci} assumes the existence and uniqueness of the generalised Ricci flow, together with completeness of the vector fields generating isometries of $\tau_1$, which as shown in \cite{Streets:2024rfo} are guaranteed on a compact manifold. This leads to

\begin{corollary}\label{cor:Riccicompactcase}
   Let $\cQ_1$ be a compact manifold with a $D_1$-invariant pair $(\tau_1,\div{}_1)$. Then there is a unique pair $(\tau_2(t), \div{}_2(t))$ solving the generalised Ricci flow on $\red E{}_2$ which is T-dual to the generalised Ricci flow solution $(\tau_1(t), \div{}_1(t))$ on $\red E{}_1$, for all times $t\in[0,T)$.
\end{corollary}

\medskip

\subsubsection{Local T-Duality of the Generalised Ricci flow}~\\[5pt]
The proof of Theorem~\ref{thm:compatibilityRicci} requires that the Killing vector fields for $\tau_1$ integrate to a flow on $\red E{}_1$, thus this cannot be done in the local case. We shall now give the appropriate statement of the T-duality of generalised Ricci flows in the local case.  The construction of the T-dual generalised Ricci tensor can be done immediately as a consequence of Theorem \ref{thm:mainlocal} and Remark \ref{rmk:tdualdivlocal}. In the course of the proof of Theorem \ref{thm:compatibilityRicci}, we assumed  that the generators of $\frk{k}_{\tau_1}$ were  given by complete vector fields integrating to isometries, and that the generalised Ricci flow solution is unique, in order to assert that $\frk{k}_{\tau_1(0)}\subset \frk{k}_{\tau_1(t)}$. Dropping the former assumptions and simply assuming the latter condition yields

\begin{theorem}\label{thm:localRicciflow}
    Let $\left( \tau_1(t), \div{}_1(t) \right)$ be a solution of the generalised Ricci flow on $\red E{}_1$ for  $t\in[0,T)$ with initial condition $\tau_1(0)=\tau_1$ and $\div{}_1(0)=\div{}_1$. Let $\set{U^{(i)}}_{i\in I}$ be an open cover of $\cQ_1$ and suppose that $\left( \tau_1(t), \div{}_1(t) \right)$ is locally $D_1$-invariant, where $D_1$ is generated by vector fields \smash{$\set{X_1^{(i)},\dots,X_{r_1}^{(i)}}_{i\in I}$}, for all $t\in [0,T)$. Then there is a unique family $(\tau_2(t), \div{}_2(t))$ which is T-dual to $(\tau_1(t), \div{}_1(t))$ and which is a solution to the generalised Ricci flow on $\red E{}_2$.
\end{theorem}

\medskip

\subsubsection{Correspondence Spaces}~\\[5pt]
The (local) T-duality of the generalised Ricci flow can be formulated in the setting of Subsection~\ref{sec:correspondencediv}, wherein $\cQ_1$ and $\cQ_2$ are fibre bundles over a common base $\cB$. The classic    example is given by T-duality between torus bundles as described in \cite[Theorem~6.5]{Garcia-Fernandez:2016ofz} and \cite[Theorem~1.2]{Streets:2017506}. That is,
we consider the case where $\cQ_1$ and $\cQ_2$ are $\sfT^k$-bundles over the same base $\cB$, with correspondence diagram
\begin{equation}
        \begin{tikzcd}
             & M \arrow[swap]{dl}{\varpi_1} \arrow{dr}{\varpi_2} & \\
            \cQ_1 \arrow[swap]{dr}{\pi_1} &  & \cQ_2 \arrow{dl}{\pi_2} \\
             & \cB & 
        \end{tikzcd}
    \end{equation}
where $M=\cQ_1 \times_{\cB} \cQ_2$. 

Suppose that $\cQ_1$ and $\cQ_2$ are T-dual, with T-duality relation $R \colon \red E{}_1 \rel \red E{}_2$. If $(\tau_1(t), \div{}_1(t))$ is a family of $\sfT^k$-invariant pairs on $\red E{}_1$ then, as seen in Proposition \ref{prop:garciastreetstdual} and Example \ref{eg:garciastreetsRicci}, for each $t$ there is a unique T-dual pair $(\tau_2(t), \div{}_2(t))$ such that $R \colon (\red E{}_1, \tau_1(t)) \rel (\red E{}_2, \tau{}_2(t))$ is a generalised isometry and \smash{$\gric_{\tau_1(t),\div{}_1(t)} \sim_R \gric_{\tau_2(t),\div{}_2(t)}$}. From Theorem \ref{thm:localRicciflow} with $U^{(1)} = \cQ_1$, it follows that if $(\tau_1(t), \div{}_1(t))$ solves the generalised Ricci flow on $\red E{}_1$, then so does $(\tau_2(t), \div{}_2(t))$ on $\red E{}_2$. This is the content of \cite[Theorem~10.19]{Garcia-Fernandez:2020ope}.

In the case that $\cB$ is compact, starting with a $\sfT^k$-invariant pair $(\tau_1, \div{}_1)$, since the vector fields generating $\frk{k}_{\tau_1}=\frt^k$ are complete we may apply Corollary \ref{cor:Riccicompactcase}.

\medskip

\section{Examples} \label{sect:examples}
We have already described the behaviour of the T-duality relation for classic   classes of examples such as principal torus bundles and Poisson-Lie T-duality. In this final section we look at some explicit examples that are not covered by these cases discussed thus far throughout the paper.

\medskip

\subsection{Geometric T-Duality for Hamilton's Cigar Soliton}~\\[5pt]
 Let $\cQ_1 = \IR^2 \setminus \set{0}$ with the cigar metric $$g_1 = \frac{\de x\otimes\de x + \de y\otimes\de y}{1+x^2+y^2} \ , $$ where $(x,y)\in\cQ_1$. In geodesic coordinates this can be written as $$g_1 = \de s \otimes \de s +\tanh^2(s) \, \de \theta \otimes \de \theta \ . $$ Then $\partial_\theta=\frac\partial{\partial\theta}$ is a Killing vector for this metric, and it generates a circular foliation $\red \cF{}_1$ with leaf space $\cQ_1/ \red \cF{}_1 = \IR_{>0}$.

The family of metrics
\begin{align}
     g_1(t) = \frac{\cosh^2(s)}{\e^{4\,t}+\sinh^2(s)}\,\de s \otimes \de s + \frac{\sinh^2(s)}{\e^{4\,t}+\sinh^2(s)} \,\de \theta \otimes \de \theta
\end{align}
on $\cQ_1$ is  a Ricci soliton with initial condition $g_1(0) = g_1$. That is, it is a solution to the  Ricci flow $$\frac\partial{\partial t} g_1(t) = -2\,\ric_{g_1(t)} $$ which is self-similar, evolving the metric $g_1$ purely by pullback under the one-parameter family of diffeomorphisms $\big(\sinh(s),\theta\big)\mapsto\big(\e^{2\,t}\sinh(s),\theta\big)$. This Ricci soliton is called Hamilton's cigar~\cite{Hamilton1982}. In string theory it appears as the Euclidean Witten black hole background, which is a fixed point of the  renormalisation group flow of a gauged Wess-Zumino-Witten model (see e.g.~\cite{Lambert:2012tq}). In this subsection we demonstrate the behaviour under geometric T-duality of this classic solution to the Ricci flow.

Let $\cQ_2 = \IR \times \IR_{>0}$ with coordinates $(z,s)$, and  foliation $\red\cF{}_2$ given by the projection to the first factor, so that $\cQ_2/\red\cF{}_2=\IR_{>0}$. Let $M = \cQ_1 \times_{\IR_{>0}} \cQ_2$, with foliations $\cF_1$ generated by $\partial_z=\frac\partial{\partial z}$ and $\cF_2$ generated by $\partial_\theta$, and take $$B = \de \theta \wedge \de z \ . $$ Denoting $\partial_s = \frac\partial{\partial s}$, we obtain the subbundle $R\subset\IT\cQ_1\times\IT\cQ_2$ whose sections are given by
\begin{align}
     \sfgamma(R) = \text{Span}_{C^\infty(\IR_{>0})}\set{(\partial_s,\partial_s)\,,\, (\de s, \de s)\,,\, (\partial_\theta, \de z)\,,\, (\de \theta, \partial_z)} \ .
\end{align}

Since $\partial_\theta$ is a Killing vector of $g_1(t)$ for all $t$, by Theorem \ref{thm:localRicciflow} we get the generalised isometry
\begin{align}
    R \colon \big(\IT \cQ_1, 0 , g_1(t)\big) \rel \big(\IT \cQ_2, 0 , g_2(t)\big) \ ,
\end{align}
with the T-dual family of metrics $g_2(t)$ on $\cQ_2$ given by
\begin{align}
    g_2(t) &= \frac{\cosh^2(s)}{\e^{4\,t}+\sinh^2(s)}\,\de s \otimes \de s + \frac{\e^{4\,t}+\sinh^2(s)}{\sinh^2(s)}\, \de z\otimes \de z \ ,
\end{align}
and $R$-related divergence operators $\div_{\mu_{g_1(t)}}\sim_R\div_{\phi}(t)$ where
\begin{align}
    \div_\phi(t) &= \div_{\mu_{g_2(t)}} +  \langle \de\phi,\, \cdot \,  \rangle_2 = \div_{\mu_{g_2(t)}} -4\, \langle \Fi^+, \,\cdot\, \rangle_2 \ ,
\end{align}
with $$\phi=- \log \left (\frac{\e^{4\,t}+\sinh^2(s)}{\sinh^2(s)} \right)  \quad , \quad \Fi^+ = -\frac1{4}\,\de\phi =  -\frac{\e^{4\,t}\, \cosh (s)}{2\sinh^{3}(s) + 2\e^{4\, t} \sinh( s)} \, \de s \ . $$

Then the pair $(g_2(t), \div_\phi(t))$ solves the system
\begin{align}\label{eq:cigarflow}
    \frac{\partial}{\partial t}g_2(t) = -2\, \ric_{g_2(t)} - 2\, \nabla^{g_2(t)} \Fi^+\quad , \quad \frac{\partial}{\partial t} \div_\phi(t) = -\de\, \big( \mathrm{R}_{g_2(t)} - 4\e^{\phi}\,\Delta_{g_2(t)} \e^{-\phi} \big) \ .
\end{align}
Analogously to the original cigar soliton, the T-dual solution is a \emph{generalised} Ricci soliton~\cite{Paradiso:2021wuw,Streets:2024rfo}, i.e.~a self-similar solution of the generalised Ricci flow \eqref{eq:cigarflow}, evolving the metric $g_2$ and dilaton field $\phi$ purely by pullback under the one-parameter family of Courant algebroid automorphisms given by the diffeomorphisms $\big(\sinh(s),z\big)\mapsto\big(\e^{2\,t}\sinh(s),z\big)$.

\medskip

\subsection{Three-Dimensional Hyperbolic Space}~\\[5pt]
 In this subsection we shall present a new example of geometric T-duality that will illustrate how the notion of relations between generalised Ricci tensors aids as a computational tool as well, in addition to yielding an explicit T-dual solution of the generalised Ricci flow.

\medskip

\subsubsection{Geometric T-Duality}~\\[5pt]
 Consider the pair of trivial circle bundles
\begin{gather}
    \cQ_1 = \cQ_2 = \sfS^1 \times \IR^2_{>0} = \set{(\vartheta, r, z) \ | \  \vartheta \in [0,2\pi) \ , \ r,z>0 } \ .
\end{gather}
Denote by $\partial_{\theta_1}$ the coordinate vector field on $\cQ_1$ in the $\vartheta$-direction, and similarly for $\partial_{\theta_2}$ on $\cQ_2$. The foliations $\red \cF{}_i$ generated by $\partial_{\theta_i}$ induce a Hausdorff-Morita equivalence between $(\cQ_1, \red \cF{}_1)$ and $(\cQ_2, \red \cF{}_2)$, with correspondence diagram
\begin{equation}
    \begin{tikzcd}
        & M \arrow[swap]{dl}{\varpi_1} \arrow{dr}{\varpi_2} & \\
        \cQ_1 \arrow[swap]{dr}{\pi_1} & & \cQ_2 \arrow{dl}{\pi_2} \\ & \IR^2_{>0} &
    \end{tikzcd}
\end{equation}
where $M = \cQ_1 \times_{\IR_{>0}^2} \cQ_2$.

On $\cQ_1$, take the hyperbolic metric $g_1$ and closed three-form $H_1$ given by
\begin{align}
    g_1 = \lambda^2 \, g_{\sfH^3} := \frac{\lambda^2}{z^2}\,\big(r^2 \, \theta_1 \otimes \theta_1 + \de r \otimes \de r + \de z \otimes \de z\big)  \quad ,  \quad H_1 = h\,\mu_{\sfH^3} := h\,r\, \theta_1 \wedge \de r \wedge \de z \ ,
\end{align} 
where $\theta_i \in \sfOmega_{\rm cl}^1(\cQ_i)$ is dual to $\partial_{\theta_i}$ and $\lambda,h \in \IR$.
We denote by $\widehat\partial_{\theta_i}$ the vector field dual to $\widehat\theta_i:=\varpi_i^* \theta_i$, for $i=1,2$. Then we obtain the foliations $\cF_1$ and $ \cF_2$ of $M$ generated by $\widehat \partial_{\theta_2}$ and $\widehat \partial_{\theta_1}$ respectively. 

Introduce the two-form $$B = \widehat\theta_1 \wedge \big(\,\widehat \theta_2 + \tfrac12\,h\, r{}^2\, \de  z\big)$$ on $M$, such that $\de B = \varpi_1^* H_1$. Then $B$ is non-degenerate as a map from $T \cF_1 \otimes T\cF_2$ to $\IR$.
With $K_1 = T \cF_1$ and $K_2 = \e^{-B}\,(T \cF_2)$, it follows that $K_1\cap K_2 = \set{0}$ and $K_2\cap K_1^\perp \subset K_1$. By~\cite[Theorem~5.8]{DeFraja:2023fhe} we obtain the diagram of standard Courant algebroids
\begin{equation}
    \begin{tikzcd}
        & (\IT M , \varpi_1^*H_1) \arrow[tail,swap]{dl}{Q(K_1)} \arrow[tail]{dr}{Q(K_2)} & \\
        (\IT \cQ_{1} , H_1) \arrow[dashed]{rr}{R} & & (\IT\cQ_{2}, 0)
    \end{tikzcd}
\end{equation}
where we may explicitly write the relation $R$ from its generating sections as the $C^\infty(\IR_{>0}^2)$-module
\small
\begin{align} \label{eqn:relhyp}
    \sfgamma(R) = \text{Span}_{C^\infty(\IR_{>0}^2)} \set{(\partial_r,\partial_r)\,,\, ( \partial_{\theta_1}, \theta_2 + \tfrac{h\,r^2}{2}\, \de z)\,,\, (\partial_z, \partial_z - \tfrac{h\,r^2}{2}\, \partial_{\theta_2})\,,\, ( \de r, \de r)\,,\, (\theta_1,\partial_{\theta_2})\,,\, (\de z, \de z)} \ .
\end{align}
\normalsize

Note that $\partial_{\theta_1}$ is a Killing vector for $g_1$, and \smash{$\pounds_{\widehat\partial_{\theta_1}} B= \pounds_{\widehat \partial_{\theta_2}} B = 0$}. Thus taking the isometry algebra $\frk{k}_{g_1} = \text{Span}_\IR \set{\partial_{\theta_1}}$, it follows that $g_1$ defines a $T \red \cF{}_1$-invariant generalised metric on $\IT\cQ_1$. From Theorem~\ref{thm:maingeneral} it follows that $$R \colon \left( \IT \cQ_1, H_1,g_1 \right) \rel \left( \IT \cQ_2,0, g_2 \right)$$ is a generalised isometry, where the T-dual metric on $\cQ_2$ is given by
\begin{align}\label{eqn:g2hyperbolic}
    g_2 = \frac{z^2}{\lambda^2\, r^2}\,\theta_2\otimes \theta_2 + \frac{\lambda^2}{z^2}\,\de r \otimes \de r +  \frac{h\,z^2}{\lambda^2 }\, \de z \odot \theta_2 + \left (\frac{\lambda^2}{z^2} + \frac{h^2\,r^2\,z^2}{4\,\lambda^2}\right )\, \de z \otimes \de z \ .
\end{align}
By Proposition \ref{prop:tau2isD2invariant}, $g_2$ is then $T\red \cF{}_2$-invariant, with isometry algebra $\frk{k}_{g_2} = \text{Span}_\IR \set{\partial_{\theta_2}}$.

\medskip

\subsubsection{Generalised Ricci Tensor}~\\[5pt]
 The Ricci tensor of $g_1 = \lambda^2\,g_{\sfH^3}$ is $$\ric_{g_1} = -2\,g_{\sfH^3} \ , $$ hence $g_1$ is an Einstein metric. We further find $$H_1^2 = \frac{2\,h^2}{\lambda^4}\,g_{\sfH^3} \quad , \quad \de^* H_1 = 0 \ . $$ 
 
 Then the Ricci tensor for $g_2$ can be calculated from the Ricci-Buscher rules of Subsection~\ref{sec:RicciBuscher} as
\begin{align*}
    (\ric_{g_2})_{rr} = -\frac{h^2}{2\,\lambda^4\, z^2} - \frac{2}{r^2} \quad , & \quad  (\ric_{g_2})_{zz} = -\frac2{r^2} \quad , \quad (\ric_{g_2})_{\theta\theta}  = \frac{h^2\, z^2}{2\, \lambda^8\, r^2} - \frac{2\,z^4}{\lambda^4\, r^4} \ , \\[4pt]
    (\ric_{g_2})_{rz} = \frac{2}{r\, z} - \frac{h^2\, r}{2\,\lambda^4} & \quad , \quad
    (\ric_{g_2})_{r\theta} =  -\frac{h}{\lambda^4\, r} \quad , \quad (\ric_{g_2})_{z\theta} = 0 \ .
\end{align*}

\medskip

\subsubsection{Generalised Ricci Flow}~\\[5pt] 
Consider now the geometry with the dilaton turned off on $\cQ_1$, i.e. \smash{$\div{}_1 = \div_{\mu_{g_1}} = \div_{\mu_{g_{\sfH^3}}}$}. Allowing $ \lambda$ and $h$ to depend on time, the generalised Ricci flow on $\IT\cQ_1$ is the one-parameter family of pairs $(g_1(t),\div{}_1(t))$ with $g_1(t)=\lambda(t)^2\,g_{\sfH^3}$ and $\div{}_1(t)=\div_{\mu_{g_1(t)}}= \div_{\mu_{g_{\sfH^3}}}$ such that
\begin{align}
    \frac{\partial}{\partial t}\lambda(t) = \frac2{\lambda(t)} + \frac{h(t)^2}{2\,\lambda(t)^5} \quad , \quad  \frac\partial{\partial t}h(t) = 0 \quad ,
     \quad \frac{\partial}{\partial t}\, \div{}_1(t)= \de \,\Big(\frac{6}{\lambda(t)^2} + \frac{h(t)^2}{2\,\lambda(t)^6}\Big) = 0  \ .
\end{align}

The metric $g_2(t)$ given by Equation \eqref{eqn:g2hyperbolic} is then the solution of the system 
\begin{align}
    \frac{\partial}{\partial t} g_2(t) = -2\,\ric_{g_2(t)} -2\, \nabla^{g_2(t)}\, \de \log\big({z^2}/{r^2}\big) \quad , \quad \frac{\partial}{\partial t} \div_\phi(t) = -\de\,\Big(\mathrm{R}_{g_2(t)}  - \frac{4\,z^2}{r^2}\, \Delta_{g_2(t)}\big(r^2/z^2\big) \Big) \ ,
\end{align}
with dilaton field $$\phi = \log(z^2/r^2) \ . $$
These contain the same information as the original flow equations.

Thus the generalised Ricci flow becomes the evolution of the scale factor $\lambda(t)$, with constant Kalb-Ramond flux $h$. For large $\lambda(t)$ the solution asymptotes to $\pm\,\sqrt {\lambda_0^2+2\,t}$, where $\lambda_0=\lambda(0)$, and therefore the flow diverges for long times independently of $h$. In particular, the behaviour for large $\lambda(t)$ is essentially constant. For small $\lambda(t)$ the flow diverges as $\pm\,\sqrt[6]{\lambda_0^6+3\,h^2\,t}$. In particular, the flow has no fixed points for any $h\in\IR$, and so the solution does not converge to any generalised string background.

\medskip

\subsection{Klein Bottle}\label{sec:Kleinbottle}~\\[5pt]
In this subsection we consider an example of geometric T-duality between non-principal $\sfS^1$-bundles. 

Consider the Klein bottle $\sfK$ which is the compact non-orientable surface defined by
\begin{align}
    \sfK = \set{(x,y) \in [0,1]\times [0,1]} \big/ \sim \ ,
\end{align}
where $(x,0) \sim (x,1)$ and $(0,y) \sim (1,1-y)$.
We regard $\sfK$ as a circle bundle over $\sfS^1$, with $\sfS^1$ fibres given by $\sfK_{x_0} = \set{(x_0,y)}$, which we may depict as
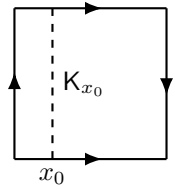
\small
\begin{center}
\begin{tikzpicture}[scale=0.5,>={Latex},
thick,decoration={
    markings,
    mark=at position 0.58 with {\arrow{>}}}
    ] 
    \draw[postaction={decorate}] (-2,-2)--(2,-2);
    \draw[postaction={decorate}] (2,2)--(2,-2);
    \draw[postaction={decorate}] (-2,2)--(2,2);
    \draw[postaction={decorate}] (-2,-2)--(-2,2);
    \draw[dashed] (-1,-2) node[below] {$x_0$} -- node[right] {$\sfK_{x_0}$} (-1,2);
\end{tikzpicture}    
\end{center}
\normalsize

\medskip

\subsubsection{Correspondence Space}~\\[5pt]
Let $\mathsf{K}_1 = \mathsf{K}_2=\sfK$ be Klein bottles, viewed as $\sfS^1$-bundles over a common circle $\sfS^1$, with projection maps $\pi_i \colon \sfK_i \to \sfS^1$ given by $\pi_i(x,y) = x$. Their fibred product is given by
\begin{align}
    M = \sfK_1 \times_{\sfS^1} \sfK_2 = \set{((x,y)\,,\,(x,z)) \in [0,1]^{\times 4}} \big/ \sim \ ,
\end{align}
with 
\begin{gather}
    \big((x,0)\,,\,(x,z)\big)\sim \big((x,1)\,,\,(x,z)\big) \quad , \quad  \big((x,y)\,,\,(x,0)\big)\sim \big((x,y)\,,\,(x,1)\big) \ , \\[4pt]
    \big((0,y)\,,\,(0,z)\big) \sim \big((1,1-y)\,,\,(0,z)\big) \sim \big((1,1-y)\,,\,(1,1-z)\big) \ .
\end{gather}
This is a $\sfT^2$-bundle over $\sfS^1$, which may be more succinctly described as
\begin{align}
    M = \set{(x,y,z) \in [0,1]^{\times 3}} \big/ \sim
\end{align}
with
\begin{gather}
    (x,0,z)\sim (x,1,z) \quad , \quad  (x,y,0)\sim (x,y,1) \quad , \quad
    (0,y,z) \sim (1,1-y,1-z) \ ,
\end{gather}
and projection map $\widehat\pi:M\to\sfS^1$ given by $\widehat\pi(x,y,z) = x$.

We cover $M$ with open charts given by
\begin{align}
    U^{(1)} = \set{(x,y,z) \in M  \ | \ x \in (0,1)} \quad &, \quad U^{(2)} = \set{(x,y,z) \in M \ | \ x \neq \tfrac 12} \ ,\\[4pt]
    \Fi^{(1)}\colon U^{(1)} \longrightarrow (0,1) \times [0,1]^{\times 2} \quad &, \quad \Fi^{(2)} \colon U^{(2)} \longrightarrow (0,1) \times [0,1]^{\times 2} \ ,
\end{align}
where
\begin{align}
    \Fi^{(1)}(x,y,z) = (x,y,z) \quad , \quad \Fi^{(2)}(x,y,z) = \begin{cases}
        (x + \frac 12, y, z) \quad , \quad x\in[0,\frac 12)\\[4pt]
        (x - \frac 12, 1-y, 1-z) \quad , \quad x \in (\frac 12, 1]
    \end{cases} \ .
\end{align}
On their intersection $$U^{(1)} \cap U^{(2)} = \set{(x,y,z) \in M  \ | \ x\in(0,\tfrac 12)} \, \sqcup \, \set{(x,y,z) \in M  \ | \ x \in (\tfrac 12 ,1)} \eqqcolon W^{(1)} \sqcup W^{(2)}$$ the transition functions are given by
\begin{align}
\arraycolsep=1.4pt
    \Fi^{(2)} \circ \big(\Fi^{(1)}\big)^{-1} = \begin{cases}
          \small \Bigg(\begin{matrix}
            1 & \ 0 & \ 0 \\[-5pt]
            0 & \ 1 & \ 0 \\[-5pt]
            0 & \ 0 & \ 1
        \end{matrix}\Bigg) \normalsize \qquad \text{on} \quad W^{(1)} \\[4pt]
         \small \Bigg(\begin{matrix}
            1 & 0 & 0 \\[-5pt]
            0 & -1 & 0 \\[-5pt]
            0 & 0 & -1 
        \end{matrix}\Bigg) \normalsize \qquad \text{on} \quad W^{(2)}
    \end{cases} \ .
\end{align}
It follows that $M$ is an orientable $\sfT^2$-bundle over $\sfS^1$.

The foliations $\cF_1$ and $ \cF_2$ of $M$ are inherited from the fibration structures of the bundles $\sfK_2$ and $\sfK_1$ respectively, while the quotient maps $\varpi_1$ and $ \varpi_2$ are given by the projections from $M$ to $\sfK_1$ and $ \sfK_2$ respectively. 

\medskip

\subsubsection{Reduction}~\\[5pt]
 Consider the untwisted standard Courant algebroid $(\IT M,0)$ over $M$. Define $K_1 = T \cF_1 \subset\IT M$. Let  $f^{(i)}(x)$ be smooth bump functions with support in $ U^{(i)}$ such that the vector fields $Z^{(1)} = f^{(1)} \, \partial_z$ and $Z^{(2)} = f^{(2)} \, \partial_z$, which can be extended globally as spanning sections of $K_1$, are not simultaneously zero. Similarly, define \smash{$Y^{(i)} = f^{(i)} \, \partial_y$}, and also \smash{$\lambda^{(i)} = f^{(i)} \, \de y$}. Note that $\partial_x$ and $\de x$ are already globally defined. 
 
As global sections of $K_1^\perp$ we then take 
\begin{align}
    \sfgammabas(K_1^\perp) = {\rm Span}_\IR\set{\partial_x, Y^{(1)}, Y^{(2)}, Z^{(1)}, Z^{(2)}, \de x, \lambda^{(1)},\lambda^{(2)}} \ .
\end{align}
One easily sees that these are basic, for instance
\begin{align}
    [Z^{(i)}, \partial_x] &= -(\partial_x f^{(i)})\, \partial_z \  \in \ \sfgamma(K_1) \quad , \quad
    [Z^{(i)},Y^{(j)}] = 0 \ ,\\[4pt]
    \pounds_{Z^{(i)}}\, \de x = \de\, \iota_{Z^{(i)}}\, \de x = 0 \quad &, \quad 
    \pounds_{Z^{(i)}}\, \lambda^{(j)} = \iota_{Z^{(i)}}\, \de (f^{(j)} \, \de x) = Z^{(i)}(f^{(j)})\, \de x - \de f^{(j)}\,( \iota_{Z^{(i)}}\,\de x) = 0 \ .
\end{align}
Thus the basic sections span  $K_1^\perp$ pointwise.

Alternatively, since $H=0$, consider the standard splitting  $\sigma_1 \colon TM \to \IT M$  given by the inclusion $X\mapsto X+0$. Then $\sigma_1(X)$ is $K_1$-basic whenever $X$ is $\cF_1$-projectable, hence $\sigma_1$ gives a $K_1$-adapted splitting. In either case, the assumptions of Theorem \ref{thm:reduction} hold, and we may form the reduced Courant algebroid 
\begin{align}
    \frac{K_1^\perp}{K_1} \Big/ \cF_1 =  (\IT \sfK_1,0) \ .
\end{align}
 In the same way, we can take $K_2 = T \cF_2 \subset\IT M$, and form the reduced Courant algebroid $(\IT \sfK_2,0)$ over $\sfK_2$.

\medskip

\subsubsection{Geometric T-Duality}~\\[5pt]
Consider  the two-form $B$ defined on each chart $U^{(i)}$ as
\begin{align}
    B = \de y \wedge \de z \ .
\end{align}
It satisfies $\big(\Fi^{(2)} \circ (\Fi^{(1)})^{-1}\big)^*B = B$ on $U^{(1)} \, \cap \, U^{(2)}$, and hence gives rise to a globally defined two-form $B\in\sfOmega^2(M)$.

The two-form $B$ is non-degenerate on $T\cF_1 \otimes T\cF_2$, and so it follows that $K_1\cap \e^{-B}\,(K_2) = \set{0}$ and $\e^B\,(K_1) \cap K_2^\perp \subset K_2$. By {\cite[Theorem 5.8]{DeFraja:2023fhe}} we may form the T-duality relation $R$ supported on $\red C=M$. On the coordinate chart $U^{(i)}:=\cV^{({i})}\times\sfT^2 \subset \red C\,$, it is given by
\begin{align}
    \sfgamma\big(\cV^{(i)},R\big) = \text{Span}_{C^{{}^\infty}(\cV^{(i)})}\set{(\partial_x, \partial_x)\,,\, (\de x, \de x)\,,\, (\partial_y, \de z)\,,\, (\de y, \partial_z)} \ .
\end{align}

Consider the flat metric on $\sfK_1$ induced by the standard flat metric $$g = \de x \otimes \de x + \de y \otimes \de y$$ on $\IR^2$, which is given locally by the same formula. On $\cV^{(1)}\times\sfS^1$ and $\cV^{(2)}\times\sfS^1$, the vector field $\partial_y$ is a local Killing vector of $g$ generating the isometry distribution $D_1$. By Theorem \ref{thm:mainlocal}, the relation $$R \colon (\IT \sfK_1,0, g) \rel (\IT \sfK_2,0, g)$$ is a generalised isometry. This gives the geometric self-duality of the Klein bottle. 

\medskip

\subsubsection{Generalised Ricci Flow}~\\[5pt]
The generalised Ricci flow is the same on both manifolds $\sfK_1$ and $\sfK_2$. It is given by the usual Ricci flow of the flat metric, i.e. the trivial stationary flow.

\medskip 

\subsection{Klein Bottle Fibration}~\\[5pt]
In this subsection we extend the example of geometric T-duality for non-principal circle bundles from Subsection~\ref{sec:Kleinbottle} to include Kalb-Ramond flux, and hence topology change. 

Let $$\cQ_1 = \sfK_1 \times \sfS^1 \ , $$ where $\sfK_1=\sfK$ is the Klein bottle as in Subection \ref{sec:Kleinbottle}, and $\sfS^1 = [0,1]/ (0\sim 1)$. We regard $\cQ_1$ as a circle bundle over a torus $\sfT^2$, with open charts
\begin{align}
    U^{(1)} = \set{(x,y,z) \in \cQ_1 \ | \ x \in (0,1)} \quad , & \quad U^{(2)} = \set{(x,y,z) \in \cQ_1 \ | \ x \neq \tfrac 12} \ , \\[4pt]
    \Fi^{(1)}\colon U^{(1)} \longrightarrow (0,1) \times [0,1]^{\times 2} \quad , & \quad \Fi^{(2)} \colon U^{(2)} \longrightarrow (0,1) \times [0,1]^{\times 2} \ ,
\end{align}
where
\begin{align}
    \Fi^{(1)}(x,y,z) = (x,y,z) \quad , \quad 
    \Fi^{(2)}(x,y,z) = \begin{cases}
        (x+\frac 12, y , z) \qquad \text{on} \quad W^{(1)}\\[4pt]
        (x-\frac 12, 1-y , z) \qquad \text{on} \quad W^{(2)}
    \end{cases} \ ,
\end{align}
with $W^{(1)} \sqcup W^{(2)} = U^{(1)} \cap U^{(2)}$, similarly to Subsection~\ref{sec:Kleinbottle}.

Let
\begin{align}
    H_1 = \tfrac{\pi}{2}\sin(\pi\, x)\, \de x \wedge  \de y \wedge \de z
\end{align}
be a local representative of the non-trivial cohomology class in $\sfH^3(\cQ_1, \IZ) = \IZ_2$. 

\medskip

\subsubsection{Topological T-Duality}~\\[5pt]
Let us work out the T-dual manifold $\cQ_2$. We construct the correspondence space $M$ as a circle bundle over $\cQ_1$, with projection map $\varpi_1:M\to\cQ_1$, using $U^{(i)}$ as trivialising charts. The additional circle coordinate is denoted by $\tilde y$. We seek a $B$-field containing the component $\de y \, \wedge \, \de \tilde y$, and satisfying $\de B = \varpi_1^*H_1$. Thus locally we take $$B = \de y \wedge \big(\de \tilde y - \tfrac{1}{2}\cos(\pi\, x)\, \de z\big) \ . $$

This form is indicative of a topology change on $M$, which arises from the splitting of the tangent bundle $TM = {\sf Ver}(M) \oplus {\sf Hor}(M)$ into vertical and horizontal parts: the vertical subbundle ${\sf Ver}(M)$ is given locally by the span of $\partial_{\tilde y}$, while the horizontal subbundle ${\sf Hor}(M)$ is   given by the span of $\partial_x$, $\partial_y$ and $\partial_{z} + \frac{1}{2}\cos(\pi\, x)\, \partial_{\tilde y}$. It suggests charts for $M$ given by
\begin{align}
    \widetilde\Fi{}^{(1)}(x,y, \tilde y,z) = (x,y, \tilde y,z) \quad , \quad 
    \widetilde\Fi{}^{(2)}(x,y, \tilde y,z) = \begin{cases}
        (x+\frac 12, y , \tilde y, z) \qquad \text{on} \quad W^{(1)} \\[4pt]
        (x-\frac 12, 1-y , 1-\tilde y, z+\tilde y) \qquad \text{on} \quad W^{(2)}
    \end{cases} \ ,
\end{align}
which is well-defined provided we identify $$(0, y, \tilde y, z) \sim (1, 1-y, 1- \tilde y, z+ \tilde y)$$ on $M$. This also ensures that the two-form $B$ is well-defined.

Upon quotienting $M$ by the foliation $\cF_2$ generated by $\partial_y$, we obtain the leaf space $$\cQ_2 = [0,1]^{\times 3} \big/ \sim \ , $$ where
\begin{gather}
    (x,0,z)\sim (x,1,z) \quad , \quad  (x,\tilde y,0)\sim (x,\tilde y,1) \quad , \quad
    (0,\tilde y,z) \sim (1,1-\tilde y,z + \tilde y) \ .
\end{gather}
This too is a circle bundle over $\sfT^2$.

We thus get the T-duality relation 
\begin{align}
    R:(\IT\cQ_1,H_1) \rel (\IT\cQ_2,0) \ .
\end{align}
On $U^{(i)} = \cV^{(i)}\times\sfS^1$, it is locally generated  by
\begin{align}
    \sfgamma\big(\cV^{(i)},R\big) = \text{Span}_{C^\infty(\cV^{(i)})} \{& (\partial_x, \partial_x)\,,\, (\partial_y, \de \tilde y - \tfrac{1}{2} \cos(\pi\, x)\, \de z)\,,\, (\partial_z,\partial_z + \tfrac{1}{2}\cos(\pi\, x)\, \partial_{\tilde y})\,,\, \\
    & \, (\de x, \de x)\,,\, (\de y, \partial_{\tilde y})\,,\,(\de z, \de z)\} \ .
\end{align}

The general framework for topological T-duality between affine torus bundles over a common base $\cB$ was considered in \cite{Baraglia:2014nsv} using the language of local coefficient systems. In this setting T-duality is stated after a choice of local coefficient system $\IZ_\xi$, i.e. an element $\xi$ of $\sfH^1(\cB, \IZ_2)$. In our present example, a representative of $\xi$ is equivalent to a choice of charts. Since the torus $\sfT^2$ has twisted cohomology group $\sfH^2(\sfT^2, \IZ_\xi) = \IZ_2$, there are two circle bundles over $\sfT^2$ for this choice of charts. Turning on the Kalb-Ramond flux $H_1$ on $\cQ_1$ (of which there is only one cohomologically non-trivial choice since $\sfH^3(\cQ_1, \IZ) = \IZ_2$), it follows that T-duality passes from the trivial to the non-trivial element of $\sfH^2(\sfT^2, \IZ_\xi)$. This can also be formulated in terms of twisted connections, as considered in \cite{Baraglia:2014nsv} where T-duality for twisted cohomology was treated, whereby the one-form $\de \tilde y - \frac 12 \sin(\pi\, x)\,\de z$ defines a twisted connection on the double cover \smash{$\widetilde \cQ_2$} of $\cQ_2$ with non-trivial curvature. 

Compared to~\cite{Baraglia:2014nsv}, the advantage of the formulation considered here is that it does not consider the affine structure of the torus bundle, nor even the group structure on the fibres, only the fibration structure itself. It also permits consideration of the T-duality of the geometric structures, as we illustrate in Subsection~\ref{sec:KleinfibgeomT} below, which was not considered in the approach of~\cite{Baraglia:2014nsv}.

\medskip

\subsubsection{Geometric T-Duality}\label{sec:KleinfibgeomT}~\\[5pt]
Let $g_1$ be the flat metric on $\cQ_1$ which is locally given by $$g_1 = \de x \otimes \de x +\de y \otimes \de y +\de z \otimes \de z \ . $$ Then the T-duality relation $R$ is a generalized isometry between $g_1$ and the T-dual metric $g_2$ on $\cQ_2$ given by 
\begin{align}
    g_2 = \de x \otimes \de x +  \de\tilde y \otimes \de\tilde y - \cos(\pi\, x)\, \de\tilde y \odot \de z + \left(1+\tfrac 14 \cos^2(\pi\, x) \right ) \de z \otimes \de z \ .
\end{align}

On $\cQ_1$ we further find $$H_1^2 = \tfrac{\pi^2}{2}\sin^2(\pi\, x)\, g_1 \quad , \quad \de^* H_1 = \tfrac{\pi^2}{2} \cos(\pi\, x)\, \de y \wedge \de z \ . $$ Using the Ricci-Buscher rules of Subsection~\ref{sec:RicciBuscher}, from this we can calculate the non-zero components of the Ricci tensor of $g_2$ to be
\begin{align}
    (\ric_{g_2})_{ y y} &=  
    \tfrac{\pi^2}{8}\sin^2(\pi\, x) \ ,\\[4pt]
    (\ric_{g_2})_{z y} &= 
    \tfrac{\pi^2}4 \cos(\pi\, x) +\tfrac{\pi^2}{16}\sin^2(\pi\, x) \cos( \pi\, x) \ ,\\[4pt]
    (\ric_{g_2})_{xx} &= -\tfrac{\pi^2}{8} \sin^2(\pi\, x) \ ,\\[4pt]
    (\ric_{g_2})_{zz} 
    &= -\tfrac{\pi^2}{8} \sin^2(\pi\, x) + \tfrac{\pi^2}{4}\cos^2(\pi\, x) + \tfrac{\pi^2}{32}\, \sin^2(\pi\, x) \cos^2(\pi\, x) \ .
\end{align}

The generalised Ricci flow for this example will be a highly complicated non-linear system due to the trigonometric functions appearing.
We do not attempt to solve it here, as it requires some involved numerical computation which is beyond the scope of this paper. 

\begin{appendix}
\section{The \texorpdfstring{$\sigma$}{Sigma}-Adjoint Action} \label{app:twistedaction}

In this appendix we prove some properties of the $\sigma$-adjoint action that we use in the main text. Recall that it is given by Equation \eqref{eqn:sigmatwistact} as
\begin{align}
    {\sf ad}^\sigma_{e_1}\, e_2 \coloneqq \llbracket e_1 , e_2 \rrbracket - \rho^*\,{\sigma}^*\llbracket e_1 , {\sigma}\,\rho(e_2) \rrbracket \ ,
\end{align}
for $e_1,e_2\in\sfgamma(E)$ and an isotropic splitting $\sigma:TM\to E$ of an exact Courant algebroid $E$ over $M$. This operation was introduced (but not thusly named) in~\cite[Definition 5.21]{DeFraja:2023fhe}.

\begin{lemma}
    $\rho(\mathsf{ad}^\sigma_{e_1}\, e_2) = \pounds_{\rho(e_1)}\, \rho(e_2)$ \ for all $e_1,e_2\in\sfgamma(E)$.
\end{lemma}

\begin{proof}
    This follows immediately from the fact that the anchor map $\rho:E\to TM$ is a bracket homomorphism and $\rho\circ\rho^*=0$.
\end{proof}

\begin{lemma} \label{lem:sigmaadleibniz}
    The $\sigma$-adjoint action satisfies the anchored Leibniz rule 
    \begin{align}
        \ad^\sigma_{e_1} (f\, e_2) = f \, \ad^\sigma_{e_1}\,e_2 + \left( \pounds_{\rho(e_1)} f  \right) e_2 \ ,
    \end{align}
    for all $e_1 , e_2 \in \sfgamma(E)$ and $f \in C^\infty(M).$
\end{lemma}

\begin{proof}
 For any $e_1 , e_2 \in \sfgamma(E)$ and $f \in C^\infty(M)$ we compute
\begin{align}
\ad^\sigma_{e_1} (f\, e_2) &= \cbrak{e_1, f\, e_2} - \rho^* \, \sigma^*\cbrak{e_1, f\,\sigma\, \rho(e_2)} \\[4pt]
&= f\, \cbrak{e_1, e_2} + \left( \pounds_{\rho(e_1)} f \right) e_2 - \rho^*\, \sigma^* \big( f\, \cbrak{e_1, \sigma\,\rho(e_2)} + (\pounds_{\rho(e_1)} f)\, \sigma\,\rho(e_2) \big)\\[4pt]
&= f \, \ad^\sigma_{e_1}\,e_2 + \left( \pounds_{\rho(e_1)} f\right) \big( e_2 - \rho^*\, \sigma^*\, \sigma\, \rho(e_2) \big) \ ,
\end{align}
and the result now follows from $\sigma^* \circ \sigma = 0$. 
\end{proof}

\begin{lemma}\label{lem:pairingpreserved}
    The pairing $\ip{\,\cdot\,,\, \cdot\,}$ is preserved by $\ad_e^\sigma$, for all $e \in \sfgamma(E)$.
\end{lemma}

\begin{proof}
    Let $e, e_1, e_2 \in \sfgamma(E)$. Then
    \begin{align}
        \pounds_{\rho(e)} \ip{ e_1, e_2} &= \ip{\llbracket e, e_1 \rrbracket, e_2} + \ip{e_1, \llbracket e, e_2 \rrbracket} \\[4pt]
        &=\ip{\ad_e^\sigma\,e_1, e_2} + \ip{\rho^*\, \sigma^*\llbracket e, \sigma \,\rho(e_1) \rrbracket, e_2} + \ip{e_1, \ad_e^\sigma\,e_2} + \ip{e_1, \rho^*\, \sigma^*\llbracket e, \sigma \,\rho(e_2) \rrbracket} \\[4pt]
        &=\ip{\ad_e^\sigma\,e_1, e_2} + \ip{\llbracket e, \sigma \,\rho(e_1) \rrbracket, \sigma \,\rho(e_2)}  + \ip{e_1, \ad_e^\sigma\,e_2} + \ip{\sigma\,\rho(e_1), \llbracket e, \sigma \,\rho(e_2) \rrbracket} \\[4pt]
        &=\ip{\ad_e^\sigma\,e_1, e_2} + \ip{\llbracket e, \sigma \,\rho(e_1) \rrbracket, \sigma \,\rho(e_2)}+ \ip{e_1, \ad_e^\sigma\,e_2} \\ & \quad \,  + \pounds_{\rho(e)}\,\ip{\sigma\,\rho(e_1), \sigma \,\rho(e_2)} - \ip{\llbracket e, \sigma\,\rho(e_1) \rrbracket, \sigma\,\rho(e_2)} \\[4pt]
        &=\ip{\ad_e^\sigma\,e_1, e_2} + \ip{e_1, \ad_e^\sigma\,e_2} \ ,
    \end{align}
    as required.
\end{proof}

\begin{lemma} \label{lem:bracketpreserving}
    Under the Courant algebroid isomorphism $E \cong (\IT M,H)$ induced by $\sigma$, the Dorfman bracket $\llbracket\,\cdot\,,\,\cdot\,\rrbracket_H$ is preserved by $\ad^\sigma_{X+\alpha}$ if and only if $\iota_XH$ is a closed two-form, for all $X+\alpha\in\sfgamma(\IT M)$.
\end{lemma}

\begin{proof}
     Under the isomorphism $E \cong (\IT M,H)$ we can identify $\sigma(X) = X$ for all $X \in \sfgamma(TM)$. Then
    \begin{align}
       {\sf ad}^\sigma_{X + \alpha} \llbracket Y+ \eta , Z+ \xi \rrbracket_H\, - &\,\llbracket {\sf ad}^\sigma_{X+ \alpha} ( Y+ \eta ), Z+ \xi \rrbracket_H - \llbracket Y+\eta,  {\sf ad}^\sigma_{X+ \alpha} (Z+ \xi) \rrbracket_H \\[4pt]
       &= -\iota_{[Y,Z]}\, \iota_X H + \llbracket \iota_Y\, \iota_X H, Z+\xi \rrbracket_H + \llbracket Y+\eta, \iota_Z\, \iota_X H \rrbracket_H \\[4pt]
       &= -\iota_{[Y,Z]}\, \iota_X H - \iota_Z\, \de\, \iota_Y\, \iota_X H  + \pounds _Y\, \iota_Z\, \iota_X H\\[4pt]
       &= \iota_Z\, \iota_Y\, \de\, \iota_X H \ ,
    \end{align}
    where in the final step we make use of various Cartan calculus identities. Thus the Dorfman bracket is invariant under this action if and only if $\iota_X H$ is closed.
\end{proof}

Let $Z \in \sfgamma(TE)$ be a vector field on $E$ covering a vector field $\overline{Z} \in \sfgamma(TM)$ on $M$. The \emph{generalised Lie derivative} along $Z $ at $e\in E$ is defined by
\begin{align}
    (\boldsymbol{\pounds}_Z\, v)_e \coloneqq  \frac{\de}{\de s} (\boldsymbol\Phi^*_s v)_e\,\Big|_{s=0} \quad , \quad \boldsymbol\pounds_Z\, f := \pounds_{\overline{Z}}\, f
\end{align}
for all $v \in \sfgamma(E)$ and $ f \in C^\infty(M)$, where $\boldsymbol\Phi_s\in\mathsf{Aut}(E)$ is the flow of $Z$ through $e$. This definition easily extends to tensor powers of $E$ via the Leibniz rule.

\begin{lemma}\label{lem:adpreserve}
    On the space of vector fields $X \in \sfgamma(TM)$ for which $\ad^\sigma_X$ preserves the Dorfman bracket, there are infinitesimal Courant algebroid automorphisms $Z_X\in \sfgamma(TE)$ such that
    \begin{itemize}
        \item $\overline{Z_X} = X$ \ ,\\[-3mm]
        \item $\boldsymbol\pounds_{Z_X} v = \ad^\sigma_X \, v \ , $ \ for all $v\in\sfgamma(E)$,\\[-3mm]
        \item $[Z_X, Z_Y] = Z_{[X,Y]} \ $ and the map $X\mapsto Z_X$ is $\IR$-linear.
    \end{itemize}
\end{lemma}

\begin{proof}
    For a section $e$ of $E$ and an integral curve $\gamma(t)\subset M$ of $X$, the vector field $Z_X$ is given by the tangent to the curve $(e-t\,\ad_X^\sigma\, e)_{\gamma(t)} \subset E$ at $t=0$; a standard argument (cf.~\cite{Severa2015}) shows that $Z_X$ does not depend on the choice of auxiliary section or curve. Then the first two properties follow immediately, while the final property follows from $[\ad^\sigma_X, \ad^\sigma_Y] = \ad^\sigma_{[X,Y]}$, which can be proven using various Cartan identities and the fact that $\iota_{[X,Y]}\,H$ is closed (and hence $\llbracket \llbracket X, Y \rrbracket_H ,\,\cdot \,\rrbracket_H = \llbracket \sigma([X,Y]), \, \cdot \, \rrbracket_H$.)
\end{proof}

\begin{remark}
    If $\sigma_0(X) = X $ and $\ad^{\sigma_0}_{X}$ preserves the Dorfman bracket as above, then $\pounds_X H = \de\, \iota_X H = 0$, and hence $X$ integrates to a Courant algebroid automorphism induced by a diffeomorphism $\Fi \in \mathsf{Diff}(M)$ which preserves $H$, i.e. $\Fi^*H = H$.
\end{remark}

\begin{remark}
    If the infinitesimal Courant algebroid automorphism corresponding to $e \in \sfgamma(E)$ (in the sense of \cite[Proposition~2.7]{Streets:2024rfo}) preserves a generalised metric $V^+\subset E$, i.e. $\llbracket e, v^+ \rrbracket \in \sfgamma(V^+)$ for all $v^+\in\sfgamma(V^+)$, then $V^+$ is also invariant under the $\sigma$-adjoint action ${\sf ad}^\sigma_e$, since $\rho^*\, \sigma^* \llbracket e, \sigma\, \rho(v^+) \rrbracket = 0$ and hence $\ad_e^\sigma$ acts through the usual Dorfman bracket on $\sfgamma(V^+)$. For instance, if $B\in\sfOmega^2(M)$ and $\sigma_B(X) = X+\iota_X B$ is an infinitesimal Courant algebroid isometry for $(g, b)$, then $\pounds_X (g+b) = 0$ and $\iota_X\,\de B = \de\, \iota_X b$, so $\pounds_X H = 0$. In this sense, it is a generalisation of inner automorphisms allowing for $\iota_X H$ to be closed.
\end{remark}

\begin{lemma}\label{cor:invariantsinvolutive}
    In the setting of Definition~\ref{def:invariantsection}, the space of $D$-invariant sections $\sfgamma_{D}( E)$ is involutive.
\end{lemma}

\begin{proof}
This is a consequence of Lemma~\ref{lem:bracketpreserving}.
\end{proof}

\begin{lemma}\label{lem:projectioninvariant}
    If $V^+\subset E$ is a $D$-invariant generalised metric and $v$ is a $D$-invariant section of $E$, then the projections $v^\pm\in\sfgamma(V^\pm)$ are also $D$-invariant.
\end{lemma}

\begin{proof}
    Since $\ad^{ \sigma}_{X}\, v=0$ and $\ad^{ \sigma}_{X}\, v^+ \in \sfgamma(V^+)$ for all $X\in\mathfrak{k}_\tau$, it follows that $\ad^{ \sigma}_{X}\, v^- \in \sfgamma(V^+)$. Thus for any $w^+ \in \sfgamma(V^+)$, we have
    \begin{align}
        \ip{\ad^{ \sigma}_{X}\, v^-, w^+} = \pounds_X \ip{v^-,w^+} - \ip{ v^-, \ad^{ \sigma}_{X}\, w^+} = 0 \ .
    \end{align}
    Hence $\ad^\sigma_X\,v^-$ must vanish, and so then must $\ad^\sigma_X\,v^+$.
\end{proof}

\end{appendix}

\bibliographystyle{ourstyle}
\bibliography{bibprova3}

\end{document}